\documentclass[a4paper,11 pt]{article}
\usepackage{amsmath}
\usepackage{amssymb}
\usepackage{amsthm}
\usepackage{float}
\floatstyle{boxed}
\restylefloat{figure}
\usepackage[font=small,labelfont=bf]{caption}
\usepackage{graphicx}
\usepackage[colorlinks=true, linkcolor=blue]{hyperref}
\usepackage{cite}
\newcommand{\ket}[1]{\left| #1\right\rangle }
\newcommand{\inp}[2]{\left\langle #1 \mid #2\right\rangle}
\newtheorem{thm}{Theorem}[section]
\newtheorem{defin}[thm]{Definition}
\newtheorem{lem}[thm]{Lemma}
\newtheorem{cor}[thm]{Corollary}
\newtheorem{prop}[thm]{Proposition}
\newtheorem{rem}[thm]{Remark}
\newcommand{\bra}[1]{\left\langle #1\right|}
\newcommand{\pty}[2]{\textbf{PAR}_{#1,#2}}
\newcommand{\rank}{\text{rank}}

\author{Maximilian Fillinger}
\title{Data Structures in Classical and Quantum Computing}

\begin{document}
\maketitle
\begin{abstract}
This survey summarizes several results about quantum computing
related to (mostly static) data structures. First, we describe
classical data structures for the set membership and the predecessor
search problems: Perfect Hash tables for set membership from the paper
\cite{fks} by Fredman, Koml\' os and Szemer\' edi and a
data structure by Beame and Fich for predecessor search presented
in \cite{bf}. We also prove results about their space complexity
(how many bits are required) and time complexity (how many bits
have to be read to answer a query).

After that, we turn our attention to classical data structures with
quantum access. In the quantum access model, data is stored in classical
bits, but they can be accessed in a quantum way: We may read several
bits in superposition for unit cost.
We give proofs for lower bounds in this setting that show that
the classical data structures from the first section are, in some sense, asymptotically optimal
- even in the quantum model. In fact, these proofs are simpler and
give stronger results than previous proofs for the classical model
of computation. The lower bound for set membership was proved by
Radhakrishnan, Sen and Venkatesh in \cite{qu_set_mem} and the result
for the predecessor problem by Sen and Venkatesh in \cite{qu_cell_probe}.

Finally, we examine fully quantum data structures. Instead of encoding
the data in classical bits, we now encode it in qubits. We
allow any unitary operation or measurement in order to answer queries.
We describe one data structure by de Wolf in \cite{dewolf_thesis} for the set
membership problem and also a general framework using
fully quantum data structures in quantum walks by Jeffery, Kothari
and Magniez in \cite{quant_walk}.
\end{abstract}

\section{Introduction}

\subsection{Data Structures}
\label{int:data}

\em Data structures \em are a fundamental area of study in computer science
since efficient storage and retrieval of data is an important task.
In a data structure problem, we want to encode objects from some \em universe
\em $\mathcal U$ into
bit strings so that certain \em queries \em about the stored object can be
answered efficiently. The study of data structures is to find and analyse
trade-offs
between the length of the bit string and the time it takes to answer queries.
The time is measured in terms of the number of bits or blocks of memory that
we must read to answer a query.
Data structure problems can be \em static \em or \em dynamic\em .
For static problems,
we are content with having queries answered. For dynamic problems, we also want
the data structure to efficiently support some operations that change the stored object.
This survey is mostly about static problems.

Examples of data structure problems are the \em set membership problem \em
and its close relative, the \em dictionary problem\em . In the set membership
problem, we want to store a set of integers $S$ so that we can efficiently
find out whether some number $x$ is contained in $S$ or not. The set $S$
has size $\leq n$ and the numbers it contains are all $< m$ for some integers
$n \ll m$.
The dictionary
problem is the same except that each number in the set $S$ is associated with
some additional data.
Given an integer $k$, we want to be able to find out whether $k$ is in the set
and if yes, what data it is associated with.
For example, we might think of a phone book: The integers would be
representations of names and the associated data would be phone numbers.
How would we store such a phone book on a computer in a way that is both
efficient in terms of memory and allows us to quickly retrieve the phone number
of any given person?

A simple solution to the set membership problem is the \em bit vector \em
method. We encode our set as an $m$-bit string where the $i$th bit (counting
from 0 to $m-1$) is set to 1 if and only if $i\in S$. This allows us to answer
set membership queries by reading a single bit, but it uses a lot of space.
When $n \ll m$, only a small fraction of bits is set to 1, so it seems like
we are wasting a lot of space.

\em Hash tables \em offer a practical solution to the set membership problem.
We take a ``random looking'' hash function $h : \{ 0,\dots , m-1\}\to
\{ 0,\dots , n-1\}$.
We would like to store a set
$S$ of size $n$ in $n$ blocks of memory by storing each
$i\in S$ in block $h(i)$. However, since $m$ is larger than $n$,
our set may contain numbers $i,j$ such that $h(i) = h(j)$. We call such pairs
\em collisions \em of $h$. If the set $S$ we want to
store contains a collision, we need to resolve it in some way.
The easiest approach is to store in slot $k$ a pointer to the head
of a linked list that contains all the $i\in S$ with $h(i) = k$. See
Figure \ref{fig_hash} for an illustration.

\begin{figure}[h]
\centering
\includegraphics{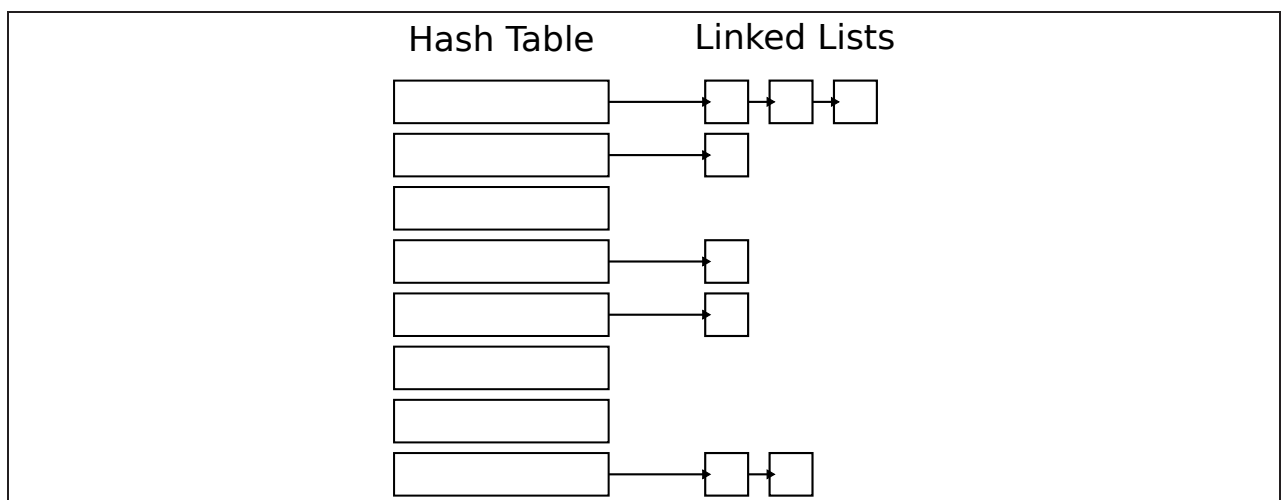}
\caption{The $i$th cell in the hash table contains a pointer to the head of a
linked list which contains the stored integers $j\in S$ with $h(j) = i$.}
\label{fig_hash}
\end{figure}

This translates easily into a solution for the dictionary problem: Instead
of integers, we store pairs $(k, p)$ where $k$ is an integer and $p$ is a
pointer to the data associated with $k$.

When the set $S$ that we store is selected at random, there is a high
probability that
we can quickly find out whether any given integer $k$ is in $S$ or not. However,
in the worst case we must read a lot of memory blocks: If $n\leq \sqrt m$ then
there exists a set $S$ with at most $n$ elements such that all elements of $S$
have the same hash value.
In that case, our encoding of $S$ is clearly not better than a simple linked
list. Thus, we must read $n$ blocks of memory in the worst case.

In Section \ref{mem_class}, we describe the Perfect Hashing scheme by
Fredman, Koml\' os and Szemer\' edi which works efficiently for every set $S$.
The downside of this method is that, while ordinary hash tables also allow to add and
remove elements of the set easily, there is no straightforward way to do so in
the Perfect Hashing scheme (other than encoding the changed set ``from scratch'').
Nevertheless, if we are only concerned with the static set membership problem,
the Perfect Hashing scheme is asymptotically optimal when we require that our
queries are answered correctly with probability 1, as we will see in Sections
\ref{mem_class} and \ref{quantum_set}.

Another problem that we study in this survey is the \em predecessor problem\em .
Again, we encode sets $S\subseteq \{ 0,\dots , m-1 \}$ of integers, but this
time, the data structure should allow us to quickly find the
\em predecessor \em of any integer $i < m$ in $S$. That is, we want to find
out if $S$ contains an element smaller than $i$, and if yes, we want to find
the largest $j\in S$ such that $j < i$. Examples of data structures for
the predecessor problem are \em Fusion Trees \em by Fredman and
Willard in \cite{fusion} and \em X-fast Tries \em by Willard in
\cite{xfast}. Fusion Trees need
$O(n)$ blocks of $\log m$ bits of memory and support predecessor queries that
read $O(\log n/\log \log m)$ blocks. X-fast Tries
use $O(n\log m)$ blocks of size $\log m$ and support queries reading
$O(\log\log m)$ blocks. Combining these two data structures with a contribution
of their own, Beame and Fich describe a data structure in \cite{bf} that uses
$O(n^2 \log n/\log\log n)$ blocks of size $\log m$ and answers predecessor
queries reading
$$ O \left( \min\left( \frac{\log\log m}{\log\log\log m},
\sqrt{\frac{\log n}{\log \log n}}\right)\right)$$
blocks.

\subsection{Quantum Computing -- Informal}

A \em quantum computer \em is a mostly hypothetical computing device that
operates on the basis of quantum mechanics, unlike classical
computers.\footnote{It is not quite true that actual classical computers are based on
classical physics. Quantum effects have to be taken into account in the
construction of classical computer hardware since the transistors used for
this are now so small that quantum effects cannot be ignored anymore.
However, they are designed to behave classically.}
Actual quantum
computers have only been built on a very small scale of a few qubits
(quantum bits).
While there does not seem to be any physical law preventing the construction of
large-scale quantum computers, it is a hard engineering problem.
This does not prevent theorists from inventing algorithms for quantum computers.
The most famous example is \em Shor's factoring algorithm \em presented in
\cite{shor} which efficiently computes the prime factors of a given number and
could be used to break the widely-used RSA encryption scheme. Another important
algorithm is \em Grover's search algorithm \em described in \cite{grover} which
can search a list of length $N$ in the order of $\sqrt N$ computational steps.

It is a natural question to ask what quantum computers could do for data
structures. We examine two different models for the interaction of quantum
computers and data structures. The first one is the \em quantum access \em model
where the data is still encoded into a classical bit string but is accessed
in a quantum way.

However, it turns out that for the data structure problems that we study in this
survey, the quantum access model of computation has no advantage over the
classical model, at least in asymptotic terms. Nevertheless, the theory of
quantum computing is a useful mathematical tool: The proofs of lower
bounds for these problems in the quantum access model are easier to
understand and yield stronger results than earlier proofs given for the
classical model. Since a quantum computer can do anything that a classical
computer can do, this also gives us classical lower bounds. There are many more
areas of computer science where results from quantum computing are
relevant for classical computing. Unlike quantum algorithms, these results
are useful even if large-scale quantum computers are never built and even if
quantum mechanics turns out not to be an accurate description of reality.
For a survey of such results, see \cite{qu_method}.

The second model can be called \em fully quantum \em data structures. We
encode our data not in bits, but in qubits. To answer queries, we may use
all the operations available in quantum computing.

\section{Preliminaries}

\subsection{Notation}

For every positive integer $n$, we let $[n]$ denote the set
$\{ 0,\dots ,n-1\}$. The base-2 logarithm is written as $\log$ and
the natural logarithm as $\ln$. When an integer $n$ is not a power
of 2, we implicitly round up $\log n$ so that $\log n$ is
the minimum number of bits required to denote any number from $[n]$.
The symbol $a\oplus b$ stands for the bit-wise XOR of two bit strings
$a$ and $b$ and
$a\circ b$ is the concatenation of two strings. We identify
non-negative integers with their binary representation and in this sense, we
may talk about a bit string being larger or smaller than another one,
or about a bit string being a prefix of some number.

\subsection{Quantum Computing -- Formal}

Let us now make the idea of quantum computing a little more precise. A
$k$-qubit quantum memory register is modeled as a unit-length vector in the
complex vector space $\mathbb C^{2^k}$ with the standard scalar product and
norm.\footnote{Such vector spaces are examples of \em Hilbert spaces\em . A
Hilbert space is a complex vector space with a scalar product. In general,
Hilbert spaces may be of infinite dimension, but we only consider finite spaces
in this survey.} We also call this vector space the \em state space \em of
our system.
The vectors of the standard basis
of this space are denoted as $\ket b$ for $b\in \{ 0,1 \} ^k$ or, using the
convention of identifying non-negative integers with their binary representation,
$b\in [2^k]$. Thus, every
possible state of the quantum memory can be written as
$$\sum _{b\in \{ 0,1 \} ^k} \alpha _b\ket b$$
where $\sum _b |\alpha _b|^2 = 1$. The complex number $\alpha _b$ is called
the \em amplitude \em of basis state $\ket b$. When we read the quantum memory
or \em measure \em it, we will get result $b$ with probability $|\alpha _b|^2$
and the state will \em collapse \em to $\ket b$. We can operate on a quantum
register by applying \em unitary transforms \em on it. A unitary transform
is a length-preserving isomorphism, i.e., it is linear, bijective and the length
of a vector does not change when the transform is applied to it.

An easy, yet important example of a unitary transform is the \em Hadamard
transform \em $H$. It operates on a single qubit and maps
$\ket 0\mapsto (\ket 0 + \ket 1)/\sqrt{2}$ and
$\ket 1\mapsto (\ket 0 - \ket 1)/\sqrt 2$. This already determines its
behaviour on the whole vector space $\mathbb C^2$ since it must be linear.
Written as a matrix, the Hadamard transform looks as follows:
$$H = \frac 1{\sqrt 2}\left(\begin{array}{cc} 1 & 1\\ 1 & -1\end{array}\right)$$
By a simple computation, one can verify that $H$ is its own inverse and hence
bijective. It is also easy to see that it is length-preserving. We write
$\ket + = H\ket 0$ and $\ket - = H\ket 1$.

Sometimes, we need to model  uncertainty about a quantum state. As an example,
imagine that we measure a qubit in state $\ket +$ in the
computational basis, but forget the result of the measurement.
Then we do not know whether our
qubit is in state $\ket 0$ or $\ket 1$, but we know that the probability of
either state is 1/2. If we now measure the qubit again, we will see outcome
0 or 1 with probability 1/2 each, just as if it still were in state $\ket +$.
But we know that our qubit is not in that state anymore. \em Density matrices
\em capture the distinction between these cases. Let $\ket{\psi _1},\dots ,
\ket{\psi _n}$ be quantum states and $p_1,\dots , p_n$ positive real numbers
that sum to 1. A system that is in state $\ket{\psi _i}$ with probability
$p_i$ is modeled as a matrix $\rho = \sum _i p_i\ket{\psi _i}\bra{\psi _i}$.
If $n = 1$, we say that our system is in a \em pure state \em and it can be
described by a state vector. If $n > 1$, we speak of a \em mixed state\em .
Every density matrix $\rho$ has trace $\text{Tr} (\rho ) = \sum _i \rho _{ii}
= 1$ and is \em positive semi-definite\em ,
i.e., $\bra{\phi}\rho\ket{\phi} \geq 0$ for every vector $\ket{\phi}$ in the
state space. Conversely, one can show that every matrix with these properties
can be expressed as a sum $\sum _i p_i\ket{\psi _i}\bra{\psi _i}$ for appropriate
probabilities and quantum states. When we measure a qubit that is described
by $\rho$ in the computational basis, with probability $\rho _{ii}$ we will
see outcome $i$ and the system will collapse to pure state $\ket i\bra i$.
When we apply a unitary $U$ to a system described by $\rho$, the result is
$U\rho U^*$.

Let us now apply the density matrix formalism to our example. The density
matrix of $\ket +$ is
$$\ket +\bra +
= \frac 1 2 \left( \begin{array}{cc} 1 & 1\\ 1 & 1\end{array}
\right)$$
whereas the state after the measurement, when we forget the result, is
$$\frac 1 2\ket 0\bra 0 + \frac 1 2\ket 1\bra 1
= \frac 1 2 \left( \begin{array}{cc} 1 & 0\\0 & 0\end{array}\right)
+ \frac 1 2 \left( \begin{array}{cc} 0 & 0\\0 & 1\end{array}\right)
= \frac 1 2 \left( \begin{array}{cc} 1 & 0\\0 & 1\end{array}\right)
\text .$$

An introduction to quantum computing can be found in the book
``Quantum Computation and Quantum Information'' by Michael A. Nielsen and
Isaac L. Chuang \cite{mikeike}. We will assume some familiarity with the
basics of quantum computing, but give a short review of the
\em quantum cell-probe \em model of computation in Section \ref{comp}.

\subsection{Bit-Probe and Cell-Probe Algorithms}
\label{comp}

We now describe the computational models which we use for the algorithms
that answer data structure queries.
In the \em bit-probe \em model,
any computation is for free, but reading a bit from the input bit string $x
= x_0\dots x_{n-1}\in \{ 0,1 \} ^n$ carries unit cost. \em Which \em bits are
read may depend on the results
of previous bit-probes. In the \em cell-probe \em model, we
fix some \em cell-size \em or \em block-size \em $w$ and
view our input bit string $x$ as a sequence of cells or blocks of length $w$.
That is, the first block is $x_0\dots x_{w-1}$, the second is
$x_w\dots x_{2w-1}$ and so on. Instead of single bits, we
may read a whole block at once for unit cost and computation
is still for free. We assume that the length of $x$ is a multiple of $w$ which
can always be achieved by appending some padding.
Obviously, the bit-probe model is the cell-probe
model with cell-size $w = 1$.

We can formalize algorithms in this model as \em decision trees\em .
In a decision tree, each non-leaf node $u$ is labeled with an integer
$i_u\in [n/w]$ and each leaf is labeled with a possible output of the
algorithm. A non-leaf node $u$ in a decision tree has exactly $2^w$ children
and every edge leading to a child of $u$ is labeled with a unique number
from $[2^w]$. A decision tree is evaluated on input $x$ by starting at the root
$r$ and proceeding along the
edge labeled with $x_{i_r}$ to a child $u$ of $r$. We then proceed
along the edge labeled $x_{i_u}$ and so on, until we reach a leaf.
The label of the leaf is the output of the algorithm. The cell-probe complexity
is the depth of the tree.

\em Probabilistic \em cell-probe algorithms are modeled as probability
distributions over finite sets of decision trees and are evaluated by randomly
selecting a decision
tree according to the distribution and evaluating it. The complexity
is the maximal depth among all the trees that have positive probability.
Cell probe algorithms with \em auxiliary input \em from finite domain
$\mathcal Q$ map every $q\in\mathcal Q$ to some decision tree or probability
distribution over decision trees. The auxiliary input does not count towards
the complexity of the algorithm.

\em Quantum cell-probe \em algorithms have classical input $x = x_0\dots x_{n-1}$
and classical output, but compute on qubits instead of classical bits.
Furthermore, they are allowed to read several bits/blocks at once
in superposition for unit cost.
The state space that a quantum algorithm operates on is given as
$\mathcal H = \mathcal H_L\otimes\mathcal H_B\otimes\mathcal H_Z$.
The Hilbert space $\mathcal H_L$ consists of the address qubits and is
used for denoting which blocks of $x$ are to be read next. It consists
of $\log (n/w)$ qubits. The state space $\mathcal H_B$ describes
$w$ qubits which
are called the data qubits. They store the result when blocks of $x$ are read.
Finally, $\mathcal H_Z$ consists of an arbitrary number of qubits which
are used as workspace for the algorithm.

A quantum cell-probe algorithm of complexity $t$ is given
by a sequence $U_0,\dots ,U_t$ of unitary transforms on $\mathcal H$
which are independent of the input $x$. The input is accessed by a
unitary oracle transform $O_x$ that depends on $x$. For
computational basis states
$\ket l_L\in\mathcal H_L$, $\ket b_B\in\mathcal H_B$ and $\ket z_Z\in
\mathcal H_Z$, the action of $O_x$ is given by
$$\ket l_L\ket b_B\ket z_Z\mapsto \ket l_L\ket{b\oplus x_{lw}\dots x_{(l+1)w-1}}_B\ket z_Z\text .$$
The algorithm is evaluated as follows: The state space is initialized
to some state $\ket{\phi }$ which may encode some auxiliary input.
Then, we apply the unitary transform
$$U_tO_xU_{t-1}O_x\dots U_1O_xU_0$$
to $\ket{\phi}$. The output of the algorithm is obtained by measuring
the rightmost qubits of the work space. Quantum cell
probe algorithms can simulate both deterministic and probabilistic
classical cell-probe algorithms.

If the block-size is 1, we can alternatively define the oracle transform
by
$$O_{x,\pm }:\ket l_L\ket b_B\ket z_Z\mapsto (-1)^{b\cdot x_l}\ket l_L\ket b_B\ket z_Z$$
which is equivalent to $O_x$ in the sense that one can be used to implement the other
with the help of the Hadamard transform $H:\ket b\mapsto (1/\sqrt 2)(\ket 0 + (-1)^b\ket 1)$
which can be included in the transforms before and after the query.
We have
\begin{align*}
\ket l_L\ket 0_B\ket z_Z\stackrel{H_B}{\longmapsto}
\frac{1}{\sqrt 2}\ket l_L\left( \ket 0_B + \ket 1_B\right)\ket z_Z
&\stackrel{O_x}{\longmapsto}\frac{1}{\sqrt 2}\ket l_L\left( \ket 0_B + \ket 1_B\right)\ket z_Z
\stackrel{H_B}{\longmapsto}\ket l_L\ket 0_B\ket z_Z\\
\ket l_L\ket 1_B\ket z_Z\stackrel{H_B}{\longmapsto}
\frac{1}{\sqrt 2}\ket l_L\left( \ket 0_B - \ket 1_B\right)\ket z_Z
&\stackrel{O_x}{\longmapsto}(-1)^{x_l}\frac{1}{\sqrt 2}\ket l_L\left( \ket 0_B - \ket 1_B\right)\ket z_Z\\
&\stackrel{H_B}{\longmapsto}(-1)^{x_l}\ket l_L\ket 1_B\ket z_Z\\
\ket l_L\ket b_B\ket z_Z\stackrel{H_B}{\longmapsto}
\frac{1}{\sqrt 2}\ket l_L\left( \ket 0_B + (-1)^b\ket 1_B\right)\ket z_Z
&\stackrel{O_{x,\pm}}{\longmapsto}\frac{1}{\sqrt 2}\ket l_L\left( \ket 0_B + (-1)^{b\oplus x_l}\ket 1_B\right)\ket z_Z\\
&\stackrel{H_B}{\longmapsto}\ket l_L\ket{b\oplus x_l}_B\ket z_Z
\end{align*}
The $\pm$-type oracle transform is
represented by a diagonal matrix where the entries
on the diagonal are either $-1$ or $1$ which is helpful
in a proof presented in Section \ref{quantum_set}.

\subsection{Data Structure Problems and Solutions}

A static data structure problem is given by a finite universe $\mathcal U$, a
finite set of queries $\mathcal Q$, a finite set of answers $\mathcal A$ and
a function $f:\mathcal U\times \mathcal Q\to \mathcal A$. A
classical data structure for such a problem
is given by a function $\phi$ that encodes elements of $\mathcal U$ as bit strings
and a classical cell-probe algorithm with auxiliary input domain $\mathcal Q$
that computes $f(u, q)$ on input $\phi (u)$ and auxiliary input $q$.
This algorithm is called the \em query algorithm\em\footnote{Sometimes,
algorithms in the bit-probe and cell-probe model are called query
algorithms and the bit-probes or cell-probes are called queries. We do not use
this terminology here to avoid confusion with data structure queries.}
of that data structure.
The \em space complexity \em
of the data structure is the maximal number of blocks/bits in $\phi (u)$
and the \em time complexity \em is the maximal number of blocks/bits read by the
query algorithm. We can either consider deterministic algorithms or
we can use probabilistic algorithms and allow some error probability.

A classical data structure with quantum access is defined similarly,
except that instead of classical algorithms, we have a
quantum cell-probe algorithm that, given the initial state $\ket q\ket 0$
and oracle transform $O_{\phi (u)}$,
computes $f(u,q)$. As for classical data structures, we can consider
either exact quantum algorithms which always have to return the
correct answer or we can allow some error probability.

Every data structure problem has two trivial solutions: The first
one minimizes time complexity by listing the answers to all queries.
This, in general, requires a lot of space. The bit vector scheme mentioned
in Section \ref{int:data} is an application of this type of solution to the set
membership problem. The second one minimizes
space complexity by encoding elements of $\mathcal U$ using the information
theoretic minimum of bits, $\log |\mathcal U|$.
This usually causes a large time complexity.
In studying data structures, we look for interesting trade-offs between
these extremes.

Every data structure that uses block-size $w$ and has space complexity
$s$ and time complexity $t$ can be converted to a data structure with
block-size 1, space complexity $ws$ and time complexity $wt$. The
converse is not necessarily true. For a survey on classical data
structures and the classical cell-probe model see \cite{dstruct_survey}.

\subsection{Set Membership and Predecessor Search}

The two problems that we focus on in this survey are \em set membership \em
and \em predecessor search\em . These problems are parametrized by
positive integers $m$ and $n$ with $m\geq n$. The universe $\mathcal U$ in both cases
consists of the sets $S\subseteq [m]$ such that $|S|\leq n$.
In the set membership problem, we want to store a set $S$ so
that we can answer for every $i\in [m]$ the question
``Is $i\in S$?''. In the predecessor search problem, we want to
store $S$ so that we can answer for every $i\in [m]$ the question
``Is there some $j\in S$ with $j < i$ and if yes, which is the greatest $j$
with that property?''. The block-sizes we consider are 1 and $\log m$.

\subsection{Outline}

In Section \ref{class_data}, we present a data structure for each
of the problems we just described: The Perfect Hashing scheme for set
membership and a data structure given by Beame and Fich for predecessor search.
Section \ref{quant_acc} explains proofs for lower bounds in the quantum access
model which show that the data structures given in the previous
section are optimal in the sense that faster query algorithms would require
more space. Finally, Section \ref{full_quant}
describes some ``fully quantum'' data structures that encode
the data not in classical bits but in qubits and shows an application of such
data structures in \em quantum random walk \em algorithms.

\section{Classical Data Structures}
\label{class_data}

\subsection{Static Set Membership: Hash Tables}
\label{mem_class}

Suppose that we want to store a subset $S\subseteq [m] = \{ 0,\dots ,m-1 \}$
for $m\in\mathbb N$ and we want to be able to answer
membership queries ``Is $i\in S$?'' for every $i\in [m]$.

The most straightforward way is to store such a set $S$
as a bit vector $\phi (S)$ with $\phi (S)[i] = 1$ if and only if $i\in S$.
This allows to answer membership queries reading exactly one cell:
We just have to read the cell that contains the $i$th bit of $\phi (S)$
and output the value of that bit. However, this structure always uses
$m$ bits of space or $\lceil m/\log m \rceil$ blocks for block-size
$\log m$, so if the size $n$ of $S$ is small compared
to $m$, this method wastes a lot of space. In this subsection, we will
see the \em Perfect Hashing \em scheme
developed by Fredman, Koml\' os and Szemer\' edi in
\cite{fks} that stores a set $S$ of size $n$ using $O(n)$
blocks of size $\log m$ while answering set membership queries requires
only $O(1)$ cell-probes. Our presentation follows \cite[Section 11]{algo}.

Recall the hash table method of storing a set from Section \ref{int:data}.
Given the name, it is not
surprising that the Perfect Hashing scheme is similar to this method. While hash
tables are quite fast in practice, we saw that they have a worst-case time
complexity of $\Omega (n)$.
To avoid this, we are going to make the following changes: First, we
do not rely on a single hash function but we use a \em universal
class of hash functions\em , i.e., a class of hash functions with the
property that, for any
two given integers $i,j\in [n]$ and $h$ randomly selected from that
class, it is unlikely that $i$ and $j$ collide. Second, we expand
the size of our table to $O(n^2)$ and show that a universal class
of hash functions for that table size contains some function that
has no collision on $S$ of size $n$.
Third, to reduce the size back
to $O(n)$, we use a hash table of size $n$ and resolve collisions
by storing the set $S_i$ of elements of $S$ that collide in slot $i$
(here, $S_i$ is a random variable depending on the selected hash function
$h$) in
a hash table of size $O(|S_i|^2)$. We show that we can choose a hash function
from a universal class such that $\sum _i |S_i|^2 < 2n$.
\begin{defin}
A \em universal class of hash functions \em from $[m]$ to $[k]$ is
a set $\mathcal H$ of functions $h:[m]\to [k]$ such that, for any
two distinct $i,j\in [m]$, if we select $h\in\mathcal H$ uniformly
at random, the probability that $h(i) = h(j)$ is at most $1/k$.
\end{defin}
The following lemma gives an upper bound on the probability that a
hash function selected uniformly at random from a universal class
of hash functions from $[m]$ to $[n^2]$ has a collision on a fixed
set $S$ of size $n$.
\begin{lem}
\label{lem1}
Let $\mathcal H$ be a universal class of hash functions
from $[m]$ to $[n^2]$
and $S\subseteq
[m]$ of size $\leq n$. The probability that a randomly selected hash function
$h\in\mathcal H$ has a collision in $S$ is less than 1/2.
\end{lem}
\begin{proof}
Let $C$ be the random variable that counts the collisions on $S$ for
$h\in \mathcal H$ selected uniformly at random. That is,
$C = |\{ (i,j)\in S^2\mid i< j,h(i) = h(j)\}|$. If we let $C_{i,j}$
be the indicator random variable for $h(i) = h(j)$, we have
$$C = \sum _{i,j\in S,i<j} C_{i,j}$$
and since $\mathcal H$ is universal, $\mathbb E[C_{i,j}] \leq 1/n^2$. By
linearity of expectation we have
\begin{align*}
\mathbb E[C] = \sum _{i,j\in S,i<j} \mathbb E[C_{i,j}] \leq
    \sum _{i,j\in S,i<j}\frac{1}{n^2}
    = {n\choose 2} \frac{1}{n^2}
    = \frac{n^2 - n}{2}\cdot\frac{1}{n^2}
    = \frac{1}{2} - \frac{1}{2n}
    < \frac 1 2
\end{align*}
Since we have $\text{Pr}[C\geq 1]\leq\mathbb E[C]$ by Markov's inequality,
we can conclude that $h\in\mathcal H$ selected uniformly at random
has a collision on $S$ with probability less than 1/2.
\end{proof}

To construct the $O(n^2)$-size data structure, we still need something more than that:
We need, for every $m$ and $n$, a universal class of hash functions
such that each function can be uniquely identified by $O(\log m)$ bits.
We will now give a construction for classes of hash
functions that we later show to be universal and that can fulfill
this size requirement.
\begin{defin}
\label{concrete_hash}
For $m\in \mathbb N$ and $k \leq m$ and a prime number $p\geq m$, define the class
of hash functions $\mathcal H_{p,m,k}$ from $[m]$ to $[k]$ in the following
way:
For
$a\in \mathbb Z_p^*$, let
$$h_{p,m,k,a} : [m]\to [k], x\mapsto (ax\bmod p)\bmod k$$
and let
$$\mathcal H_{p,m,k} = \{ h_{p,m,k,a} \mid a\in\mathbb Z_p^*\}$$
\end{defin}
If $m$ is known, $h_{p,m,k,a}$ can be easily computed given
$p,k,a$. Clearly, $k$ can be stored in one $\log m$-bit block since
$k\leq m$. We can also choose $p\geq m$ such that $p$ and $a$
can be represented in a constant number of blocks, as follows:

By \cite[Theorem 7.32]{kl08}, there exists a constant $c$
so that for any $r\geq 1$, the number of primes
that require exactly $r+1$ bits to be represented\footnote{I.e.,
the number of primes $p$ with $2^r\leq p < 2^{r+1}$}
is at least
$c\cdot 2^{r}/(k+1)$. Substituting $\log m$ for $k$, we see that there are
at least $c\cdot m/(\log m +1)$ primes $p$ such that $m\leq p < 2\cdot m$.
Hence, there is a prime $p\geq m$ that can be represented with at most
$2\log m$ bits. Thus, $p$ and $a$ can both be stored in two $\log m$-blocks
each. The total number of $\log m$-blocks to store $p$, $k$ and $a$ is
therefore $l = 5$.
We can probabilistically find such a prime $p$ in time
polynomial in $m$ by randomly checking numbers in the appropriate
range for primality
(see \cite[Sections 7.2.1 and 7.2.2]{kl08}).

Let us now prove that the classes in Definition \ref{concrete_hash} are indeed
universal.
\begin{thm}
The classes $\mathcal H_{p,m,k}$ from Definition \ref{concrete_hash}
are universal classes of hash functions.
\end{thm}
\begin{proof}
Consider any two distinct natural numbers $y,z < p$. Let
$\tilde{h} _{a} (x) = ax\bmod p$. We first show that for $a\in
\mathbb Z_p^*$ selected uniformly at random,
$\tilde{h}_{a}(y)-\tilde{h}_{a}(z)$ is a uniformly random element of $\mathbb Z_p^*$.
Indeed, $\tilde{h}_{a}(y)-\tilde{h}_{a}(z) = a(y-z)\bmod p$.
We have $y-z\bmod p\neq 0$
since $y\neq z$ and $y,z < p$. For any $b\in\mathbb Z_p^*$, the probability
that $a = b(y-z)^{-1}$ is $1/(p-1)$.

We have $h_{p,m,k,a}(y) = h_{p,m,k,a}(z)$ if and only if
$\tilde{h}_a (y) - \tilde{h}_a (z)\bmod k = 0$. Since for any given
$l\in [n]$ there are
at most $\lceil p/k \rceil - 1$ elements $b$ of $\mathbb Z_p^*$
such that $b\equiv l \mod k$, it follows that the probability
of $\tilde{h}_a (y) - \tilde{h}_a (z) \equiv 0\mod k$ is at most
$$\frac{\lceil p/k\rceil -1}{p-1}\text .$$
We have
$$\lceil p/k\rceil - 1 \leq \frac{p+k-1}k -1 = \frac{p-1}k\text .$$
and therefore, the probability that $h_{p,m,k,a}(y) = h_{p,m,k,a}(z)$
is
$$\frac{\lceil p/k\rceil -1}{p-1}\leq \frac{p-1}{k(p-1)} = \frac 1 k$$
and thus, the definition of a universal class is satisfied.
\end{proof}
\begin{thm}[Quadratic Hash Table]
\label{quhash}
There is a data structure with block-size $\log m$
that stores subsets of $[m]$ of
size $n\leq \sqrt{m}$
in $O(n^2)$ blocks such that the membership query algorithm needs to make
a constant number of cell-probes.
\end{thm}
\begin{proof}
We store $S\subseteq [m]$
of size $\leq n$ in the following way:
Let $p$ be a prime number larger than $m$. By Lemma \ref{lem1}
and since $\mathcal H_{p,m,n^2}$ is universal,
there exists some $h_{p,m,n^2,a}\in\mathcal H_{p,m,n^2}$ that has no collision on $S$. Let
$a$ be a number such that $h = h_{p,m,n^2,a}$ has that property.
We set aside the first $l = 5$ blocks to store
$p$, $a$ and $n$. Then, we append an array $A$ of $n^2$ blocks that we fill
in as follows:
\begin{enumerate}
\item For $i\in S$, we store $i$ in $A[h(i)]$.
\item For every $j\in [n^2]$ such that $A[j]$ has not been filled in step 1, we
indicate that no element of $S$ has hash value $j$ by storing in $A[j]$ some
$k\in [m]$ with $h(k) \neq j$.
\end{enumerate}
By our choice of $a$, no collisions can occur in step 1.
Thus, for every $i\in S$, $A[h(i)]$ contains $i$.
Step 2 makes sure that $A[h(i)]$ contains $i$ only if $i\in S$.

The query algorithm to determine whether $i\in S$ works as follows:
We read the first $l$ blocks to determine the hash function $h = h_{p,n^2,a}$ that
was used for storing the data. We then read $A[h(i)]$. If $A[h(i)] = i$, we
output ``Yes'' and otherwise ``No''. This query algorithm reads $l+1 = O(1)$ blocks.
\end{proof}
This data structure already improves on the space complexity $O(m)$ of the bit vector
method for $n\leq \sqrt{m}$. But, as promised,
we can do better.
The Perfect Hashing method by Fredman, Koml\' os and Szemer\' edi uses two
layers of hashing. The first
hash table has size $O(n)$ and each cell $i$ holds a pointer to a
quadratic hash table from the proof of Theorem \ref{quhash}. The
elements $j\in S$ that collide in slot $i$ are stored in that
quadratic hash table. We show that a universal class of hash functions
contains a function $h$ for the first layer such that the sizes of the tables
in the second layer add up to no more than $O(n)$.
\begin{thm}
\label{col_count}
Let $\mathcal H$ be a universal class of hash functions from $[m]$
to $[n]$ and let $S\subseteq [m]$ be a set of size $n$. Let
$N_i$ be the random variable that for $h\in\mathcal H$ selected
uniformly at random counts the
elements $j\in S$ such that $h(j) = i$. Then,
$$\mathbb E\left[ \sum _{i\in [n]} N_i^2 \right] < 2n$$
\end{thm}
\begin{proof}
For any integer $a$ it holds that
$$a^2 = a + 2\cdot\frac{a^2 - a}2 = a + 2\cdot {a\choose 2}$$
and thus, by linearity of expectation,
\begin{align*}
\mathbb E\left[ \sum _{i\in [n]} N_i^2 \right] &=
\mathbb E \left[ \sum _{i\in [n]} N_i + \sum_{i\in [n]} 2\cdot {N_i\choose 2}\right]\\
&= \mathbb E\left[ \sum _{i\in [n]} N_i \right] + 2\cdot \mathbb E\left[ \sum_{i\in [n]}
{N_i\choose 2}\right]
\end{align*}
Since every element of $S$ hashes to exactly one value in $[n]$,
we have $\sum _{i\in [n]} N_i = n$. Now, we show that the
second term is upper bounded by $n$. Note that ${N_i \choose 2}$
is the number of pairs $(j,k)\in S^2$ with $j<k$ such that
$h(j) = h(k) = i$. Thus, $\sum _i {N_i\choose 2}$ is the number of
collisions of $h$ on $S$, i.e., the number of pairs $(j,k)$ with
$j<k$ such that $h(j) = h(k)$. Analogous to the proof of Lemma \ref{lem1},
we can use the universality of $\mathcal H$ to conclude that
$$\mathbb E\left[ \sum _{i\in [n]} {N_i\choose 2}\right] \leq {n\choose 2}\cdot\frac 1 n
=\frac{n-1}2$$
and hence,
$$\mathbb E\left[ \sum _{i\in [n]} N_i^2 \right] \leq \mathbb E[n] + 2\cdot\frac{n-1}2 = n + n-1 < 2n$$
as claimed.
\end{proof}
\begin{cor}[Perfect Hashing]
\label{perf_hash}
There is a data structure for storing subsets of $[m]$ with size $n$
using $O(n)$ blocks of size $\log m$ such that set membership queries can be
answered by reading $O(1)$ cells.
\end{cor}
\begin{proof}
Without loss of generality, we assume that $m = \omega (n)$, for otherwise the
bit vector method already has a space requirement of only $O(n)$.
We let $p$ be a prime that is greater than $m$.
Let $S\subseteq [m]$ with $|S|\leq n$ be the set we wish to store.
For $h\in\mathcal H_{p,m,n}$, let $S_{h,i} = \{ j\in S\mid h(j) = i\}$
and $n_{h,i} = \max \left( |S_{h,i}|, 1\right)$. By Theorem \ref{col_count},
there is some $h\in\mathcal H_{p,m,n}$ such that
$\sum _{i=0}^{n-1}(|S_{h,i}|)^2 < 2n$. Let $a$ be an element of
$\mathbb Z_p^*$ such that $h = h_a$ has that property and let
$n_i = n_{h,i}$.
For every $i\in [n]$, select $a_i\in \mathbb Z_p^*$
such that the function $h_i : x\mapsto (a_ix\bmod p)\bmod n_i^2$ has no
collisions on $S_i = S_{h,i}$. The existence of these $a_i$ follows from Lemma
\ref{lem1}.

Now, we store $a$, $p$ and $n$ in the first $l$ blocks.
We append an array $A$ with $n$ entries such that, for every
$i\in [n]$, $A[i] = (a_i,n_i)$, and an array $P$ of the same size with
$P[i] = \sum _{j=0}^{i-1}|S_j|^2$.
The array $A$ stores $n$ entries consisting of 3 blocks each, so it uses
$O(n)$ blocks. Since $\sum _{i=0}^{n-1}|S_i|^2 < 2n$, the
array $P$ requires $O(n\cdot \log m)$ bits or $O(n)$ blocks.

We construct for each $i$ a quadratic hash table for $S_i$ as in the
proof of Theorem \ref{quhash} and concatenate all the hash tables in
ascending order for $i$ (we leave out the first $l$ blocks of each table
containing the information about the hash function,
since that is already stored in array $A$). The entries in array $P$ were
selected so that $P[i]$ is the starting position of the $i$th hash table
in this concatenation. We call the concatenated
table $T$ and append it to the data structure. That table requires
$O\left( \sum _{i=0}^{n-1}|S_i|^2\right) = O(n)$ blocks of space.
Thus, we use $O(n)$ blocks in total.
See Figure \ref{fig_perfect} for an illustration of this data structure.

Now, queries whether $j\in S$ are answered in the following way:
\begin{enumerate}
\item Read the first $l$ blocks to find out which hash function
$h = h_{p,m,n,a}$ was used for the primary hash table.
\item Compute $i = h(j)$ and retrieve the parameters $(a_i, n_i) = A[i]$ of
the secondary hash table we need to access and its starting point $s_i = P[i]$.
\item Read $j' = T[((a_i\cdot j\bmod p)\bmod n_i^2) + s_i]$. If $j = j'$,
answer ``Yes'' and otherwise ``No''.
\end{enumerate}
This algorithm requires reading $l + \log (2n)/\log m + 4 = O(1)$ blocks.
\end{proof}

\begin{figure}[h]
\centering
\includegraphics[scale=0.6]{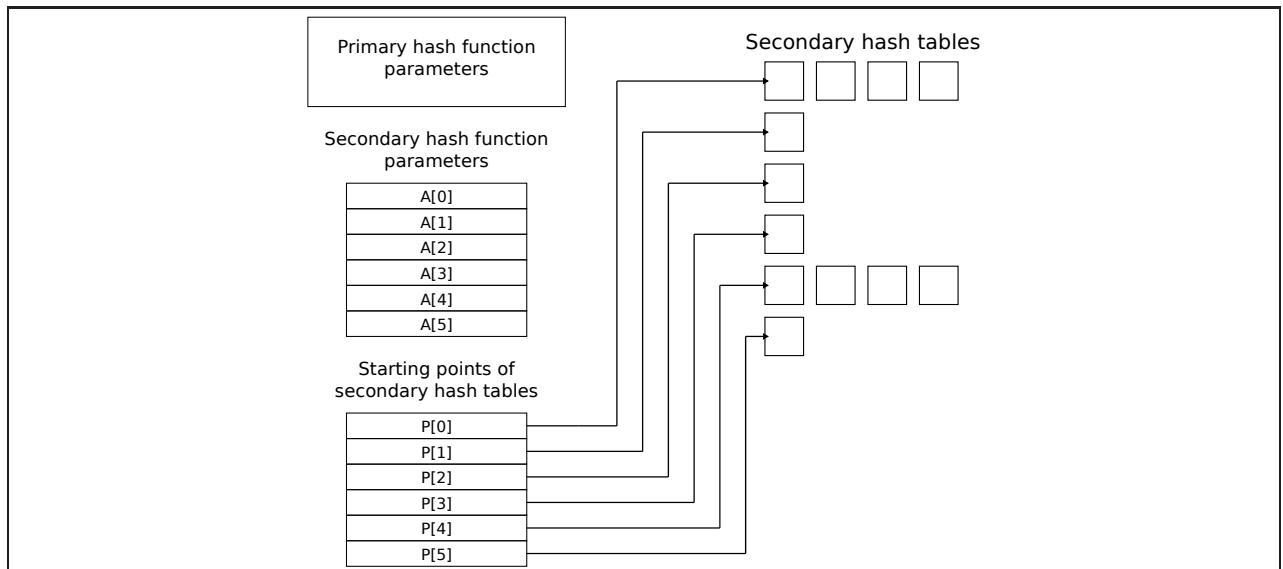}
\caption{A Perfect Hash table with $n = 6$. The $i$th entry of the table
labeled ``Starting points of secondary
hash tables'' stores the beginning of the table that stores the
elements of $P$ with hash value $i$. The $i$th entry of the table labeled
``Secondary hash function parameters'' stores the parameters for that table.
From the sizes of the secondary tables, we can see that two elements collide
at hash value 0 and at hash value 4 and that no other collisions occur.}
\label{fig_perfect}
\end{figure}
\begin{rem}
\label{rank_hash}
The Perfect Hashing scheme also helps us solve the dictionary problem.
If the data associated with each integer fits in $O(1)$ blocks, we can simply
store this data next to the integer in the hash table without affecting the
asymptotic time and space complexity.
We will make use of this fact in Section \ref{pred_bound}
to prove a lower bound on the time complexity of data structures
for the predecessor search problem, given a certain upper limit for the space,
by storing the following information in the dictionary.

Let the rank of some $i\in [m]$ in $S$ be defined
by $\text{rank}_S(i) = |\{ j\in S\mid j\leq i\}|$, the number of elements
in $S$ that are not greater than $i$.
If we allow two cells for each entry in the hash table,
we can store in addition to each $i\in S$ its rank $\text{rank}_S(i)$.
Note that this is \em not \em an efficient data structure for
the rank problem where we want to store a set $S$ such that we can
determine the rank in $S$ of \em any \em $i\in [m]$.
\end{rem}
The following lemma shows that (for $n$ small enough)
the Perfect Hashing scheme is asymptotically
optimal in the classical $(\log m)$-bit cell-probe model.
That is, its space complexity differs from the
information-theoretic minimum only by a constant factor while its
time complexity is constant.
\begin{lem}
\label{optimal}
Let $n\leq m^d$ for some constant $d$ with $0 < d < 1$. Then, the
minimum number of bits required for storing subsets of $[m]$ with
size at most $n$ is
$$\log \sum_{i=0}^n{m\choose i} \geq n(1-d)\log m = \Omega (n\log m)$$
and the minimum number of $(\log m)$-cells is therefore
$\Omega (n)$.
\end{lem}
\begin{proof}
Since the number of sets $S\subseteq [m]$ with $|S|\leq n$ is
$\sum _{i=0}^n{m\choose i}$, the minimum number of bits required for
a data structure that stores such sets is $\log \sum_{i=0}^n{m\choose i}$.
If $n \leq m^d$, we have
\begin{align*}
\log \sum_{i=0}^n{m\choose i}\geq \log {m\choose n}
\geq \log\left(\left( \frac m n\right) ^n\right)
&= n(\log m - \log n)\\
&\geq n(\log m - d\log m)\\
&= n(1-d)\log m\\
&= \Omega (n\log m)
\end{align*}
as claimed.
\end{proof}
However, this result does not show that the Perfect Hashing scheme is
asymptotically optimal in the \em bit-probe \em model.
In the bit-probe model, the Perfect Hashing scheme has a time complexity of
$O(\log m)$ bit-probes and a space complexity of $O(n\log m)$ bits. The space
complexity is asymptotically optimal.
But can we reduce the time complexity without increasing the space complexity?
In Section \ref{quantum_set}, we will see that the answer is ``no'' as long as
we do not allow \em two-sided error\em . The Perfect Hashing
scheme is even asymptotically optimal in the quantum bit-probe model,
both exact and with one-sided error.
However, if we allow two-sided error, we can be faster: In \cite{bmrv},
Buhrman, Miltersen, Radhakrishnan and Venkatesh describe a classical data structure
that uses space $O((n/\epsilon ^2)\log m)$ and answers membership queries
with two-sided error probability at most $\epsilon$ using only one bit-probe.
By setting $\epsilon$ to some small constant, we get a
data structure for the set membership problem with time complexity 1
and (up to a constant factor) minimal space complexity.

\subsection{Predecessor Search: Beame \& Fich}
\label{classical_pred}
If we want to be able to
quickly find the predecessor of $x\in [m]$ in a stored set $S\subseteq [m]$,
i.e., the largest $y\in S$ such that $y < x$, the data structures
described in Section \ref{mem_class} are not very helpful. A better
solution was found by Beame and Fich in \cite{bf}. Their data
structure can store sets $S\subseteq [m]$ with $|S|\leq n$ in
$O(n^2\log n/\log\log n)$ blocks of size $\log m$ bits
and answers predecessor queries with
$$O\left(\min \left( \frac{\log\log m}{\log\log\log m},
\sqrt{\frac{\log n}{\log\log n}}\right)\right)$$
cell-probes.
In the same paper, they also proved a matching lower bound in the classical
deterministic cell-probe model for the time complexity under the condition
that the space complexity is $O(n^2\log n/\log\log n)$.
A simpler proof was given by Sen and Venkatesh in \cite{qu_cell_probe}
for a restricted version of quantum cell-probe algorithms. This restricted model
still encompasses the classical probabilistic and deterministic
cell-probe model.
The proof by Sen and Venkatesh is described in Section \ref{pred_bound} of this survey.

The data structure invented by Beame and Fich
needs to be combined with other data structures to be efficient, namely,
\em X-fast Tries \em and \em Fusion Trees\em . But first we will focus on the
contributions by Beame and Fich.
One building block for their data structure is the parallel hash
table which, given a large enough block-size, supports membership queries
to multiple sets with a constant number of queries.
\begin{lem}[Parallel Hash Table]
\label{par_hash}
Let $q$ be a positive integer and $w = q\log m$.
There is a data structure that stores $q$ sets $S_0,S_1,\dots ,S_{q-1}
\subseteq [m]$ each of size at most $n$ using $O((2n)^q) = O(2^{(\log n +1)q})$ blocks of
size $w$ such that every $q$-tuple of queries of the form
queries $(x_0\in S_0 ?,x_1\in S_1?,\dots ,x_{q-1}\in S_{q-1}?)$ can be
answered with a constant number of cell-probes, independent of
$q$.\footnote{This is possible since $q$ is absorbed into the block-size.}
\end{lem}
\begin{proof}
To represent a collection $S_0,S_1,\dots ,S_{q-1}$ of sets, we first
create for every $i$ the arrays $A$, $P$ and $T$ from the proof of Corollary
\ref{perf_hash} (with block-size $\log m$).
Let $p_i$, $a_i$ and $n_i$ be the parameters of the
primary hash function $h_i$ used in storing $S_i$
and let $A_i$ and $P_i$ be the first two arrays in the resulting Perfect Hash
table for $S_i$. Let $T_i$ be the concatenation of the secondary tables of
the $i$th Perfect Hash table.

The first things that we store are $p_i,a_i,n_i$ for
every $i$. Using $(q\log m)$-size blocks, this requires $O(1)$ blocks.
Let $j_0,\dots ,j_{q-1}\in [n]$ and $j = j_{0}\circ \dots \circ j_{q-1}$ be the concatenation
of the binary representations of these numbers. We construct arrays
$A'$ and $P'$ of size $n^q$ such that for each such $j$,
\begin{align*}
A'[j] &= (A_0[j_0], A_1[j_1], \dots , A_{q-i}[j_{q-1}])\\
P'[j] &= (P_0[j_0], P_1[j_1], \dots , P_{q-i}[j_{q-1}])
\end{align*}
Each entry in these arrays can be stored in a constant number of
$(q\log m)$-blocks, so the total size of each of these tables is
$O(n^q)$. Furthermore, we construct an array $T'$ of size $(2n)^q$
such that for $j_0,\dots ,j_{q-1}\in [2n]$ and $j$ their concatenation
$$T'[j] = (T_0[j_0], \dots , T_{q-1}[j_{q-1}])\text .$$
Again, each entry in the table can be stored in a constant number of
$(q\log m)$-blocks. The size of the whole table is thus $O((2n)^q)$. The
data structure consists of the arrays $A'$, $P'$ and $T'$ and thus
requires $O((2n)^q)$ blocks of size $q\log m$.

Let us now see how we answer $q$ parallel queries: We have
$x_0,x_1,\dots ,x_{q-1}\in [m]$ and we want to answer the questions
$x_i\in S_i?$ for every $i = 0,\dots ,q-1$. We read a constant number of blocks
to find $p_i$, $a_i$ and $n_i$ for $i = 0,\dots ,q-1$ and can now
compute the hash functions $h_i : x\mapsto (a_ix\bmod p_i)\bmod n_i$ that were
used in the Perfect Hashing tables.
We let $j_i = h_i(x_i)$ and $j = j_0\circ \dots \circ j_{q-1}$. We read
$A'[j]$; let $(a_0',n_0'),(a_1',n_1'),\dots ,(a_{q-1}',n_{q-1}')$ denote the
content.
We also read $P'[j]$ and let $s_0,s_1,\dots ,s_{q-1}$ be
the values stored there. Now let $k_i = (a'_ix\bmod p)\bmod n_i'^2 + s'_i$
and $k = k_0\circ k_1\circ \dots \circ k_{q-1}$. Finally, we read $T'[k]$. Thus,
we get for each $i = 0,\dots ,q-1$ the value $T_i[k_i]$ where
$k_i$ is the position of $T_i$ that we would read when we search for $x_i$ in the
Perfect Hash table of $S_i$. Hence, we can answer all the
queries. We have to read a constant number of $(q\log m)$-blocks for
this algorithm.
\end{proof}

We also need the concept of \em tries\em .
\begin{defin}
Let $\Sigma$ be some alphabet. A \em trie \em over $\Sigma$ is a
tree where each node is some word from $\Sigma ^*$ and
each edge is labeled with a letter from $\Sigma$ such that
the following conditions are satisfied:
\begin{itemize}
\item The root is the empty word.
\item If $u$ is a non-root node and $v$ its parent,
then there is some $\sigma\in \Sigma$ such that
$\sigma$ is the label of the edge between $u$ and its parent
and $u = v\sigma$.
\end{itemize}
We say that a trie \em stores \em a word $\delta$ if it is one of the trie's leaves.
Similarly, we say that it stores a set of words if it stores exactly those
words that are in the set.
\end{defin}
If we think of a trie as a deterministic
automaton where the leaves are the accepting states and the root
is the starting state, the set of
words that the trie stores is the language accepted by it.
The Beame \& Fich data structure
uses tries to store sets $S\subseteq [m]$ where we view elements of $S$ as
words over the alphabet $\left[ m/2^c\right]$ for some $c$. At a
certain subset of the nodes we store information that helps in
finding predecessors. The nodes that we select for this
are determined by a property that is called \em heaviness \em which is defined
as follows:
\begin{defin}
Let $S$ be a set of $s$ strings of length $L$ over the alphabet $[m]$
with $0 < s \leq n$.
Let $T$ be the trie of depth $L$
that stores $S$. A node $u$ in $T$ is called $n$-\em heavy\em , or simply
\em heavy \em since $n$ is always understood from the context, if
the subtree rooted at $u$ has at least $\max (s/n^{1/L}, 2)$ leaves.
\end{defin}
The root is always heavy and a parent of a heavy node is heavy as well.
See Figure \ref{fig:heavy} for an illustration. The depicted trie over the
alphabet $\{ 0,1 \}$ stores
$S = \{ 0010, 0011, 0100, 1101, 1110, 1111 \}$, so we have
$s = 6$ and $L = 4$. The black nodes are $n$-heavy for $n = 16$, since
we have $s/n^{1/L} = 6/16^{1/4} = 6/2 = 3$ and thus, a node is $n$-heavy
if and only if its subtree has at least $3$ leaves. For $n < 16$, a
node needs more than $3$ leaves in its subtree to be heavy. For such $n$,
only the root is $n$-heavy in our illustration.

\begin{figure}[h]
\centering
\includegraphics{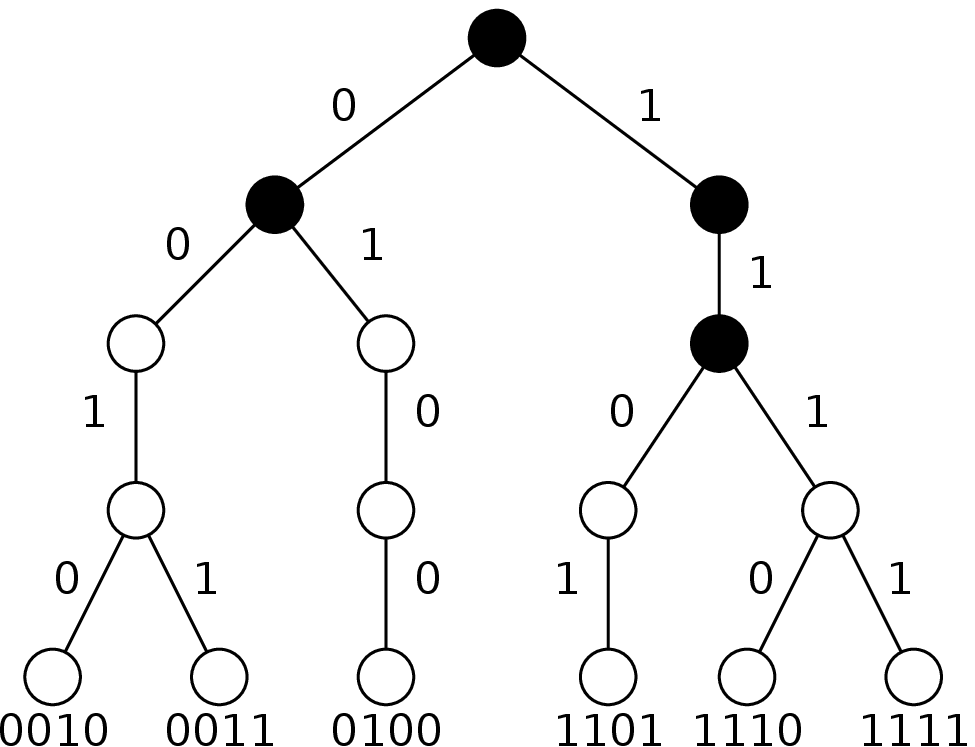}
\caption{A trie over the alphabet $\{ 0,1 \}$ storing a set $S$ of $s = 6$
strings of length $L = 4$. For $n = 16$, the $n$-heavy nodes are colored black.}
\label{fig:heavy}
\end{figure}

Using parallel hashing, we can construct a data structure to
store a trie that allows to search for the longest heavy prefix
of a string with a constant number of cell-probes.
\begin{lem}
\label{par_trie}
Let $T$ be a trie over the alphabet $[2^k]$ of depth $L$ with at most $n$
leaves such that $2L(L-1)\leq\log n$.  There is a data structure with block-size
$b = \Theta (kL)$ that can store $T$ in $O(n)$ blocks
such that, given a string $x=x_1x_2\dots x_{k}$ with $x_i\in [2^k]$, the
longest proper prefix $x'$ of $x$ such that $x'$ is a heavy node can be found
with $O(1)$ cell-probes.
\end{lem}
\begin{proof}
For $d = 1,\dots ,L-2$ let $S_d$ be defined as
$$S_d = \{ z\in [2^k]\mid yz\text{ is a heavy node at depth }d\text{ for some }y\text .\}\text .$$
Since there are at most $n^{1/L}$ heavy nodes at each depth, we have $|S_d|\leq n^{1/L}$.
We store $S_1,\dots ,S_{L-2}$ in a parallel hash table, as in
Lemma \ref{par_hash}. The block-size is $qk = (L-1)k$ and we need
$O\left( 2^{((\log (n)/L) + 1)(L-2)}\right)$
blocks to store the parallel hash table. We have
\begin{align*}
\left( \frac{\log n}{L} + 1\right) (L-2) < \left(\frac{\log n}{L} +2\right)(L-1)
&= \frac{\log (n) (L-1) + 2L(L-1)}{L}\\
&\leq \frac{\log (n) (L-1) + \log n}{L}\\
&= \log n
\end{align*}
and hence we need $O(2^{\log n}) = O(n)$ blocks of memory.

To find the longest heavy proper prefix of $x = x_1x_2\dots x_{k}$,
we evaluate the queries $x_1\in S_1?, x_2\in S_2?,\dots ,x_{k-1}\in S_{k-1}?$
in parallel. This requires only $O(1)$ cell-probes. Since the predecessor of a
heavy node is a heavy node as well, there must be some index $i$
with $0\leq i\leq k-1$ such that, for all $j\leq i$, we have
$x_j\in S_j$ and for $j> i$, $x_j\not\in S_j$.
If $i = 0$, the longest heavy prefix of $x$ is the empty string.
Otherwise, the string $x' = x_1\dots x_i$ is the correct answer.
\end{proof}
The data structure by Beame and Fich is constructed recursively;
in the proof of the following lemma we describe that construction.
\begin{lem}
\label{b+f}
Let $a,c,u,L,n,s$ be integers such that $a,c\in [u+1]$, $n\geq u^u$,
$1\leq L\leq u$ and $s\leq n^{a/u}$. For any $b \geq \left( 2(u-1)^2-1\right)Lu^c$
there is a data structure using block-size $b$ that allows
to store a set of $s$ integers from $\left[ 2^{Lu^c}\right]$ in $O(sn/u^2)$
blocks and allows to answer predecessor queries with $O(a+c)$ cell-probes.
\end{lem}
\begin{proof}
There are two different base cases for the data structure:
The first base case is $a = 0$.
We then have $s\leq 1$, so the sets that we store are singletons or empty.
We can store the single element in such a set using one block of
space.
The second base case is $L = 1$ and $c = 0$. Then, the integers in the set that we store come
from the universe $\{ 0,1 \}$, so we can store it as a 2-bit characteristic
vector.

If $L = 1$ and $c\geq 1$, we replace $L$ by $L' = Lu$ and
$c$ by $c' = c-1$. This change does not affect the value of $Lu^c$ and
thus it does not affect the universe size, block-size or any of the
other parameters. In the recursive instances, either $L$ will be set
to 1 (and we apply the substitution just described, if $c\geq 1$), or $a$ is reduced
by one, until we reach one of the base cases.

Now we assume that $a>0$, $L > 1$ and $c\geq 0$ and let $S$ be a subset
of $\left[ 2^{Lu^c}\right]$ of size at most $s$. Let $T_0$ be the
binary trie of depth $Lu^c$ that stores the set $S$. For $j$ such that
$0<j\leq c$, let $T_j$ be the trie of depth $Lu^{c-j}$ that consists of the
nodes in $T_0$ at all levels divisible by $u^j$; the parent relation in $T_j$
is defined as follows: For a node $v$ at level $k$ in $T_j$ (and hence at level
$u^{jk}$ in $T_{j-1}$), we let the parent of $v$ be the ancestor of $v$ in
$T_{j-1}$ at level $u^{j(k-1)}$. The trie $T_j$ encodes the set $S$ if we view
it as a set of strings over $\left[ u^j\right]$ of length $Lu^{c-j}$. For
$j = c$, we have a trie $T_c$ of depth $L$ that stores $S$ as a set of
length-$L$ strings over $\left[ u^c\right]$. Figure \ref{fig:reduce_tries}
illustrates the construction of $T_0$ and $T_1$ for $L = u = 2$ and $c = 1$.
Both tries store the set $S = \{ 1, 2, 9, 13 \}\subseteq [16]$. In $T_0$, these
numbers are encoded in their binary representations as bit strings of length 4,
in $T_1$ they are encoded in their base-4 representation as strings of length
$2$ over the alphabet $[4]$.

\begin{figure}[h]
\centering
\includegraphics{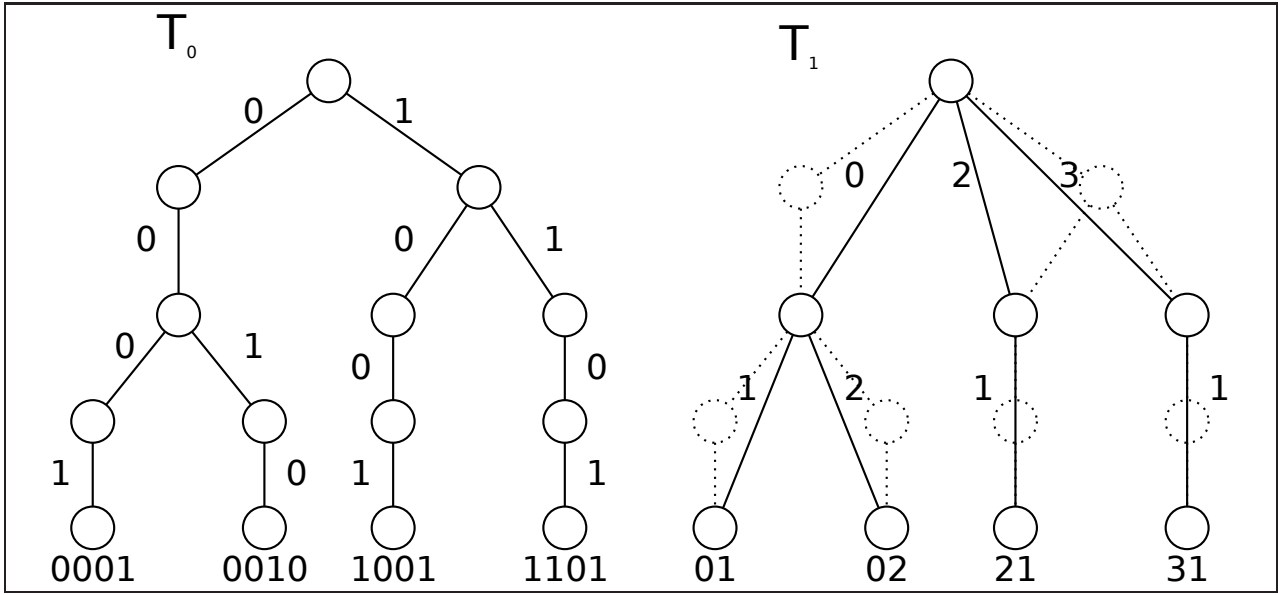}
\caption{An illustration of the construction of the tries $T_0$ and $T_1$.}
\label{fig:reduce_tries}
\end{figure}

For every node $v\in T_c$, we let
$\min _S(v)$ denote the minimal element of $S$ with prefix $v$ and
$\max _S(v)$ the maximal element with prefix $v$.
Every recursive instance of the data structure will be associated
with a node of one of the tries $T_0,\dots ,T_c$ for analysis purposes.

Our data structure consists of the following parts:
\begin{enumerate}
\item An instance of the data structure from Lemma \ref{par_trie} storing $T_c$.
\item For every heavy node $v$ in $T_c$ and every node $v$ that has a heavy parent,
we store $\min _S(v)$, $\max _S(v)$ and $\text{pred}(\min _S(v), S)$, the
predecessor of $\min _S(v)$ in $S$.
\item For every heavy node $v$ that has
at least two children, we store the labels of edges that lead
to the non-heavy children of $v$ in a Perfect Hash table. We also store
the set $S_v = \left\{ w\in \left[ 2^{u^c}\right] \mid vw\text{ is a child of }v\text{ in }T_c \right\}$
in a recursive instance of our data structure (with $L = 1$ and
the other parameters as before). That instance will be associated
with the node $v$ in the trie $T_{c-1}$.
\item For every non-heavy node $v$ at depth $d$ with $0 < d < L$
in $T_c$ that has a heavy parent and at least two leaves in its subtree, we store
the set $S'_v = \left\{ \left. w\in \left[ 2^{u^c}\right] ^{L-d}\right|
vw\in S\right\}$ in a recursive instance of our data structure.
The set $S'_u$ has size at most $s/n^{1/L}\leq s/n^{1/u}\leq n^{(a-1)/u}$.
We associate this instance with node $v$ in $T_c$ if $d<L-1$ and with
node $v$ in $T_{c-1}$ if $d = L-1$.
\end{enumerate}

Let us illustrate this by continuing the example from Figure
\ref{fig:reduce_tries}. Let $L = u = 2$, $n = 4$ and $S = \{ 1, 3, 9, 13 \}$.
Since $u/n^{1/L} = 2/\sqrt{4} = 1$, a node is heavy if and only if its
subtree contains at least two leafs (recall that a node needs at least 2 leaves
in its subtree to be heavy, regardless of the parameters). Thus, in the trie
$T_1$, the heavy nodes are the root and the node $0$. For the root node $r$,
we store $\min _S(r) = 1$, $\max _S(r) = 13$ and
$\text{pred}(\min _S(r), S) = \bot$.
We create a Perfect Hash table storing the labels of the edges to non-heavy
children of $r$, i.e., $2$ and $3$ and we store the set $S_r = \{ 0, 2, 3\}$
in a recursive instance of the data structure with parameters $L = 1$
and $u$ and $c$ as before (these parameters are then transformed to
$L = 2$ and $c = 1$). Since the node $0$ is heavy too, we store
$\min _S (0) = 1$, $\max _S(0) = 2$ and $\text{pred}(\min _S(0), S) = \bot$.
We store the labels of the edges to its non-heavy children in a Perfect Hash
table. That is, we store the set $\{ 1,2 \}$. We also store this set in a
recursive instance of our data structure with parameters $L = 1$ and $u$ and
$c$ as before. The leaf $01$ has a heavy parent, so we store
$\min _S(01) = 1$, $\max _S (01) = 1$ and $\text{pred}(\min _S(01), S) = \bot$.
Likewise, we store $\min _S(02) = 2$, $\max _S(02) = 2$ and
$\text{pred}(\min _S(02), S) = 1$. The node $2$ has a heavy parent, but
only one leaf in its subtree, so we store $\min _S(2) = 9$, $\max _S(2) = 9$
and $\text{pred}(\min _S(2), S) = 2$. For the node $3$, we store
$\min _S(3) = 13$, $\max _S(3) = 13$ and $\text{pred}(\min _S(3), S) = 9$.

The algorithm for finding the predecessor of some number $x$ in $S$
goes as follows: We view $x$ as a string of length $L$ over the alphabet
$\left[ 2^{u^c}\right]$.
First, we find the longest prefix $x'$ of $x$ such that $x'$ is a heavy node
in $T_c$. Since we store $T_c$ using the data structure from Lemma
\ref{par_hash}, we can find this prefix by making $O(1)$ cell-probes. We now
have to consider several cases.

If $x'$ has exactly one
child, then either $\min \{ y\in S\mid y \geq x\} = \min _S(x')$ or
$\max \{ y\in S\mid y < x\} = \max _S(x')$. This holds because if
$x'$ is heavy then its child $x'\sigma$ is heavy too. Then,
$x'\sigma$ is a prefix of all elements of $S$ that have $x'$ as
prefix, so $\min _S(x') = \min _S(x'\sigma )$ and $\max _S(x') = \max _S(x'\sigma )$.
But it is not a prefix of $x$, for otherwise, $x$ would
have a heavy prefix that is longer than $x'$. Therefore, either
$x$ is smaller than all elements of $S$ with $x'$ as prefix or
it is larger than all of them. Suppose that
there is some $y\in S$ with $\max _S(x'\sigma ) < y < x$ or
$\min _S(x'\sigma ) > y > x$. Since $x$, $\min _S(x'\sigma )$ and
$\max _S(x'\sigma )$ have the prefix $x'$ in common, it follows
that $y$ must have prefix $x'$ too. But since it is in $S$, it
must also have the prefix $x'\sigma$, contradicting
$\min _S(x'\sigma )> y$ and $\max _S(x'\sigma )<y$.

Hence,
$$\text{pred}(x,S) = \begin{cases} \max _S(x')&\text{ if }x> \min _S(x')\\
\text{pred}(\min _S(x'), S)&\text{ if }x \leq \min _S(x')\end{cases}$$
which we can compute reading $O(1)$ blocks in our data structure.

Suppose that $x'$ has at least two children. In our data structure,
we have a hash table that stores the labels on the edges that lead
to non-heavy children of $x'$. We use this table to find out whether
any of the non-heavy children is a prefix of $x$ which can be
done by checking whether $x_{d+1}$, the $(d+1)$st letter of $x$,
is in the table. This takes $O(1)$ cell-probes.

If there is no such child, then for each child $x'\sigma$ of $x'$ it holds
that the leaves of the subtree at $x'\sigma$ are either all larger or all
smaller than $x$. Thus, either $x$ is smaller than all leaves of the subtree
rooted at $x'$ or there is some child $x'\sigma$ of $x'$ such
that all leaves of the tree rooted at $x'\sigma$ are smaller than $x$. The
largest $\sigma$ with that property is the predecessor of $x_{d+1}$ in $S_{x'}$.
Thus, we have
$$\text{pred}(x,S) = \begin{cases} \text{pred}(\min _S(x'),S)&\text{ if }x\leq \min _S(x')\\
\max _S(x'\circ \text{pred}(x_{d+1},S_{x'}))&\text{ if }x > \min _S(x')\end{cases}$$
where $x_{d+1}$ is the $(d+1)$th letter of $x$.
We can decide with $O(1)$ cell-probes whether $x\leq\min _S(x')$.
We also can find $\text{pred}(\min _S(x'))$ with $O(1)$ cell-probes. We
find $\sigma =\text{pred}(x_{d+1},S_{x'})$ using the recursive instance of
our data structure. Since $x'\sigma$ is the child of a heavy node, we can
read off $\max _S(x'\sigma)$ with $O(1)$ cell-probes once we found $\sigma$.

If $x'$ has a non-heavy child $y$ that is a prefix of $x$ and if
that child has exactly one leaf in its subtree, we have
$$\text{pred}(x,S) = \begin{cases} \text{pred}(\min _S(y),S)&\text{ if }x\leq \min _S(y)\\
\min _S(y)&\text{ if }x > \min _S(y)\end{cases}$$
which we can compute making $O(1)$ cell-probes. Finally,
if $x'$ has a non-heavy child $y$ that is a prefix of $x$ and has
at least two leaves in its subtree, we have
$$\text{pred}(x,S) = \begin{cases} \text{pred}(\min _S(y),S)&\text{ if }x\leq \min _S(x')\\
y\circ \text{pred}(x_{d+1}\dots x_{L},S'_{y})&\text{ if }x > \min _S(x')\end{cases}$$
which can be computed reading $O(1)$ blocks, except
for the recursive call to find $\text{pred}(x_{d+1}\dots x_{L},S'_{y})$.

We show that this algorithm uses $O(a+c)$ cell-probes by induction. In
the base cases, we can read the whole data using $O(1)$ cell-probes. Let
$T(a,c)$ be the worst-case number of cell-probes to answer a query for the
given parameter values $a$ and $c$. Since in a recursive call, $a$ is reduced
by one, $c$ is reduced by one, or both and any computation outside of
recursive calls requires $O(1)$ cell-probes, we have
$$T(a,c) \leq \max \{ T(a-1,c), T(a,c-1), T(a-1,c-1)\} + k\text{ with }
k = O(1)\text .$$
We show that $T(a,c) = O(a+c)$ by proving that $T(a,c)\leq k(a+c)$. As
induction hypothesis, we assume that the inequality $T(a,c)\leq k(a+c)$
holds for all $a,c$ such that $a+c < n$. We prove that it then also holds when
$a+c = n$.  By induction hypothesis, we have
$$\max \{ T(a-1,c), T(a,c-1), T(a-1,c-1)\} \leq k(a+c-1)\text .$$ Thus,
$T(a,c) \leq k(a+c-1) + k = k(a+c)$.
Therefore, it holds that $T(a,c)= O(a+c)$.

As an example, we show how to find the predecessors of $8$ and $3$, continuing
from the example given in Figure \ref{fig:reduce_tries}. We find the predecessor
of $8$ as follows: In base-4, we
write $8$ as $20$. We first find the longest heavy prefix of $20$ in
$T_1$ which is the root $r$ of the trie. There is a hash table that stores
the labels of edges from $r$ to its non-heavy children. We check if this
table contains $2$, which it does. The node $2$ has exactly one leaf in
its subtree. We read $\min _S(2) = 9$. We now know that among all the leaves
in $T_1$, the leaf $9$ ($21$ in base-4) has the longest common prefix with
$8$ ($20$). Thus, the predecessor of $8$ in $S$ is also the predecessor of $9$
in $S$. We have stored the predecessor of $9$ in our data structure, so we
simply need to read it to answer the query. Thus, we learn that $2$ is the
predecessor of $8$ in $S$.

Let us now find the predecessor of $3$. We find that the longest heavy
prefix is $0$. Checking the hash table, we see that no element stored
in the trie has $03$ as prefix. We have $3 > \min _S(0) = 1$, so we
look for the predecessor of $3$ in $S_0 = \{ 1,2 \}$ which we have stored
in a recursive instance. This predecessor is $2$. Now we know that the
maximal element stored in the subtree rooted at $02$ is the predecessor
of $3$. We read $\max _S(02) = 2$ and find that $2$ is the predecessor of
$3$ in $S$.

It remains to check how much space our data structure requires. We
first count the space required for part 2 of our data structure,
including the part 2 of the recursive instances (and their recursive
instances, etc.). The trie $T_j$ contains at most $sLu^{c-j} +1$ nodes.
Each recursive instance is associated with
a subtree of some $T_j$ and for some of the nodes in these trees,
we store a constant number of memory words of length $Lu^j\leq u^{j+1}$.
In total, we store
$$O\left( \sum _{j = 0}^c \left( sLu^{c-j} + 1 \right) u^{j+1}\right)
= O\left( scLu^{c+1}\right)$$
bits.

We now count the bits for the hash tables in part 3 (over all recursive instances).
Every trie $T_j$ has $s$ leaves and thus at most $2(s-1)$ nodes that
have siblings. Each of those nodes contributes a constant number of
$u^j$-bit entries to its parent's hash table. The hash tables must
therefore have
$$O\left( \sum _{j = 0}^c su^j\right) = O\left( scu^c\right)$$
bits.

Now for part 1: There are at most $2(s-1)$ nodes in addition to the
root in $T_c$ that have a recursive instance associated with them
since these nodes are non-heavy children of heavy nodes, so they
must have siblings (see part 4). For each such node, we store
an instance of the data structure of Lemma \ref{par_trie}. This
requires $O(sLu^cn)$ bits in total. For $j<c$, each node in
$T_j$ for which we store a recursive instance either has a sibling
or is a node in $T_{j+1}$ with at least two children. Each tree
has $s$ leaves and therefore contains at most $s-1$ nodes with
at least two children and $2(s-1)$ nodes with siblings. For all
$T_j$ with $j<c$ taken together, we need
$$O\left( \sum _{j=0}^{c-1} su^{j+1}n \right) = O\left( scu^cn\right)$$
bits.

Taking all this together, our data structure consists of
$$O\left( scLu^{c+1}\right)+O\left( scu^c\right)+O\left( sLu^cn\right) + O\left( scu^cn\right)
= O\left( sLu^{c+2-u}n\right) + O\left( sLu^cn\right) = O\left( sLu^cn\right)$$
bits or $O\left( sn/u^2\right)$ blocks of size $b$.
\end{proof}
Setting the parameters to $L = 1$, $a = u$ and $s = n$, we obtain the
following lemma:
\begin{lem}
\label{bfdata}
Let $m$, $n$, $u$ and $c$ be integers such that $n\geq u^u$, $c\leq u$ and
$m\leq 2^{u^c}$. Let $b\geq 2u^{c+2}$.
Then, there is a data structure using block-size $b$ that stores
subsets $S$ of $[m]$ of size at most $n$ in $O(n^2/u^2)$ blocks such that
predecessor queries can be answered with $O(u+c) = O(u)$ cell-probes.
\end{lem}
The block-size for this data structure is not $\log m$ and there is
a lower bound on the values for $n$ we can choose. To compensate
for this, we combine this data structures with other data structures
for the predecessor problem. If $n$ is small enough, we store the data in a
\em fusion tree\em , a data structure by Fredman and Willard which they
described in \cite{fusion}. We will not describe this data structure here,
but summarize its properties.
\begin{lem}[Fusion Trees]
\label{lem:fusion}
Let $b$, $n$ and $m$ be positive integers such that $n < m < 2^b$. There
exists a data structure for the predecessor problem with block-size $b$,
space complexity $O(n)$ and time complexity $O(\log _b n) = O(\log n/\log b)$.
\end{lem}
If $n$ is larger, we need \em x-fast tries\em , a data structure invented by
Willard. Since this data structure will be connected more closely with the
structure from Lemma \ref{bfdata} than the fusion trees, let us have a closer
look at it.
\begin{thm}[X-fast Tries, \cite{xfast}]
Let $n < m$. There exists a data structure for the predecessor problem with
block-size $\log m$, space complexity $O(n\log m)$ and time complexity
$O(\log\log m)$.
\end{thm}
\begin{proof}
We interpret the numbers in $S$ as bit strings and
store them in a binary trie of depth $\log m$ with
the following augmentations:
Every node $u$ that has no right child additionally stores a pointer to
the maximal element of the subtree rooted at $u$. Similarly,
every node without a left child stores a pointer to the minimal element of the
subtree. For every level $j$ of the trie, we store
all the nodes at that level together with a
pointer to their position in the trie
in a Perfect Hash table. Furthermore, the leaves
in the trie form a sorted, doubly linked list. The
trie has at most $n\log m$ nodes, so it uses $O(n\log m)$ blocks. Since
every node is stored in only one of the hash tables, the hash tables
taken together also use $O(n\log m)$ blocks. This shows that the whole
data structure has space complexity $O(n\log m)$.

We now show how to implement predecessor queries.
Let $x\in [m]$. To find its predecessor, we first find
a node $u$ in the trie such that $u$ is the longest prefix of $x$
contained in the trie. We do this by binary search on the prefix
length and by using the hash tables to find out whether a prefix
of some given length is contained in the trie.

Having found this prefix $u$, there are three possible situations.
\begin{itemize}
\item $u = x$: We use the list structure of the leaves to
find the predecessor of $u$.
\item $u$ has no left child: Then, every leaf in the subtrie rooted
at $u$ is larger than $x$. We follow the pointer to the minimum
of that subtrie to the successor of $x$. Using the linked list
of the leaves, we find the predecessor of $x$.
\item $u$ has no right child: Every leaf in the subtrie rooted
at $u$ is smaller than $x$. Following the pointer at $u$ leads
us to the predecessor of $x$.
\end{itemize}
The node $u$ cannot have two children because in that case one of them would
have to be a prefix of $x$.

Since $x$ consist of $\log m$ bits, the binary search requires $O(\log\log m)$
lookups in the Perfect Hash tables, each of which can be done with $O(1)$
cell-probes.  After we have found the longest prefix,
only a constant number of cell-probes are required. This shows
that the time complexity is $O(\log\log m)$.
\end{proof}
Now, everything is in place for proving the main result of this section.
\begin{thm}
There is a data structure for the static predecessor problem
that stores $S\subseteq [m]$ with $|S|\leq n$ in
$O(n^2\log n/\log\log n)$ blocks of size $\log m$. Predecessor
queries are answered with
$$O\left(\min \left( \frac{\log\log m}{\log\log\log m},
\sqrt{\frac{\log n}{\log\log n}}\right)\right)$$
cell-probes.
\end{thm}
\begin{proof}
Let $b = \log m$ be the block-size. We distinguish two cases, depending on $n$.\\
\textbf{Case 1:} Suppose that
$n < 2^{4(\log\log m)^2/\log\log\log m}$. In that case, we store $S$ in a
fusion tree. For $m$ large enough,
$$\log\log n < 2+2\log\log\log m - \log\log\log\log m\leq 2\log\log\log m$$
and it follows that
$$\log\log m = \frac{2\log\log m}{\sqrt{\log\log\log m}}\cdot \frac{\sqrt{\log\log\log m}} 2
> \sqrt{\frac{\log n\log\log n}8}$$
which gives us
$$\frac{\log n}{\log b} = \frac{\log n}{\log\log m}
< \sqrt{\frac{8\log n}{\log\log n}}\text .$$
We also have
$$\frac{\log n}{\log b} = \frac{\log n}{\log\log m} < \frac{4\log \log m}{\log\log\log m}$$
and thus, by Lemma \ref{lem:fusion}, we can store $S$ in a fusion tree that uses
$O(n)$ blocks of space and supports predecessor queries with
$$O\left(\frac{\log n}{\log b}\right) = O\left(\min \left( \frac{\log\log m}{\log\log\log m},
\sqrt{\frac{\log n}{\log\log n}}\right)\right)$$
cell-probes.\\
\\
\textbf{Case 2:} If $n\geq 2^{4(\log\log m)^2/\log\log\log m}$, we have
$$\sqrt{\frac{\log n}{\log\log n}}\geq \sqrt{\frac{4(\log\log m)^2}{\log\log\log m(2+2\log\log\log m - \log\log\log\log m)}}
\geq \frac{\log\log m}{\log\log\log m}$$
for large enough $m$.
In that case, we combine the data structure by Beame and Fich with x-fast tries.
Let $u$ be the smallest integer such that $u^u\geq \log m$. Then,
\begin{equation}
\label{eq:beamefich}
\frac{\log\log m}{\log\log\log m}\leq \frac{u\log u}{\log u + \log\log u}
\leq u \leq \frac{2\log\log m}{\log\log\log m}
\leq 2\sqrt{\frac{\log n}{\log\log n}}
\end{equation}
where the inequality $u \leq 2(\log\log m)/(\log\log\log m)$
can be seen as follows: Since $u$ is the least integer with $u^u\geq \log m$,
we have $(u-1)^{u-1} < \log m$ and this implies
\begin{align*}
\frac{2\log\log m}{\log\log\log m} &>
                \frac{2\log \left((u-1)^{u-1}\right)}{\log\log \left((u-1)^{u-1}\right)}\\
        &= \frac{2(u-1)\log (u-1)}{\log (u-1) + \log\log (u-1)}\\
        &= \frac{2(u-1)}{1+(\log \log (u-1))/\log (u-1)}\\
        &\geq u\text{ for $m$ large enough.}
\end{align*}
The remaining inequalities in Statement \eqref{eq:beamefich} are easy to see.
Another inequality we need for in the proof is $2u^{u-2}\leq \log m$.
Since $(u-1)^{u-1} < \log m$, it follows that
$$\frac{u^u}{\log m}\leq \frac{u^u}{(u-1)^{u-1}}\leq u\left( 1+\frac{1}{u-1}\right) ^{u-1}\leq ue$$
and multiplying each side of this inequality with $2(\log m)/u^2$ gives
\begin{equation}
\label{eq:beamefich2}
2u^{u-2}\leq \frac{2e\log m}u\leq \log m
\end{equation}
for $m$ large enough so that $u\geq 2e$.

Let $S' = \{ x\in [2^{4\log u}]\mid x\text{ is a prefix of some element of }S\}$.
Let $k = \log m/2^{4\log u}\leq u^u/u^4 = u^{u-4}$.
For every $x\in S'$, let $S_x = \{ y\in [2^k]\mid x\circ y\in S\}$.
We store $S'$ in an x-fast trie. For each $x\in S'$, we store
$S_x$ in an instance of the data structure from Lemma \ref{bfdata}
with $c = u-4$ and $b = \log m$ which is a large enough block-size
for this data structure since
$2u^{c+2} = 2u^{u-2} \leq \log m$ by Inequality \eqref{eq:beamefich2}. Also,
the size of the universe
is $2^k\leq 2^{u^{u-4}} = 2^{u^c}$ and thus, all premises
of the lemma are satisfied.
At each leaf $x$ of the x-fast trie, we also store a pointer to
the data structure storing $S_x$.

The x-fast trie uses space $O(n\log u)$. Each instance of the
data structure of Lemma \ref{bfdata} uses $O(n^2/u^2)$ blocks. There
are $u^4$ leaves in the trie, so all these instances take up space
$O(n^2u^2) = O(n(\log n)/\log\log n)$ in total. Thus, the complete
data structure has space complexity $O(n^2(\log n)/\log \log n)$,
as claimed.

We find the predecessor of $x\in [m]$ as follows: First,
we partition $x$ in a prefix $x_1$ of length $4\log u$ and
a suffix $x_2$ of length $k$. We search the x-fast trie to
find out whether $x_1\in S'$ which takes $O(\log u)$ cell-probes.
If yes, we search $S_{x_1}$ for the predecessor of $x_2$
which requires $O(u)$ cell-probes. If
we find a predecessor $x_2'$ in that set, we output $x_1\circ x_2'$.
If $x_1$ is not in $S'$ or if $S_{x_1}$ contains no
predecessor of $x_2$, we find the predecessor $x_1'$ of $x_1$ in $S'$
and return $x_1'\circ x_2'$ where $x_2'$ is the maximal element
of $S_{x_1'}$. If $x_1$ has no predecessor in $S'$, we conclude
that $x$ has no predecessor in $S$.

In total, this algorithm makes $O(u)$ cell-probes and by Inequality
\eqref{eq:beamefich}, we have
$$O(u) = O\left(\frac{\log \log m}{\log\log\log m}\right)
\leq O\left(\sqrt{\frac{\log n}{\log\log n}}\right)\text .$$
\end{proof}
\section{Classical Data Structures with Quantum Access}
\label{quant_acc}
\subsection{Lower Bounds for Set Membership}
\label{quantum_set}
Radhakrishnan, Sen and Venkatesh proved the following result from which
lower bounds for data structures for the set membership problem can
be derived:
\begin{thm}[{\cite[Theorem 1]{qu_set_mem}}]
\label{mem_exact}
Suppose there is a scheme for storing sets $S\subseteq [m]$ with
$|S|\leq n$ in $s$ bits so that membership queries can be answered
by an exact quantum algorithm that makes at most $t$ bit-probes. Then,
the following inequality must hold:
$$\sum_{i=0}^{n}{m\choose i}\leq\sum_{i=0}^{n\cdot t}{s\choose i}$$
This inequality also holds if the query algorithm is probabilistic
with one-sided error, i.e., if the algorithm always returns ``No''
when $i\not\in S$ but gives the wrong answer when $i\in S$ with
probability at most $\epsilon$ for some $\epsilon < 1$.
\end{thm}
This result improves upon a result by Buhrman et al. in
\cite{bmrv} who show that ${m\choose n}\leq {s\choose nt}2^{nt}$.
The proof of Theorem \ref{mem_exact} is based on linear algebra. Before going into the details
of that proof, let us see how this result allows us to establish
lower bounds. We either fix $t$ to some value and see how large
$s$ must be to satisfy the inequality or vice versa.
\begin{cor}
If the query algorithm of a data structure for the set membership problem
only makes one bit-probe, the data structure must use space $s\geq m$.
Thus, the bit vector data structure described at the beginning of
\ref{mem_class} is optimal even in the setting of exact quantum
computation.
\end{cor}
\begin{proof}
When we set $t = 1$, then we must have $s\geq m$ in order to satisfy
the inequality in Theorem \ref{mem_exact}.
\end{proof}
\begin{cor}
Suppose that $n\leq m^d$ for some constant $d$ with $0<d<1$.
In a data structure for the set membership problem that uses
$O(n\log m)$ bits for storage, the query algorithm must make
$\Omega (\log m)$ bit-probes. It follows that the Perfect Hashing
scheme is asymptotically optimal even in the quantum bit-probe
model with one-sided error.
\end{cor}
\begin{proof}
Suppose that $s = O(n\log m)$ and let $c$ be such that
$s\leq c\cdot n\log m$.
We then have
$$
\sum_{i=0}^n {m\choose i}\leq \sum_{i = 0}^{nt} {s\choose i} \leq \sum_{i = 0}^{nt}{cn\log m\choose i}
\leq \left( \frac{ecn\log m}{nt}\right) ^{nt}
= \left( \frac{ec\log m}t\right) ^{nt}
$$
Taking logarithms on both sides of the inequality and using Lemma
\ref{optimal}, we can conclude that
$$n(1-d)\log m \leq nt\log\left( \frac{ec\log m}{t}\right)\text .$$
Let $b$ be such that $t = (\log m)/b$. If we can prove that $b = O(1)$,
it follows that $t = \Omega (\log m)$. From the inequality above,
it follows that
\begin{align*}
&n(1-d)\log m \leq \frac n b\log m \log \left( \frac{ebc\log m}{\log m}\right)\\
\Rightarrow\ &1-d\leq \frac 1 b(\log (ec) + \log b)\\
\Rightarrow\ &b\leq\frac{1}{1-d}(\log (ec) + \log b) = O(\log b)
\end{align*}
which shows that $b = O(1)$ because $O(\log b)$ cannot
grow faster than $b$.
It follows that $t = \Omega (\log m)$.
\end{proof}
Let us now prove Theorem \ref{mem_exact}.
\begin{proof}
Let $s$ be the number of bits required for the data structure and
let $t$ be the maximal number of cell-probes by the query algorithm.
Let $U_0,U_1,\dots ,U_t$ be the unitary transforms that form the
quantum algorithm and $\phi$ the function that encodes sets $S$ into bit strings.
For every set $S\subseteq [m]$ with $|S|\leq n$, we define
$$W_S = U_tO_{\phi (S)}U_{t-1}\dots U_1O_{\phi (S)}U_0$$
where $O_{\phi (S)}$ is a quantum bit-probe oracle of the $\pm$-type.

We prove the result for exact algorithms ($\epsilon = 0$) and then indicate
how it follows for probabilistic algorithms with one-sided error.
We first show that the unitary transforms in
$\mathcal W = \{ W_S^{\otimes n}\}_{S\subseteq [m], |S|\leq n}$
are linearly independent. Then, we prove that they are contained in a vector space of dimension at most
$\sum _{i=0}^{nt}{s\choose i}$. Since $\{ W_S^{\otimes n}\}_{S\subseteq [m], |S|\leq n}$
has $\sum _{i = 0}^n{m\choose i}$ elements, this proves the inequality of
Theorem \ref{mem_exact},
$$\sum _{i = 0}^n{m\choose i}\leq\sum _{i=0}^{nt}{s\choose i}\text .$$

Let $\mathcal H$ be the Hilbert space that the algorithm operates on.
Let $d$ denote its dimension.
We let $A_1$ denote the subspace of $\mathcal H$ that contains
those states that result with probability 1
in answer ``Yes'' when measured and
we let $A_0$ denote the space where the answer is ``No''.
The subspaces $A_0$ and $A_1$ are orthogonal.

Suppose
that there is a nontrivial linear combination
$$\sum _{S\subseteq [m], |S|\leq n} \alpha _SW_S^{\otimes n} = 0\text .$$
Let $T = \{ i_1,\dots ,i_k\}$
be a maximal set such that $\alpha _T\neq 0$. Define
$$\ket{\psi _T} = \ket{i_1}^{\otimes n-k+1}\otimes \ket{i_2}\otimes \dots  \otimes\ket{i_k}$$
where $\ket{i_j}$ is the starting state encoding the query ``$i_j\in S?$''.
For any $S$, we have $W_S^{\otimes n}\ket{\psi _T} = W_S^{\otimes n-k+1}\ket{i_1}^{\otimes n-k+1}\otimes W_S\ket{i_2}\otimes \dots  \otimes W_S\ket{i_k}$.
Since $i_j\in T$ for every $j$, $W_T\ket{i_j}\in A_1$ and thus,
$W_T^{\otimes n}\ket{\psi _T} \in A_1^{\otimes n}$. For $S\neq T$ with $\alpha _S \neq 0$
on the other hand, there is some $j$ such that $i_j\not\in S$
since $T$ is maximal. Thus $W_S\ket{i_j}\in A_0$ and it follows that
$W_S^{\otimes n}\ket{\psi _T}$ is orthogonal to $A_1^{\otimes n}$. If
we let $P_1$ be the projection on $A_1^{\otimes n}$, we have
$$
\sum _{S\subseteq [m], |S|\leq n} \alpha _SW_S^{\otimes n} = 0
\Rightarrow P_1\left( \sum _{S\subseteq [m], |S|\leq n} \alpha _SW_S^{\otimes n}\ket{\psi _T}\right) = 0
$$
and since $W_S^{\otimes n}\ket{\psi _T}$ for $S\neq T$ is orthogonal to
$A_1^{\otimes n}$ and $W_T^{\otimes n}\ket{\psi _T}$ is contained in
$A_1^{\otimes n}$, the result of the projection is
$$
P_1\left( \sum _{S\subseteq [m], |S|\leq n} \alpha _SW_S^{\otimes n}\ket{\psi _T}\right) = \alpha _TW_T^{\otimes n}\ket{\phi _T}\text .
$$
It follows that $\alpha _T W_T^{\otimes n}\ket{\psi _T} = 0$, but this
contradicts our assumption that $\alpha _T\neq 0$. Therefore, the
elements of $\mathcal W$ are linearly independent.

We now go on to show that they are contained in a vector space
of the correct size.
For a set $T = \{t_1,\dots ,t_k\}\subseteq [s]$, let $[\phi (S)]_T$ be the
parity of the bits at locations $t_1,\dots ,t_k$ in $\phi (S)$.
Since $O_{\phi (S)}$ is diagonal, we have
$$
(W_S)_{i,j} = \sum_{k_0,\dots ,k_{t-1}\in [d]} (U_t)_{i,k_{t-1}}(O_{\phi (S)})_{k_{t-1},k_{t-1}}(U_{t-1})_{k_{t-1},k_{t-2}}(O_{\phi (S)})_{k_{t-2},k_{t-2}}\dots (U_1)_{k_1,k_0}(O_{\phi (S)})_{k_0,k_0}(U_0)_{k_0,j}\text .
$$
We can rewrite this as follows, for appropriate sets $l_{k_i}\subseteq [n]$
where each $l_{k_i}$ is either a singleton or empty:
$$
(W_S)_{i,j} = \sum_{k_0,\dots ,k_{t-1}\in [d]} (U_t)_{i,k_{t-1}}(-1)^{[\phi (S)]_{l_{k_{t-1}}}}(U_{t-1})_{k_{t-1},k_{t-2}}(-1)^{[\phi (S)]_{l_{k_{t-2}}}}\dots (U_1)_{k_1,k_0}(-1)^{[\phi (S)]_{l_{k_0}}}(U_0)_{k_0,j}\text .
$$
We can simplify this term by introducing some additional notation:
For $\Delta$ the
symmetric difference between sets,
we define $T_{k_0,\dots ,k_{t-1}} = l_{k_0}\Delta \dots  \Delta l_{k_{t-1}}$.
This gives us
\begin{align*}
(W_S)_{i,j} &= \sum_{k_0,\dots ,k_{t-1}} (-1)^{[\phi (S)]_{T_{k_0,\dots ,k_{t-1}}}} (U_t)_{i,k_{t-1}}(U_{t-1})_{k_{t-1},k_{t-2}}\dots (U_1)_{k_1,k_0}(U_0)_{k_0,j}\\
&= \sum _{T\subseteq [s], |T|\leq t} (-1)^{[\phi (S)]_{T}}\sum _{k_0,\dots , k_{t-1} \text{ such that } T = T_{k_0,\dots ,k_{t-1}}}(U_t)_{i,k_{t-1}}(U_{t-1})_{k_{t-1},k_{t-2}}\dots (U_1)_{k_1,k_0}(U_0)_{k_0,j}\\
&= \sum _{T\subseteq [s], |T|\leq t} (-1)^{[\phi (S)]_{T}}(M_T)_{i,j}
\end{align*}
where the $M_T$ are unitary transforms that only depend on
$U_0,\dots ,U_t$ and $T$, not on $S$. Thus, we have
\begin{align*}
W_S^{\otimes n} &= \sum_{1\leq i\leq n, T_i\in [s], |T_i|\leq t} (-1)^{[\phi (S)]_{T_1}}\dots (-1)^{[\phi (S)]_{T_n}} (M_{T_1}\otimes \dots  \otimes M_{T_n})\\
&= \sum _{T\subseteq [s],|T|\leq nt} (-1)^{[\phi (S)] _T} N_T
\end{align*}
where $N_T$ is independent of $S$, namely,
$$N_T = \sum _{T_1,\dots , T_n\text{ such that }T_1\Delta \dots  \Delta T_n = T} (M_{T_1}\otimes \dots  \otimes M_{T_n})\text .$$
This shows that every element of $\mathcal W$ is contained in the
subspace spanned by $\{ N_T \} _{T\subseteq [s], |T|\leq nt}$ which
has dimension at most
$$\sum _{i=0}^{nt} {s\choose i}$$
as claimed. Thus, we proved that all $\sum _{i=0}^n{m\choose i}$
elements of $\mathcal W$ are linearly independent vectors of a
subspace with dimension at most $\sum _{i=0}^{nt}{s\choose i}$.
This proves the inequality.

To see that the result also holds if we allow
one-sided error, note that the only part where
we made use of the correctness of the algorithm is the proof
that the vectors in $\mathcal W$ are linearly independent. We
used that $W_S\ket{i}\in A_0$ for $i\not\in S$ so that
$W_S^{\otimes n}\ket{\psi _T}$ is orthogonal to $A_1^{\otimes n}$
for $S\neq T$ with $\alpha _S>0$ and we used that $W_T^{\otimes n}\ket{\psi _T}$
has a non-zero projection on $A_1^{\otimes n}$. This is also achieved
by an algorithm that correctly answers membership queries for $i\not\in S$
but has some probability $\epsilon <1$ for making an error on $i\in S$.
\end{proof}
In the same paper, Radhakrishnan, Sen and Venkatesh also proved a lower bound for the space complexity of
data structures for set membership that support queries which make $t$
bit-probes and have \em two-sided \em error probability $\epsilon$ with
$m/n < \epsilon < 2^{-3t}$.
We will not give the proof of this result here, the main ideas are similar to that of Theorem \ref{mem_exact}.
\begin{thm}[{\cite[Theorem 4]{qu_set_mem}}]
Let $t\geq 1$ and let $n/m< \epsilon < 2^{-3t}$. Suppose there is a
quantum-access data structure with a space complexity of $s$ bits and
a time complexity of $t$ bit-probes for set membership with two-sided error at most
$\epsilon$. Then, we have
$$s = \Omega\left( \frac{nt\log (m/n)}{\epsilon^{1/(6t)}\log (1/\epsilon)}\right)\text .$$
\end{thm}

\subsection{Lower Bounds for Predecessor Search}
\label{pred_bound}
Beame and Fich gave an asymptotic lower bound on the time complexity of predecessor search data structures given an upper bound of $O(n^2\log n/\log\log n)$
$(\log m)$-blocks for the space complexity
in the classical deterministic setting that matches the time complexity of
their data structure. However, in \cite{qu_cell_probe} Pranab Sen and Srinivasan Venkatesh
obtained a simpler proof in a computational model that they call
the \em address-only quantum cell-probe \em model.
Algorithms in this model are quantum cell-probe algorithms but they
do not use the full power of the model of quantum computation.
The address-only model still encompasses classical deterministic and
probabilistic computation and also some famous quantum algorithms,
such as Grover's search algorithm. Therefore, their lower bound shows that
the data structure by Beame and Fich is asymptotically optimal
also in the setting of classical probabilistic computation.
The restriction of address-only algorithms is that we may use quantum
parallelism only over the address lines. This is explained below in more detail.
\begin{defin}
A quantum cell-probe algorithm $U_0,\dots ,U_t$ has the \em address-only \em
property if, before each query to the oracle, the state of the
qubits can be written as a tensor product of two quantum states
where one state consists of the data qubits
and the other consists
of the address- and workspace-qubits. That is, the data qubits may
not be entangled with the rest. Furthermore, the state of the
data qubits may only depend on the stage of the algorithm, but not on the input.

In order to describe this in more precise terms, let $\mathcal H$
be the Hilbert space that the algorithm operates on. We can write
$\mathcal H = \mathcal H_L\otimes \mathcal H_Z\otimes\mathcal H_B$ where
$\mathcal H_L$ consists of the address qubits, $\mathcal H_Z$ consists
of the workspace qubits and $\mathcal H_B$ of the data qubits. We
require that there are $\ket{b_0}_B,\dots ,\ket{b_{t-1}}_B\in\mathcal H_B$ such
that for every possible data $d$, every query $q$ and every integer $i$ with
$0\leq i<t$, we can write $U_iO_dU_{i-1}O_d\dots O_dU_0\ket{q}$ as
$\ket{\phi}\otimes \ket{b_i}_B$ for some
$\ket{\phi}\in\mathcal H_L\otimes\mathcal H_Z$.
\end{defin}
The proof of the lower bound uses a relation between data
structures with quantum access and quantum communication protocols.
In the quantum communication model, there are two parties, called
Alice and Bob, who want to compute a function $f:A\times B\to C$
where $A,B,C$ are finite sets. At the beginning, Alice holds
an encoding $\ket a$ of some $a\in A$ in the computational basis
and some workspace qubits initialized to $\ket 0$. Bob holds an encoding
$\ket b$ of $b\in B$ in the computational basis and also some workspace qubits
in state $\ket 0$.
One of the two parties begins and they take turns alternately. At each
turn, they may apply a unitary transform to the qubits that
they hold (but not to qubits held by the other party) and then give
some of their qubits to the other party. A communication protocol $\mathcal P$
for the function $f$ specifies transforms that are applied
and which qubits are sent by Alice and Bob so that they end up
with a state from which one party can learn $f(a,b)$ after performing
some measurement. We can consider exact protocols or protocols that
have some error probability.

Often, the study of communication complexity is just concerned
with the total amount of communication that occurs. For our purposes,
we need to look at a more fine-grained picture. To this end, we define
\em secure \em and \em safe \em protocols.
\begin{defin}
A protocol $\mathcal P$ is called \em secure \em if the input qubits are never
measured and never sent as messages. (Since they are in the computational
basis, it is possible for Alice and Bob to make copies of their respective
inputs, so every protocol can be made secure without increasing the
communication.)

A protocol is called $[t,c,l_1,\dots ,l_t]^A$-\em safe \em
($[t,c,l_1,\dots ,l_t]^B$-\em safe\em ) if it is a secure protocol where
Alice (Bob) starts and which has exactly $t$ rounds of communication
such that
\begin{itemize}
\item the
first message by Alice (Bob) consists of two parts: The first part has length
$c$ and its density matrix must be independent of the input. It is called the
\em safe overhead\em . The second part contains the message proper. It has
length $l_1$ and its density matrix is allowed to depend on the input. Thus,
the total length of the message is $l_1+c$.
\item for $1<i\leq t$, the $i$th message in the protocol has length
at most $l_i$.
\end{itemize}
We say that a protocol $(t,c,l_a,l_b)^A$-safe ($(t,c,l_a,l_b)^B$-safe)
if and only if it is $[t,c,l_1,\dots ,l_t]^A$-safe ($[t,c,l_1,\dots ,l_t]^B$-safe)
where $l_i = l_a$ for odd (even) $i$ and $l_i = l_b$ for even (odd) $i$.
\end{defin}
The safe overhead may be used, for example, to share EPR-pairs. We
also need to define \em public coin\em -protocols.
\begin{defin}
In a \em public-coin \em communication protocol, we have at the beginning, in
addition to the inputs and workspace qubits, another quantum state
of the form $\sum _c \sqrt{p_c}\ket c _A\ket c _B$ where the subscripts
$A$ and $B$ denote ownership by Alice and Bob respectively and the $p_c$
are positive real numbers. That state is called the \em public coin \em and it
is never measured and never sent as a message. However, Alice and Bob can make
a copy of their half using a CNOT-transform (possibly entangled with the
original coin).
\end{defin}
Here, since the $p_c$ are all positive,
the quantum state $\sum _c\sqrt{p_c}\ket c$ behaves like a
classical random variable $C$ that takes on value $c$ with probability
$p_c$ when measured. Alternatively, we can view a public-coin
protocol as a probability distribution over coinless protocols, and
safe public-coin protocols as distributions over coinless safe protocols.
The following lemma describes how data structures and quantum communication
complexity relate.
\begin{lem}
\label{data_to_com}
Let $f:D\times Q\to A$ be a static data structure problem. Suppose
that there is a data structure for this problem that has block-size
$w$, requires $s$ blocks of space and has a quantum cell probe algorithm
that answers queries with success probability $p$ making $t$ probes. Then there
is a $(2t,0,w+\log s,w+\log s)^A$-safe coinless protocol that solves
the communication problem where Bob is given $d\in D$, Alice is
given $q\in Q$ and they want to compute $f(d,q)$ with success probability $p$.
Note that, while it is usual in communication complexity that Alice has the
\em first \em input for the function that has to be computed, here, she has the
\em second \em input.
If the query algorithm for the data structure is address-only,
we have a $(2t,0,\log s,w+\log s)^A$-safe coinless protocol for this
problem.
\end{lem}
\begin{proof}
The communication protocol simply simulates the query algorithm
of the data structure. Instead of applying the oracle transform,
Alice sends the address- and data-qubits to Bob who can perform
the oracle transform since he knows $d$. Bob then sends these
qubits back to Alice who continues with the algorithm. In the
case that the algorithm is address-only, Alice does not need
to send the data-qubits. Since they are not entangled with any other
qubits and since their state is not affected by Alice's input, Bob
can prepare the appropriate data-qubits by himself.
\end{proof}
This lemma can help us prove lower bounds on $t$ for a data structure
using space $s$ for a given problem. If we can prove a lower bound
on the communication required, we also know a lower bound on $t$.
In many cases, we have $\log (s) = O(w)$. Recall, for example,
the Perfect Hashing method. There, we have $s = O(n)$ while
$w = \log m$. Since $n\leq m$, we have $\log s = O(\log n) = O(\log m)$.
In such cases, if there is an address-only query algorithm, we have
a $(2t,0,\log s,O(w))^A$-safe protocol for the communication problem.
Thus, if $s$ is small compared to $w$, Alice's messages are significantly shorter
than Bob's. We will use this asymmetry to prove the lower bounds.

We now introduce some notions from quantum information theory. More
about this subject can be found in \cite[Part III]{mikeike}.
\begin{defin}
Let $\rho$ be the density matrix of some quantum system $A$.
The \em von Neumann entropy \em of $A$ is $S(A) = S(\rho ) = -\text{Tr}(\rho \log\rho)$.
The \em mutual information \em of two disjoint quantum systems $A$ and $B$
is $I(A:B) = S(A) + S(B) - S(AB)$.
\end{defin}
We now show some properties of the von Neumann entropy function
and of mutual information.
\begin{lem}
\label{nmann}
The von Neumann entropy and mutual information have the following
properties:
\begin{enumerate}
\item The von Neumann entropy function is subadditive, i.e.,
for all quantum systems $A$ and $B$, we have
$S(AB)\leq S(A) + S(B)$. It follows that $I(A:B)$ non-negative for all quantum
systems $A$ and $B$.
\item $|S(A)-S(B)|\leq S(AB)$.
\item For disjoint quantum systems $A,B,C$, we have
$I(A:BC) = I(A:B) + I(AB:C) - I(B:C)$.
\item $0\leq I(A:B)\leq 2S(A)$.
\item If the Hilbert space of $A$ has dimension $d$, then
$S(A)\leq \log d$.
\end{enumerate}
\end{lem}
\begin{proof}
Properties 1 and 2 are proved in \cite[Section 11.3]{mikeike}. Let us prove the
remaining properties here. Property 3 holds because
\begin{align*}
I(A:BC) = S(A) + S(BC) - S(ABC) &= I(A:B) - S(B) + S(AB) + S(BC) - S(ABC)\\
&= I(A:B) + (S(AB) + S(C) - S(ABC)) - S(B) - S(C) + S(BC)\\
&= I(A:B) + I(AB:C) - (S(B) + S(C) - S(BC))\\
&= I(A:B) + I(AB:C) - I(B:C)
\end{align*}

Property 4 holds because
$$0 = S(A) + S(B) - S(A) - S(B)\leq S(A)+S(B) - S(AB) = I(A:B)$$
and
$$I(A:B) = S(A) + S(B) - S(AB)\leq S(A) + S(B) - |S(A) - S(B)|\leq 2S(A)
\text .$$

To prove that property 5 holds, we show by induction in $d$ that the term
$\sum _{i = 0}^d\alpha _i\log \alpha _i$
with $\sum _i\alpha _i = t$ and $0\leq\alpha _i\leq 1$ is minimized when
all $\alpha _i$ are equal. The base case $d = 1$ is trivial. Now suppose
that our claim is true for $d-1$. We show that it also holds for $d$. We write
$$\sum _{i = 1}^d\alpha _i\log\alpha _i = \alpha _d\log\alpha _d
+ \sum _{i = 1}^{d-1}\alpha _i\log\alpha _i$$
and find the minimum of this expression under the condition that
$\alpha _1,\dots , \alpha _d\geq 0$ and $\sum _i^d\alpha _i = t$ in two steps.
First, treating $\alpha _d$ as a variable, we minimize the term
$\sum _{i=1}^{d-1}\alpha _i\log\alpha _i$ under the condition that
$\sum _{i=1}^{d-1}\alpha _i = t-\alpha _d$. Then, we find the value for $\alpha _d$
that minimizes the whole term.

The solution for the first step is given by the induction hypothesis:
$\sum _{i = 1}^{d-1}\alpha _i\log\alpha _i$ is minimized when
$\alpha _1,\dots , \alpha _{d-1}$ are equal, i.e.,
$$\alpha _1,\dots , \alpha _{d-1} = \frac{t - \alpha _d}{d-1}
\text .$$
Using these values for the $\alpha _i$, we have
$$\sum _{i = 1}^d\alpha _i\log\alpha _i = \alpha _d\log\alpha _d
+ (t-\alpha _d)\log\left( \frac{t-\alpha _d}{d-1}\right)\text .$$
Viewing the expression above as a function in $\alpha _d$, we can easily
show that it achieves its global minimum at $\alpha _d = t/d$. For this
value of $\alpha _d$, we get
$$\alpha _1,\dots , \alpha _{d-1} = \frac{t - t/d}{d-1} =
\frac{(d-1)t}{(d-1)d} = \frac t d$$
and so the minimum is achieved when all $\alpha _i$ are equal as we claimed.

Now consider a quantum
system $A$ with density matrix $\rho _A$. We can choose a basis
$\ket{a_1},\dots , \ket{a_n}$ of quantum states such that
$\rho _A$ is a diagonal matrix with respect to that basis, i.e.,
$\rho _A = \sum _{i=1}^d\alpha _i\ket{a_i}\bra{a_i}$ for some
$\alpha _i \geq 0$ with $\sum _{i=1}^d\alpha _i = 1$. We then have
$S(A) = -\sum _{i = 1}^d\alpha _i\cdot\log\alpha_i$ which is \em maximized \em
when $\sum _{i=1}^d\alpha _i\log\alpha _i$ is minimized, i.e., when
$\alpha _1,\dots , \alpha _d = 1/d$. Thus, we have
$$S(A)\leq -\sum _{i=1}^d\frac 1 d\log \left(\frac 1 d\right) = \log d\text .$$
\end{proof}
We can encode classical random variables as quantum systems. Let
$\mathcal H$ and $\mathcal K$ be  disjoint finite-dimensional Hilbert spaces,
$X$ a system in $\mathcal H$ and $Q$ a system in $\mathcal K$.
Suppose that the density matrix of the joint system $XQ$ has a diagonal
representation $\sum _x p_x\ket{x}\bra{x}\otimes\sigma _{x}$
where $p_x > 0$, $\sum _x p_x = 1$, the $\ket x$ are orthonormal vectors
in $\mathcal H$ and
the $\sigma _{x}$ are density matrices in $\mathcal K$. We say that $X$ is
a classical random variable and that $Q$
is a quantum encoding of $X$. The reduced density matrix of $Q$ is
$$\sigma = \text{Tr}_X\left( \sum _x p_x\ket{x}\bra{x}\otimes\sigma _{x}\right)
=\sum _xp_x\sigma _x$$
so if we consider $Q$ on its own, we can describe it as having density
matrix $\sigma _x$ with probability $p_x$.
We have $S(XQ) = S(X) + \sum _x p_xS(\sigma _x)$ and
$I(X:Q) = S(X) + S(Q) - S(X) - \sum _xp_xS(\sigma _x)
= S(Q) - \sum _x p_xS(\sigma _x)$.

Now, consider two classical random variables $X$ and $Y$ and a
quantum encoding $Q$ of the joint random variable $XY$. That is,
we can write the density matrix of $XYQ$ as $\sum _{x,y}p_{x,y}\ket x
\ket y\bra y\bra x \otimes \sigma _{x,y}$ for density matrices
$\sigma _{x,y}$ in the Hilbert space of $Q$ and $p_{x,y} \geq 0$
with $\sum _{x,y} p_{x,y} = 1$. Let $q_y^x = \text{Pr}(Y = y|X = x)$.
We let $Q^x = \sum _y q_y^x\sigma _{x,y}$. The conditional mutual
information between $Y$ and $Q$ is defined as
$I((Y:Q)|X=x) = I(Y:Q^x)$.

Some properties of random variable encodings that we need are
described in the following propositions.
\begin{prop}
\label{multi_vars}
Suppose $M$ is a quantum encoding of a classical random variable
$X = X_1\dots X_n$ where the $X_i$ are independent classical random
variables. Then, $I(M:X_1\dots X_n) = \sum _i I(X_i:MX_1\dots X_{i-1})$.
\end{prop}
\begin{proof}
We prove this by induction in $n$. For $n = 1$, there is nothing
to prove since $I(M:X_1) = I(X_1:M)$. Let $n > 0$ and suppose the statement
holds for $n-1$. We show that it then holds for $n$. By part 3 of
Lemma \ref{nmann}, we have
$$I(M:X_1\dots X_n) = I(M:X_1) + I(MX_1:X_2\dots X_n) - I(X_1:X_2\dots X_n)$$
and since the $X_i$ are independent, we have $I(X_1:X_2\dots X_n) = 0$.
Applying the induction hypothesis to $I(MX_1:X_2\dots X_n)$, we
can conclude that $I(MX_1:X_2\dots X_n) = \sum _{i=2}^nI(X_i:MX_1\dots X_{i-1})$.
Hence, $I(M:X_1\dots X_n) = \sum _{i=1}^n I(X_i:MX_1\dots X_i)$.
\end{proof}
\begin{prop}
\label{expect}
Let $X,Y$ be classical random variables and $M$ an encoding of
$(X, Y)$. Then $I(Y:MX) = I(X:Y) + E_X\left[ I((Y:M)|X=x)\right]$.
\end{prop}
\begin{proof}
We have $I(Y:MX) = S(Y) + S(MX) - S(MXY)$ and
$S(MX) = S(X) + \sum _xp_xS(Q^x)$. Also,
$$S(MXY) = S(XY) + \sum _{x,y}p_{x,y}S(\sigma _{x,y}) = S(XY) + \sum _{x,y}p_xq_y^xS(\sigma _{x,y})$$
and this gives us
\begin{align*}
I(Y:MX) &= S(X) + S(Y) - S(XY) + \sum _xp_xS(Q^x) - \sum _{x,y}p_xq_y^xS(\sigma _{x,y})\\
&= I(X:Y) + \sum _xp_x\left( S(Q^x) - \sum _yq_y^xS(\sigma _{x,y})\right)
\end{align*}
Furthermore, we have $\sum _y q_y^x S(\sigma _{x,y})
= S(YQ^x) - S(Y)$. Thus,
$$S(Q^x) - \sum _yq_y^xS(\sigma _{x,y})
= S(Q^x) + S(Y) - S(YQ^x) = I(Y:Q^x)$$
and it follows that
$$I(Y:MX) = I(X:Y) + \sum _x p_xI(Y:Q^x) = I(X:Y) + E_X\left[ I((Y:Q)|X=x)\right]$$
as claimed.
\end{proof}
We now prove a proposition that gives an upper bound for the
mutual information of the first message in a safe quantum protocol
and the input that
does not depend on the size of the safe overhead.
\begin{prop}
\label{safe_bound}
Let $M_1$ and $M_2$ be finite-dimensional disjoint quantum systems
and $M=M_1M_2$ an encoding of a classical random variable $X$. Suppose
that the density matrix of $M_2$ is independent of the value $x$
of $X$, i.e., $\text{Tr}_{M_1}(\sigma _x) = \text{Tr}_{M_1}(\sigma _y)$ for
all $x$ and $y$ in the range of $X$. If $M_1$ is supported on $a$
qubits, we have $I(X:M)\leq 2a$.
\end{prop}
\begin{proof}
Let $\sigma$ be such that $\text{Tr}_{M_1}(\sigma _x) = \sigma$ for
all $x$.
First, we show that $X$ and $M_2$ have no mutual information.
\begin{align*}
I(X:M_2) &= S(X) + S(M_2) - S(XM_2)\\
&= S(X) + S\left( \text{Tr}_{M_1}\left( \sum _xp_x\sigma _x\right) \right) - S\left( \sum _x p_x\ket x\bra x\otimes \text{Tr}_{M_1}(\sigma _x)\right)\\
&= S(X) + S\left( \sum _x p_x\text{Tr}_{M_1}(\sigma _x) \right) - S\left( \left( \sum _x p_x\ket x\bra x\right)\otimes \sigma\right)\\
&= S(X) + S(\sigma ) - (S(X) + S(\sigma ))\\
&= 0
\end{align*}
By parts 3 and 4 of Lemma \ref{nmann},
\begin{align*}
I(X:M) = I(X:M_1M_2) = I(X:M_2M_1) &= \overbrace {I(X:M_2)}^{=0} + I(XM_2:M_1)
- \overbrace{I(M_2:M_1)}^{\geq 0}\\
&\leq I(XM_2:M_1)\\
&\leq 2S(M_1)\\
&\leq 2a
\end{align*}
\end{proof}
We now prove the round elimination lemma which is important for
proving the lower bound. This lemma only applies to a certain
kind of communication problem.
\begin{defin}
Let $f:A\times B\to C$ be a communication problem. For any natural
number $n$, let $f^{(n)}$ be the communication problem where Alice
receives $a_1,\dots ,a_n\in A$, Bob receives $i\in [n], a_1,\dots ,a_{i-1}$ and
some $b\in B$. The goal is to compute $f(a_i,b)$.

A similar problem, which we will need later on, is $^{(n)}f$. Here,
Alice is given $a\in A$ and $i\in \{ 1,\dots ,n \}$ and Bob is
given $b_1,\dots ,b_n\in B$. The goal is to compute
$f(a,b_i)$.
\end{defin}
Consider a protocol $\mathcal P$ for problem $f^{(n)}$ where
Alice sends the first message.
Intuitively,
it seems unlikely that the first message contains a lot of
useful information for Bob, unless Alice sends her whole input,
since Alice does not know $i$. The round elimination lemma justifies
that intuition. We can transform the protocol $\mathcal P$ to a
protocol for $f$ that uses one less round of communication. In that protocol,
Bob sends the first message. The price we have to pay for the round
elimination is an increased length of Bob's first message and a
slight increase in error probability. The increased length of the
first message is, however, limited to a safe overhead. To prove
this result, we need two lemmas that we state without proof.
The first one is a version of Yao's minimax lemma in \cite{yao}.
A proof of the second one can be found in \cite[Appendix B]{qu_cell_probe}.
\begin{lem}[Yao's Minimax Lemma]
\label{Yao}
Fix some communication problem $f:A\times B\to C$.
For every $[t,c,l_1,\dots ,l_t]^A$-safe quantum communication protocol
$\mathcal P$ for computing $f$ and every probability distribution $D$ on
$A\times B$, let $\epsilon _D^{\mathcal P}$ denote the probability
that $\mathcal P$ for inputs $a,b$ sampled according to $D$ does \em not \em
yield the result $f(a,b)$. Let $\epsilon ^{\mathcal P}$ denote the worst
case probability that $\mathcal P$ does not result in $f(a,b)$.
We have
$$\inf _{\mathcal P:\text{public coin}}\epsilon ^{\mathcal P}
= \sup _D \inf _{\mathcal P:\text{coinless}}\epsilon _D^{\mathcal P}
= \sup _D \inf _{\mathcal P:\text{public coin}}\epsilon _D^{\mathcal P}$$
\end{lem}
\begin{lem}
\label{for_elim}
Suppose $f:A\times B\to C$ is a communication problem. Let
$D$ be a probability distribution on the input set $A\times B$.
Let $\mathcal P$ be a $[t,c,l_1,\dots ,l_t]^A$-safe coinless quantum
protocol for this problem.
Let $X$ and $M$ be classical random
variables that denote Alice's
input and Alice's first message under distribution $D$.
Let $\epsilon _D^\mathcal P$ be the probability that the protocol
makes an error on an input sampled according to $D$.

There is a $[t-1,c+l_1,l_2,\dots ,l_t]^B$-safe coinless protocol
$\mathcal Q$ such that
$$\epsilon _D^{\mathcal Q}\leq\epsilon _D^{\mathcal P} +
((2\ln 2)I(X:M))^{1/4}$$
\end{lem}
This lemma shows that we can reduce the number of rounds by increasing the
safe overhead at the price of increasing the error probability by an amount that
depends on the mutual information between Alice's input and her first message.
The protocol $\mathcal Q$ is constructed in stages. The first stage is to make
Alice's first message independent of her input by replacing it with a message
that ``averages'' over all possible inputs. In the second stage, Alice does
not send the average message but Bob generates it himself which is possible
since it is independent of Alice's input. Then, they resume as in the protocol
$\mathcal P$. But to achieve the correct entanglement between Alice's and Bob's
state, Bob's first message must contain a safe overhead of $c+l_1$ qubits.

If Alice's input and her first message have little mutual information, we
can drop the first message with only a small increase in error probability.
Let us now prove the Round Elimination Lemma.
\begin{lem}[Quantum Round Elimination]
\label{qu_round_elim}
Let $f:A\times B\to C$ be a communication problem. Suppose
we have a $[t,c,l_1,\dots ,l_t]^A$-safe public coin quantum protocol
for $f^{(n)}$ with worst case error $<\delta$. Then, there also
exists a $[t-1,c+l_1,l_2,\dots ,l_t]^B$-safe public coin quantum protocol
that solves $f$ with worst case error probability less than
$\epsilon = \delta + (4l_1(\ln 2)/n)^{1/4}$.
\end{lem}
\begin{proof}
Suppose the protocol $\mathcal P$ has worst-case error $\delta ' < \delta$.
Let $\epsilon ' = \delta ' + (4l_1(\ln 2)/n)^{1/4}$. By Lemma
\ref{Yao}, it suffices to give for each distribution $D$ on $A\times B$
a protocol $\mathcal P_D$ that solves $f$ for inputs sampled according to $D$
with error probability $\epsilon _D^{\mathcal P_D}\leq\epsilon ' < \epsilon$.
Let $D$ be an arbitrary probability distribution on the input set $A\times B$.
Let $D^*$ be the distribution on $A^n\times \{ 1,\dots , n\}\times B$
that is sampled by first sampling $i$ from $\{ 1,\dots ,n\}$ uniformly at random,
sampling for every $j\in \{ 1,\dots ,n\}$ a pair $(a_j,b_j)$ according to $D$
and returning $(a_1,\dots ,a_n, i, b_i)$. We have,
by Lemma \ref{Yao} and the fact that $\mathcal P$ has worst-case error
$\delta'$,
a $[t,c,l_1,\dots ,l_t]^A$-safe protocol
$\mathcal P^*$ for $f^{(n)}$ that has error probability
$\epsilon _{D^*}^{\mathcal P^*}\leq \delta '$. Let $M$ be the random
variable for Alice's first message in the protocol and $X$ the
random variable for her input. Her first message consists of
a main part $M_1$ of $l_1$ qubits and a safe overhead $M_2$ of $c$ qubits
whose density matrix is independent of $X$.

Let $X_j$ be the random variable for the $j$th input for Alice
under distribution $D^*$. Then, $X_1,\dots ,X_n$ are independent
and $X=X_1\dots X_n$. Let $Y$ be the random variable for Bob's input
from the set $B$.
(This random variable is the same under $D$ and $D^*$.)
By Proposition
\ref{multi_vars} and \ref{safe_bound},
\begin{align*}
2l_1\geq I(X:M) = I(M:X_1\dots X_n) &= \sum _iI(X_i:MX_1\dots X_{i-1})\\
&= n\cdot \left( \sum _i\frac 1 nI(X_i:MX_1\dots X_{i-1})\right)\\
&= n\cdot \mathbb E_i\left[ I(X_i:MX_1\dots X_{i-1})\right]
\end{align*}
and by Proposition \ref{expect}
\begin{align*}
I(X_i:MX_1\dots X_{i-1}) &= I(X_i:X_1\dots X_{i-1}) + \mathbb E_{X_1\dots X_{i-1}}\left[
I((X_i:M)|X_1 = x_1,\dots ,X_{i-1} = x_{i-1})\right]\\
&= \mathbb E_{X_1\dots X_{i-1}}\left[ I((X_i:M)|X_1 = x_1,\dots ,X_{i-1} = x_{i-1})\right]
\end{align*}
where the last equality holds because the $X_j$ are independent. This
gives us
\begin{equation}
\label{bound}
\frac{2l_1}n \geq \mathbb E_{i,X}\left[ I((X_i:M)|X_1=x_1,\dots ,X_{i-1} = x_{i-1})\right]
\end{equation}
We define $D^*_{i;x_1,\dots ,x_{i-1}}$ as the conditional distribution obtained from $D^*$ by
fixing the element from $\{ 1,\dots ,n\}$ to $i$ and for all $j<i$,
fixing $X_j$ to $x_j$. We have
$$\delta ' \leq \epsilon _{D^*}^{\mathcal P^*} = \mathbb E_{i,X}\left[ \epsilon_{D^*_{i;x_1,\dots ,x_{i-1}}}^{\mathcal P^*}\right]$$
For each $i\in \{ 1,\dots ,n\}$ and $x_1,\dots ,x_{i-1}$, we define a protocol
$\mathcal P'_{i;x_1,\dots ,x_{i-1}}$ for $f$ as follows: Let
$\ket\psi = \sum _x\sqrt{p_x}\ket x$ where $p_x$ is the probability
of $x$ under distribution $D$. Let $x\in A$ be the input for Alice
and $y\in B$ the input for Bob.
Alice and Bob run the protocol $\mathcal P^*$ on
input $\ket{x_1}\dots \ket{x_{i-1}}\ket x\ket{\psi }^{\otimes n-i+1}$ for
Alice and $\ket i\ket{x_1}\dots \ket{x_{i-1}}\ket y$ for Bob and output the result.
The error probability of $\mathcal P'_{i;x_1,\dots ,x_{i-1}}$ is the same
as that of $\mathcal P^*$ under distribution $D^*_{i;x_1,\dots ,x_{i-1}}$,
that is,
$$\epsilon _D^{\mathcal P'_{i;x_1,\dots ,x_{i-1}}} = \epsilon _{D^*_{i;x_1,\dots ,x_{i-1}}}^{\mathcal P^*}$$

Since $\mathcal P^*$ is a safe coinless quantum protocol, $\mathcal P'_{i;x_1,\dots ,x_{i-1}}$
is such a protocol too. Let $X'$ be the classical random variable
denoting Alice's input in $\mathcal P'_{i;x_1,\dots ,x_{i-1}}$.
The density matrix $M'$ of Alice's first message
in $\mathcal P'_{i;x_1,\dots ,x_{i-1}}$ is the same as that of the first
message in $\mathcal P^*$ when $X_1,\dots ,X_{i-1}$ are set to
$x_1,\dots ,x_{i-1}$. Thus, by Lemma \ref{for_elim}, there exists
a $[t-1,c+l_1,l_2,\dots ,l_t]^B$-safe coinless quantum protocol with
error probability
\begin{align*}
\epsilon _D^{\mathcal P_{i;x_1,\dots ,x_{i-1}}} &\leq
\epsilon _D^{\mathcal P'_{i;x_1,\dots ,x_{i-1}}} + (2(\ln 2)I(X':M'))^{1/4}\\
&= \epsilon _{D_{i;x_1,\dots ,x_{i-1}}}^{\mathcal P*} + (2(\ln 2)I((X_i:M)|X_1 = x_1,\dots ,X_{i-1}=x_{i-1}))^{1/4}
\end{align*}
We now define a $[t-1,c+t_1,t_2,\dots ,t_n]^B$-safe public coin quantum
protocol as follows: Alice and Bob use the public coin to select $i\in\{ 1,\dots ,n\}$
uniformly at random and sample $x_1,\dots ,x_{i-1}$ independently according to
$D$. Then they run the protocol $P_{i;x_1,\dots ,x_{i-1}}$. The error
probability is
\begin{align*}
\epsilon _D^{\mathcal P} &= \mathbb E_{i,X_1,\dots ,X_{i-1}}\left[ \epsilon _D^{\mathcal P_{i;x_1,\dots ,x_{i-1}}}\right]\\
&\leq \mathbb E_{i,X_1,\dots , X_i}\left[ \epsilon _{D_{i;x_1,\dots ,x_{i-1}}}^{\mathcal P*}\right] + \left( 2(\ln 2)\mathbb E_{i, X_1,\dots , X_i}\left[ I((X_i:M)|X_1 = x_1,\dots ,X_{i-1}=x_{i-1}))\right]\right) ^{1/4}\\
&\ \ \ \ \text{ since the 4th root function is concave.}\\
&\leq \delta ' + \left( \frac{4\ln 2}n\right) ^{1/4}\text{ by Equation \eqref{bound}.}
\end{align*}
This completes the proof.
\end{proof}
We now show how the predecessor problem reduces to the \em rank parity
problem\em .
\begin{defin}[Rank Parity Problem]
In the \em rank parity \em communication problem $\pty p q$, Alice is given
a number $x$ in $[2^p]$
and Bob is given a set $S\subseteq [2^p]$ with $|S|\leq q$. The
rank of $i\in [2^p]$ in $S$ is defined as
$\rank _S(i) = |\{ j\in S\mid j\leq i\}|$,
i.e., the number of elements in $S$ that are not greater than $i$. The
goal of the rank parity problem is to determine
$\rank _S(x)\bmod 2$.
\end{defin}
\begin{prop}
\label{to_rp}
Suppose that there is a data structure for the predecessor problem
that has block-size $(\log m)^{O(1)}$, uses $n^{O(1)}$ blocks of
space and that allows to answer predecessor queries with an address-only quantum algorithm with worst-case time complexity
$t$ and error probability $\epsilon$. Then, there is a
$\left( 2t+ O(1),0,O(\log n), (\log m)^{O(1)}\right) ^A$-safe coinless
quantum protocol for $\pty{\log m}n$ with error probability
at most $\epsilon$.
\end{prop}
\begin{proof}
Let $\phi$ be a data structure for the predecessor problem as
described in the premise of the proposition. We will prove the proposition
by describing a data structure $\psi$ for rank parity queries
which can be converted to a safe coinless quantum protocol by
Lemma \ref{data_to_com}.

Let $S\subseteq [m]$.
As said in Remark \ref{rank_hash}, we can use a Perfect Hash table
to store not only the set $S$ but also the rank of each element
of $S$. The set $S$ is encoded as $\psi (S)$ consisting of
such a hash table together with $\phi (S)$. To find the rank parity
of some $x\in [m]$, we first determine the predecessor $x'$ of $x$
in $S$ for which we need to read $t$ blocks. Then, we look up
$x'$ in the hash table to find out its rank, reading $O(1)$ cells.
We check whether $x\in S$ which again requires $O(1)$ cell-probes.
Now we can compute
$$\rank _S(x) = \begin{cases}\rank _S(x')+1 &\text{ if }x\in S\\
\rank _S(x')&\text{ otherwise}\end{cases}$$
and thus $\rank _S(x)\bmod 2$. The total time complexity of
this algorithm is $t+O(1)$ cell probes. The space complexity of our data
structure is $n^{O(1)} + O(n) = n^{O(1)}$ cells. The only possible source of
error is the predecessor query algorithm. If it returns the correct result,
we obtain the correct value for $\rank _S(x)\bmod 2$. Thus, the
error probability for our query algorithm is at most $\epsilon$.

By Lemma \ref{data_to_com}, there
exists a $\left( 2(t+O(1)),0,\log \left( n^{O(1)}\right) , (\log m)^{O(1)}\right)^A$-safe
coinless quantum protocol for $\pty{\log m}n$ with error probability
at most $\epsilon$. We have $2(t+O(1)) = 2t+O(1)$ and
$\log (n^{O(1)}) = O(1)\cdot \log n = O(\log n)$.
\end{proof}
The following two propositions were proved in the classical setting by Miltersen et al.
in \cite{mnsw}.
\begin{prop}
\label{first_reduction}
Let $k$ and $p$ be integers such that $k$ divides $p$.
If there is a $[t,c,l_1,\dots ,l_t]^A$-safe coinless (public coin) quantum protocol
for $\pty p q$ with error probability $\epsilon$, then there
also is a $[t,c,l_1,\dots ,l_t]^A$-safe coinless (public coin) quantum protocol
for $\pty{p/k}{q}^{(k)}$ with the same error probability.
\end{prop}
\begin{proof}
Let $\mathcal P$ be a $[t,c,l_1,\dots ,l_t]^A$-safe quantum protocol
for $\pty p q$ with error probability $\epsilon$. We can use it
for designing a protocol for $\pty{p/k}{q}^{(k)}$ as follows.
Let $x_1,\dots ,x_k$ be the inputs for Alice. Let $x\in [2^p]$ be the number that
results from concatenating (the binary representations of) $x_1,\dots ,x_k$.
Let $S$ be the set that Bob receives as input and $i$ the number in $\{ 1,\dots ,k \}$ he receives.
Define a set $S'\subseteq [2^p]$ of size at most $n$ by
$$S' = \{ x_1\circ \dots  \circ x_{i-1}\circ y\circ 0^{p-i(p/k)}\mid y\in S\}$$
where $\circ$ denotes concatenation.
Alice computes $x$ and Bob computes $S'$. Then, they run the protocol
$\mathcal P$ on inputs $x$ and $S'$. We now show that if $\mathcal P$
does not make an error, this protocol returns the correct result.

The correct result on input $x_1,\dots ,x_k,i,S$ is $\rank _S(x_i)\bmod 2$.
For every $y\in S$, we have
$$x_1\circ \dots  \circ x_{i-1}\circ y\circ 0^{p-i(p/k)} \leq x_1\circ \dots  \circ x_k$$
if and only if $y\leq x_i$. Thus, $\rank _S(x_i)\bmod 2
= \rank _{S'}(x)\bmod 2$ which is the value that $\mathcal P$ computes.
\end{proof}
\begin{prop}
\label{second_reduction}
Suppose $k$ divides $q$ and $q$ is a power of 2. If there is a
$[t,c,l_1,\dots ,l_t]^B$-safe coinless (public coin) quantum protocol $\mathcal P$ for the problem
$\pty p q$ then there also is such a protocol for the
problem $^{(k)}\pty{p-\log k - 1}{q/k}$ that has the same error probability
as $\mathcal P$.
\end{prop}
\begin{proof}
Given $\mathcal P$, we can design a protocol for $^{(k)}\pty{p-\log k - 1}{q/k}$
as follows: Alice is given $x\in [2^{p-\log k - 1}]$ and $i\in \{ 1,\dots ,k\}$ and
Bob receives $S_1,\dots ,S_k\subseteq [2^{p-\log k - 1}]$ with $|S_j|\leq q/k$.
First, Alice computes $x' = (i-1)\circ 0 \circ x\in [2^p]$
and Bob computes for
every $j\in \{ 1,\dots ,k \}$ the set
$$S_j' = \begin{cases}
\{ (j-1)\circ 0 \circ y\mid y\in S_j\}&\text{ if }|S_j|\text{ is even}\\
\{ (j-1)\circ 0 \circ y\mid y\in S_j\}\cup \{ (j-1)\circ 1^{p-\log k}\} &\text{ if }|S_j|\text{ is odd}\end{cases}$$
Note that $S_j'$ always has an even number of elements.
Bob takes the union $S = \bigcup _{j=1}^k S_j'$.
All elements of $S$ are in $[2^p]$ and the cardinality
of $S$ is at most $q$ since the $S_j'$ all have cardinality at most $q/k$.
Now, Alice and Bob execute the protocol $\mathcal P$ on $x'$ and $S$
and output the result.

If $\mathcal P$ gives the correct result, this protocol
returns $\rank _{S_i}(x)\bmod 2$: Let $i,j\in \{ 1,\dots ,k \}$. If
$i < j$ then we have $\rank _{S'_j}(x') = 0$. If $i = j$, we
have $\rank _{S'_j} (x') = \rank _{S_i}(x)$ because
$(i-1)\circ 0\circ x\geq (i-1)\circ 0\circ y$ if and only if
$x\geq y$. If $i> j$ then $\rank _{S'_j}(x') = |S'_j| = 0\bmod 2$.
Because of this and because the $S'_j$ are disjoint, we have
$$\rank _S(x') \bmod 2= \sum _j \rank _{S'_j}(x')\bmod 2 = \rank _{S_i}(x)\bmod 2$$
as required.
\end{proof}
Now, we finally have all the tools we need to prove the lower bound
on the predecessor problem.
We start by assuming that there is some data structure that
violates the lower bound.
The main idea of the proof is
to reduce the predecessor problem to $\pty{\log m}n$ and to apply the previous two
propositions and the Round Elimination Lemma to obtain a protocol
for $\pty p q$ \em without communication \em which has error probability
smaller than $1/2$. Such a protocol is impossible, so there
can be no data structure that violates the lower bound.
\begin{thm}
Suppose that we have a data structure for the predecessor problem
for sets $S\subseteq [m]$ of size at most $n$
with block-size $(\log m)^{O(1)}$ that uses space $n^{O(1)}$. Suppose
that there is an address-only quantum cell-probe algorithm for
determining the predecessor of any $x\in [m]$ in $S$ that makes $t$
cell-probes to the representation of $S$. Suppose further
that the error probability of that algorithm is less than 1/3. Then,
it holds that:
\begin{itemize}
\item There is a function $N:\mathbb N \to \mathbb N$ such that for
$n=N(m)$, we must have
$$t = \Omega \left(\frac{\log\log m}{\log\log\log m}\right)\text .$$
\item There is a function $M:\mathbb N\to \mathbb N$ such that for
$m = M(n)$, we must have
$$t = \Omega\left( \sqrt{\frac{\log n}{\log\log n}}\right)\text .$$
\end{itemize}
This lower bound on $t$ also holds in the classical deterministic
and probabilistic setting.
\end{thm}
\begin{proof}
Let $c_1 = (4\ln 2)12^4$.
Suppose we have a data structure for the predecessor problem
that uses at most $n^{c_2}$ blocks of size $(\log m)^{c_3}$ for
some constants $c_2,c_3\geq 1$.
Let $n = 2^{(\log\log m)^2/\log\log\log m}$.
Suppose that the predecessor query
algorithm makes at most
\begin{align*}
t = \frac{\log\log m}{(c_1+c_2+c_3)\log\log\log m}
&=\frac{1}{(c_1+c_2+c_3)\log\log m}\cdot\frac{(\log\log m)^2}{\log\log\log m}\\
&\geq\frac{\log n}{(c_1+c_2+c_3)\log\log m}\\
&\geq\frac{1}{c_1+c_2+c_3}\cdot \sqrt{\frac{\log n}{\log\log n}}
\end{align*}
cell-probes and has error probability $\delta < 1/3$. We will now derive a
contradiction from this assumption.

Let $a = c_2\log n$ and $b = (\log m)^{c_3}$. By Proposition
\ref{to_rp}, there exists a $(2t,0,a,b)^A$-safe coinless
quantum communication protocol $\mathcal P$ that solves the
problem $\pty{\log m}n$ with error probability at most $\delta$.
Let $p_1 = \log m/(c_1at^4)$ and $q_1 = n$.
By Proposition \ref{first_reduction},
there is a $(2t,0,a,b)^A$-safe coinless quantum protocol
that solves $\pty{p_1}{q_1}^{(c_1at^4)}$ with error
probability at most $\delta$.

By the Quantum Round Elimination Lemma (Lemma \ref{qu_round_elim}),
it follows that there is a $(2t-1,a,a,b)^B$-safe public coin quantum
protocol for $\pty{p_1}{q_1}$ with error probability
at most $\delta + (12t)^{-1}$. Let $p_2 = p_1 - \log (c_1bt^4) -1$
and $q_2 = \lfloor q_1/(c_1bt^4)\rfloor$. By Proposition
\ref{second_reduction}, there is a $(2t-1,a,a,b)^B$-safe public-coin
quantum protocol for the problem $^{(c_1bt^4)}\pty{p_2}{q_2}$.

We have
\begin{equation}
\label{eq:noidea}
\frac{\log m}{(2c_1at^4)^i}\geq \log{c_1bt^4}+1\text{ for all $i\leq t$}
\end{equation}
which implies that
$$p_2\geq \frac{\log m}{c_1at^4} - \frac{\log m}{2c_1at^4} =
\frac{\log m}{2c_1at^4}$$
and thus, there is a $(2t-1,a,a,b)^B$-protocol for the problem
$^{(c_1bt^4)}\pty p q$ with
$$p = \frac{\log m}{2c_1at^4}, q = \frac{n}{c_1bt^4}$$

Applying the Round Elimination Lemma again, we obtain
a $(2t-2,a+b,a,b)^A$-safe public coin quantum protocol
for the problem $\pty{p}{q}$ that has error probability at
most $\delta + 2(12t)^{-1}$.

Iterating this process, we let $p_1' = p/(c_1at^4)$
and $q'_1 = q$. Proposition \ref{first_reduction} gives us a
$(2(t-1), a+b, a, b)^A$-safe protocol for $\pty{p'_1}{q'_1}^{(c_1at^4)}$.
Applying the Round Elimination Lemma, we get a $(2t-3, 2a+b, a, b)^B$-safe
protocol for $\pty{p'_1}{q'_1}$ with error probability at most
$\delta + 3(12t)^{-1}$. Now let $p'_2 = p'_1 - \log (c_1bt^4) - 1$ and
$q'_2 = \lfloor n/(c_1bt^4)\rfloor$. With Proposition \ref{second_reduction},
we get a $(2t-3, 2a+b, a, b)^B$-safe protocol for the problem
$^{(c_1bt^4)}\pty{p_2'}{q_2'}$. Because of Equation \eqref{eq:noidea}, we have
$$p_2' \geq \frac{\log m}{2(c_1bt^4)^2} - \frac{\log m}{(2c_1at^4)^2}
= \frac{\log m}{(2c_1at^4)^2}$$
and thus, we have a protocol for $^{(c_1bt^4)}\pty{p'}{q'}$ for
$$p' = \frac{\log m}{(2c_1at^4)^2}, q' = \frac{n}{(c_1bt^4)^2}$$
Applying the Round Elimination Lemma, we obtain a
$(2(t-2), 2(a+b), a, b)^A$-safe protocol for $\pty{p'}{q'}$ that has error
probability at most $\delta + 4(12t)^{-1}$.

We continue this process for $t$ iterations in total. After the $i$th
iteration, we have a $(2(t-i),i(a+b),a,b)^A$-safe public-coin quantum protocol
for the problem $\pty_{p,q}$ with
$$p = \frac{\log m}{(2c_1at^4)^i},q = \frac{n}{(c_1bt^4)^i}$$
and error probability $\delta + 2i(12t)^{-1}$ and thus, after $t$ iterations, we have a
$(0,t(a+b),a,b)^A$-safe public-coin protocol for the problem
$\pty p q$ with
$$p = \frac{\log m}{(2c_1t^4)^t} \geq (\log m)^{\Omega (1)},
q = \frac{n}{(c_1bt^4)^t}\geq n^{\Omega (1)}$$
that has error probability $\delta + 2t(12t)^{-1} = \delta + 1/6 < 1/2$.
That means that Alice can guess with a worst-case error probability better than one half
the rank parity of her input $x$ in Bob's set $S$ without communicating with Bob
and without any shared entanglement.
This clearly is impossible.
\end{proof}

\section{Fully Quantum Data Structures}
\label{full_quant}

\subsection{Introduction}

After looking at classical data structures and lower bounds in the
setting of quantum access to classical data structures, we now turn
to data structures where the data is encoded not in classical bits
but in qubits. The query algorithms may use any unitary transforms
and any measurements on the data. While we may compare the size of such a
fully quantum data structure to the size of its classical counterparts,
this approach is not comparable to the classical or quantum cell-probe model
in terms of time complexity.
Another problem in this setting
is that if a query algorithm involves measurements then the data
may be irreversibly altered. Therefore, we will also need to consider
how many times a data structure can be used.

\subsection{Set Membership}

Our first example of a fully quantum data structure is a data structure
for the set membership problem found by Buhrman, Cleve, Watrous and de Wolf
which is described in \cite[Section 8]{dewolf_thesis}. This data structure is
based on a solution to the
\em quantum fingerprinting \em problem where we want to encode
$x, y\in [m]$ as quantum states $\ket{\phi _x}, \ket{\phi _y}$ which we can
use to determine whether $x = y$ with low error probability.

First, let us have a look at classical fingerprinting. Consider
the following situation: Alice and Bob each hold a bit
string $x$ and $y$ in $\{ 0,1 \} ^n$ respectively. They want to find out whether
$x = y$ while keeping the amount of communication small.
The trivial solution would be for one party to send
the whole bit string to the other. If they want to have certainty,
this approach is actually optimal. If they are content
with a probabilistic test, there are better ways.
Let $\epsilon$ be the error probability they want to allow.
Choose a prime power $q\geq (n-1)/\epsilon$ and let $\mathbb F$ be the finite
field with $q$ elements. Let $a = a_0\dots a_{n-1}\in \{ 0,1 \}^n$
and $f_a = \sum _{i=0}^{n-1}a_iX^i$. If $x = y$, then $f_x = f_y$.
If $x\neq y$ then the polynomial $\bar f = f_x-f_y$ is non-zero.
We now use the following Lemma for which a proof can be found in
\cite[Lemma 16.4]{jukna}:
\begin{lem}
\label{jukna_lem}
Let $\mathbb F$ be a field and $f$ a non-zero polynomial of degree $d$ over that
field. Let $S$ be a finite, non-empty subset of $\mathbb F$. If we select
$r\in S$ uniformly at random, the probability that $f(r) = 0$ is at most
$d/|S|$.
\end{lem}
Therefore, the probability that $\bar f (r) = 0$ (and hence $f_x(r) = f_y(r)$)
is at most
$$\frac{\text{deg}(\bar f)}{q}\leq \frac{n-1}{q} \leq \frac{(n-1)\epsilon}{n-1}= \epsilon\text .$$

Alice and Bob could use the following protocol. Alice selects a
random $r\in \mathbb F$ and sends $r, f_x(r)$ to Bob. Bob computes
$f_y(r)$ and compares $f_x(r)$ and $f_y(r)$. If they are equal, he sends $1$
to Alice to indicate that $x = y$. Otherwise, he sends $0$.

If $x=y$, this protocol will always output the correct answer.
Otherwise, there is an error probability of at most $\epsilon$.
The communication that is required is $2\log |\mathbb F|+1$. If we choose
$\epsilon$ as some small constant (or even $\epsilon = 1/poly (n)$),
we can choose $q = O(n)$ ($q\in poly(n)$) and have communication complexity
$O(\log n)$.

This protocol depends on Alice and Bob sharing a random number. But what
if they cannot do that? Let us now consider the following scenario:
Alice and Bob again have inputs $x$ and $y$ respectively, but now there is a
\em referee \em whose task is to decide whether $x = y$. Alice and Bob have
to enable the referee to do that with good probability.
They may each send only \em one \em message to
the referee and cannot communicate with each other (in the quantum case, they
also do not share entanglement). It is clear that the scheme described above
does not help us here since Alice and Bob cannot share randomness.
However, we can use a quantum version of our previous scheme
by putting the values $f_a(r)$ in superposition. More precisely, let
$$\ket{\phi _a} =\sum _{r\in\mathbb F}\frac{1}{\sqrt{|\mathbb F|}}\ket r\ket{f_a(r)}\text .$$
Alice sends $\ket{\phi _x}$ to the referee and Bob sends $\ket{\phi _y}$.
If $x = y$, these states are identical, but if $x\neq y$, they are nearly
orthogonal. Two polynomials of degree $\leq n-1$ can have the same value on
at most $n-1$ elements of $\mathbb F$. Thus, for distinct $x$ and $y$,
$$0\leq |\inp{\phi _x}{\phi _y}| \leq \frac{n-1}{q}\leq \epsilon \text .$$
The referee then applies a \em swap test \em (see \cite[Section 8]{dewolf_thesis} for details) to determine whether these states are identical or almost
orthogonal. If the states are equal, the test will always have result 1, and if
not, it has result 1 with probability below $(1 + \epsilon^2)/2$. Repeating the
swap test several times on different fingerprints, we can tell these two cases
apart with good probability. We can again choose $\epsilon$ as some small
constant and $q\in O(n)$ to obtain a protocol that solves the problem with low
error probability and $O(\log n)$ communication.

Let us now see how to construct a data structure for the set membership problem
from these quantum states. First, we show how to store singletons.
We encode $x\in [m]$ as $\ket{\phi _x}$ where we choose $\mathbb F$ as a field
of size at least $(\log m -1)/\epsilon$. This encoding requires
$2\log |\mathbb F| = O(\log\log m - \log\epsilon)$ qubits.
Queries ``$y = x?$'' are answered
by first appending a fresh qubit initialized to $\ket 0$ to $\ket{\phi _x}$
and performing the unitary transform given by
$$\ket{r}\ket{z}\ket b\mapsto \ket r\ket{z}\ket{b\oplus [z = f_y(r)]}$$
where $[z = f_y(r)]$ denotes 1 if $z = f_y(r)$ and 0 otherwise.
The state after the transform is
$$
\ket{\phi _x} = \sum _{r\in\mathbb F,f_x(r) \neq f_y(r)}\sqrt{\frac{1}{|\mathbb F|}}\ket r\ket{f_x(r)}\ket 0+
\sum _{r\in\mathbb F,f_x(r) = f_y(r)}\sqrt{\frac{1}{|\mathbb F|}}\ket r\ket{f_x(r)}\ket 1
$$
and therefore, if $x = y$, we will always receive outcome 1 when measuring
the last qubit. If $x\neq y$, then there are less than $\epsilon\cdot q$
elements $r\in\mathbb F$ such that $f_x(r) = f_y(r)$. Thus, the probability
of measuring $1$ is less than $\epsilon$ in this case.

If $y = x$, the measurement does not alter $\ket{\psi _x}$.
If
$y\neq x$, let $S_y = \{ r\in\mathbb F\mid f_x(r)\neq f_y(r)\}$.
We have
$$
|S_y|\geq |\mathbb F| - (\log m-1)\geq \frac{\log m -1}{\epsilon} - (\log m -1)
= \frac{1-\epsilon}{\epsilon}(\log m-1)
$$
If the measurement returned 0, the state after the measurement will be
$$\sum _{r\in S_y}\frac 1{\sqrt{|S_y|}}\ket r\ket{f_x(r)}\ket 0$$
and if we make another query to this state, Lemma \ref{jukna_lem} can only guarantee
an error probability of at most $\epsilon /(1-\epsilon )$: In the worst case,
the two polynomials $f_x$ and $f_y$ agree on $\log m -1$ values for $r$ and
$(\log m-1)/\epsilon$ is already a prime, so
$|\mathbb F| = (\log m-1)/\epsilon$. Then, our new quantum state only contains
a superposition over $(1-\epsilon )(\log m -1)$ elements $r\in \mathbb F$ and
their corresponding values of the polynomial. In that case,
Lemma \ref{jukna_lem} can only guarantee a success probability of at most
$\epsilon/(1-\epsilon )$ when we check whether $z=x$.

This means that if we originally had error probability $1/k$ then
we can only guarantee $1/(k-1)$ now. If we want to handle more
queries, we can enlarge the field $\mathbb F$.

It is possible to design a data structure such that $k$ successive queries
$y_1 = x?,\dots ,y_k = x?$ are all answered correctly with probability
at least 2/3. Let $\epsilon = 1/(4k)$ and let $\mathbb F$ be a
field of size $q$ with $q = O((\log m)/\epsilon) = O(k\log m)$ and $q\geq (\log m-1)/\epsilon = 4k(\log m-1)$.
Using this field in the construction above yields such a data structure.
After $i< k$ queries, the quantum state of our data structure is
$$\sum _{r\in S_i}\frac{1}{\sqrt{|S_i|}}\ket r\ket{\phi _x(r)}$$
with $|S_i|\geq |\mathbb F| - i(\log m-1)$. Thus, the probability
that each of the $k$ queries has the correct result is at least
$$\prod _{i=1}^{k-1}\left( 1-\frac{\log m-1}{4k(\log m-1)-i(\log m-1)}\right)
= \prod _{i=1}^{k-1}\left( 1-\frac 1{4k-i}\right)\geq \left( 1-\frac 1{3k}\right) ^k\geq \frac 2 3$$
This proves the following theorem:
\begin{thm}
For positive integers $m,k$
there is a quantum data structure that encodes elements $x\in [m]$ in
$O(\log\log m+\log k)$ qubits such that for $k$ successive queries of the form
``$y = x?$'', the probability that they are all answered correctly
is at least 2/3.
\end{thm}
We can encode a set $S\subseteq [m]$ by simply storing a fingerprint
for each element. We answer the query $y\in S$ by answering
whether $y = x$ for any $x\in S$. Note that we have to reduce
the error probability for the individual fingerprints to, say, $1/(4n)$.
\begin{thm}
There is a data structure that stores sets $S\subseteq [m]$ of
size at most $n$ in $O(n(\log\log m+\log k + \log n))$ (or
$O(n(\log\log m + \log k))$ for $k$ or $m$ large enough)
qubits such that $k$ successive queries
``$y\in S?$'' are all answered correctly with probability at least
2/3.
\end{thm}
Let us compare this result to the classical setting. The information-theoretic
minimum for storing $S\subseteq [m]$ of size $n$ is $\Omega(n\log m)$
bits. The only way around this limitation would be to also consider
data structures that work for \em most \em queries but fail
on some. We could
encode single elements $x\in [m]$ by selecting $r\in\mathbb F$ uniformly
at random and storing $(r,f_x(r))$. Then, we could check whether $y = x$ by comparing
$f_x(r)$ and $f_y(r)$. This method could be extended to sets by fingerprinting
each element, as in the previous theorem.

While this requires as many bits as our quantum data structure requires
qubits, the downside of the classical version is that, while
for most $y\in [m]$, the query ``$y\in S?$'' will be answered correctly,
there are some $y\in [m]$, determined when the set $S$ is encoded,
such that ``$y\in S?$'' will \em always \em be answered incorrectly.

A lower bound from \cite{dewolf_thesis} on the size of fully quantum set membership data structures is
$\Omega (n)$.
\begin{thm}
Every fully quantum data structure for the set membership problem requires
$\Omega (n)$ qubits.
\end{thm}
\begin{proof}
We show how we can use such a data structure as a \em quantum random access code \em (QRAC)
and then apply a lower bound on such codes due to Nayak in \cite{nayak}.
A QRAC encodes bit strings $x\in\{ 0,1 \} ^n$ in $l$-qubit states
$\ket{\psi _x}$ such that for each $i\in [n]$, we can recover $x_i$ from
$\ket{\psi _x}$ with probability $p$. A QRAC has to guarantee that we can
recover \em any \em bit of our choice with good probability, but it does not
have to guarantee that
we can recover more than one bit. The lower bound by Nayak is
$l\geq (1-H(p))n$ where $H(p) = -p\log p - (1-p)\log (1-p)$ is the binary
entropy function.

We represent $x$ as $S_x = \{ i \mid x_i = 1\}$.
Given a quantum data structure for set membership that uses $S(m, n)$ qubits,
we can store $S_x$ and use this as a QRAC for $x$ by
querying it on the index $i$ that we are interested in. Thus, if our data
structure achieves a success probability greater than $1/2$, we must have
$S(m, n)\geq \Omega(n)$.
\end{proof}

\subsection{Quantum Walks and Data Structures}
\newcommand{\pr}{\text{Pr}}
In this section, we will present \em quantum walks\em , a framework for the construction
of quantum algorithms, and show how it can use fully quantum data structures.
Using this framework, one can construct algorithms that work similar to
Grover's search algorithm. This framework also makes it easy to analyse
different kinds of costs of the constructed algorithms.
In contrast to the rest of this survey, the data structures here
are dynamic. That is, it also is important that they can be
updated with low cost.
We will use the framework to give an algorithm for triangle finding that was
discovered by Jeffery, Kothari and Magniez in \cite{quant_walk}.

First, let us have a look at the classical counterpart of quantum walks: \em Random walks\em .
We will consider random and quantum walks on a graph $\mathcal G = (\mathcal V,\mathcal E)$
where each vertex has exactly $d$ neighbours for some $d$.
More generally, we can consider Markov chains instead of graphs, but for
simplicity, we stick with graphs. This suffices for the application
which we describe. Let $\mathcal I$ be a set of possible inputs.
With each $x\in\mathcal I$, we associate a set $M_x\subseteq\mathcal V$.
Let $G$ be the adjacency matrix of $\mathcal G$ and $\delta$ the spectral
gap of $\frac 1 dG$. The spectral gap can, somewhat imprecisely, be described
as the difference between the largest and second largest eigenvector. If
$\lambda _1,\lambda _2,\dots , \lambda _n$ are the eigenvalues of $\frac 1 d G$,
listed with multiplicity and sorted in descending order with respect to their absolute
values, then the spectral gap is defined as $\delta = |\lambda _1| - |\lambda _2|$.
If $G$ is the adjacency matrix of some graph, then the eigenvalue of
$\frac 1 dG$ with the largest absolute value is always
$1$, so we have $\delta = 1-|\lambda _2|$.
We also call the spectral gap of $\frac 1 dG$ the spectral gap of the graph
$\mathcal G$.

We want to construct an algorithm that finds an element of $M_x$
given $x$. A random walk works as follows:
\begin{enumerate}
\item Choose a vertex $u\in\mathcal V$ uniformly at random.
\item Repeat the following until a vertex $v\in M_x$ is found:
\begin{enumerate}
\item Check if $u\in M_x$, if yes, output $u$.
\item Do the following $\lceil 1/\delta \rceil$ times:
Select a neighbour $v$ of $u$ uniformly at random and set $u = v$.
\end{enumerate}
\end{enumerate}

We will analyse the expected cost of this algorithm in terms of \em cost vectors\em . A
cost vector may store several kinds of costs for one operation that
are considered relevant. For example, we might associate with some
algorithm a vector that contains only the bit-probe complexity or
we might consider both the bit-probe and circuit complexity.
The three operations that form the random walk are \textbf{Setup}, step 1
above, \textbf{Checking}, step 2.(a), and \textbf{Update}, step
2.(b) and the associated cost vectors are $S$, $C$ and $U$,
respectively. Then, the expected cost is roughly
$$T(\epsilon, \delta) = S+\frac{1}{\epsilon}\left( C+\frac 1\delta U\right)$$
where $\epsilon = |M_x|/|\mathcal V|$. This holds because one can show that
by making about $1/\delta$ random steps starting from any vertex, we sample
a distribution on $\mathcal V$ that is close to uniform.

Obviously, if $M_x = \emptyset$, this algorithm can never terminate. If
we know $\epsilon > 0$ such that for all $M_x\neq \emptyset$, $\epsilon \leq |M_x|/|\mathcal V|$
(a trivial lower bound would be $1/|\mathcal V|$)
then we can construct a bounded-error algorithm that determines
whether $M_x = \emptyset$ and, if not, finds
some $u\in M_x$. This is done by running $\lceil 3/\epsilon\rceil$
iterations of the loop in step 2. If we obtain some output,
we have found $u\in M_x$. If not, we conclude that $M_x = \emptyset$.
This algorithm has a worst-case cost of $O(T(\epsilon , \delta ))$.
If $M_x = \emptyset$, the algorithm cannot make an error. If $M_x\neq\emptyset$,
it will only output a vertex $u$ if it is indeed in $M_x$. The probability
that it will falsely report $M_x = \emptyset$ is less than $1/3$.

To see this, let the random variable $T_x$ be the number of iterations of the
loop in step 2 to find some $u\in M_x$. Then, $E[T_x] \leq 1/\epsilon$. We now
estimate the probability that the random walk requires more than
$3/\epsilon$ iterations. By Markov's inequality,
$$\pr \left[ T_x > \frac 3 \epsilon \right]<
\frac\epsilon 3 \mathbb E[T_x]\leq\frac 1 3$$

Data structures can improve the efficiency of random walks as follows: With
each vertex $v$ of the graph, we associate some data $d_{v, x}$ that helps
us to decide whether $v\in M_x$, i.e., knowing $d_{v, x}$ reduces the checking
cost $C$. In the Setup phase, we also store some representation of $d_{v, x}$
for the starting vertex $v$. When we move from vertex $v$ to vertex $v'$,
we update the representation of $d_{v, x}$ to $d_{v', x}$. Thus, we can trade
off an increase in the Setup cost $S$ and Update cost $U$ for a decrease in
$C$. Depending on the problem and the data, this might reduce the overall cost.

With a quantum computer, we can reduce the factors $1/\epsilon$
and $1/\delta$ to their square roots, using the quantum walk framework.
To understand quantum walks, it
is helpful to first know how Grover's search algorithm works. We will, however,
not give a full proof for Grover's algorithm. Such a proof can be found in
\cite[Chapter 6]{mikeike}, for example.

\begin{thm}[Grover's Algorithm]
Let $x$ be a  bit string of length $N = 2^n$. Suppose that for some $\epsilon > 0$
we are guaranteed that if $x$ has a non-zero entry then at least $\epsilon N$
of its entries are 1. (Such a guarantee is trivial for $\epsilon = 1/N$.)
There is a bounded-error quantum bit-probe algorithm that makes $O(\sqrt{1/\epsilon})$
bit-probes to $x$ and outputs an $i\in [n]$ such that $x_i = 1$ or reports
that $x$ is the all-zero string. In particular, $O(\sqrt N)$ bit-probes
suffice for any string $x$, using the trivial value for $\epsilon$ mentioned
above.
\end{thm}
\begin{proof}[Proof sketch]
We use the oracle
$$O_{x,\pm} : \ket{i}\mapsto\begin{cases}-\ket{i} &\text{ if }x_i = 1\\
    \ket i&\text{ if } x_i = 0\end{cases}$$
for bit-probes to $x$. First, we show how to
find an index $i\in [n]$ with $x_i = 1$ when we know that \em exactly \em
$\epsilon N$ bits of $x$ are 1. The algorithm operates on $n$ qubits which start
in the $\ket 0$-state. First, a Hadamard gate is applied to every qubit which
creates a uniform superposition $\ket{\mathcal U}$ over the states
$\ket 0,\dots , \ket{N-1}$.
Let $\ket{\mathcal G}$ be the uniform superposition over all ``good'' states,
i.e., the states $\ket i$ with $x_i = 1$, and $\ket{\mathcal B}$ the uniform
superposition over the ``bad'' states, i.e., $\ket i$ with $x_i = 0$. We can
write
$$\ket{\mathcal U} = \frac 1{\sqrt{N}}\sum _{i = 0}^{N-1}\ket i
= \sin (\theta)\ket{\mathcal G} + \cos (\theta )\ket{\mathcal B}
\text{ for } \theta = \arcsin (\sqrt{\epsilon})\text .$$
Each iteration of Grover's algorithm shifts the amplitude from the ``bad''
states towards the ``good'' ones.
After $k$ iterations, the amplitude of $\ket{\mathcal G}$ is
$\sin ((2k+1)\theta)$.
This
is achieved by applying the transform $H^{\otimes n}O_GH^{\otimes n}O_{x,\pm}$
where
$$O_G: \ket i\mapsto \begin{cases} \ket 0 &\text{ if }i = 0\\
-\ket i &\text{ if } i\neq 0\end{cases}$$
to our working state. To see that this transform has the intended effect,
notice first that for computational basis states $\ket i$
$$O_{x,\pm} : \ket i\mapsto \begin{cases} -\ket i &\text{ if }\ket i\text{
is orthogonal to }\mathcal B\\
\ket i &\text{ otherwise}\end{cases}$$
i.e., $O_{x, \pm}$ is a \em reflection \em through $\mathcal B$.
The transform
$H^{\otimes n}O_GH^{\otimes n}$ implements a reflection through
$\mathcal U$. This can be seen as follows: We can write
$O_G = 2\ket 0\bra 0 - I$. This gives us
$$H^{\otimes n}O_GH^{\otimes n} =
2\left( H^{\otimes n}\ket 0\right) \left( \bra 0H^{\otimes n}\right)
-H^{\otimes n}H^{\otimes n}
= 2\ket{\mathcal U}\bra{\mathcal U} - I$$
which shows that $H^{\otimes n}O_GH^{\otimes n}$ reflects through $\ket{\mathcal U}$.
The angle between $\ket{\mathcal U}$ and $\ket{\mathcal B}$ is $-\theta$.
If the state \em before \em the iteration was $\sin ((2k-1)\theta)\ket{\mathcal
G} + \cos ((2k-1)\theta )\ket{\mathcal B}$, the angle $(2k-1)\theta$ is first
changed to $-(2k-1)\theta$ by the reflection through $\ket{\mathcal B}$. Now
the angle between $\ket{\mathcal U}$ and our current working state is
$2k\theta$. Thus, after the reflection through $\ket{\mathcal U}$, our
working state becomes $\sin ((2k+1)\theta )\ket{\mathcal G} + 
\cos ((2k+1)\theta )\ket{\mathcal B}$.

If $\epsilon = \sin ^2(\pi /(2\cdot(2k+1)))$ for some positive integer $k$,
then after
$k$ iterations, our working state will be $\ket{\mathcal G}$. Measuring it
will give us an index $i\in [n]$ such that $x_i = 1$, by definition of
$\ket{\mathcal G}$. If $\epsilon$ is not of this form, we can nevertheless
bring our working state close to $\ket{\mathcal G}$ so that a measurement
will have a correct result with high probability by picking an integer $k$
such that $\sin ^2(\pi /(2\cdot (2k+1)))$ is as close as possible to $\epsilon$.

Since $\sin x\approx x$ for small $x$, we choose
$$k\approx \frac{\pi}{4\sqrt{\epsilon}}-\frac 1 2 = O\left(\frac 1{\sqrt{\epsilon}}\right)\text .$$
See Figure \ref{fig_grover} for an illustration of
the circuit. But what can we do when we do not know $\epsilon$ exactly?
If we apply too many iterations of Grover's algorithm, we will end up \em
decreasing \em the amplitude of the ``good'' states again.

\begin{figure}[h]
\centering
\includegraphics{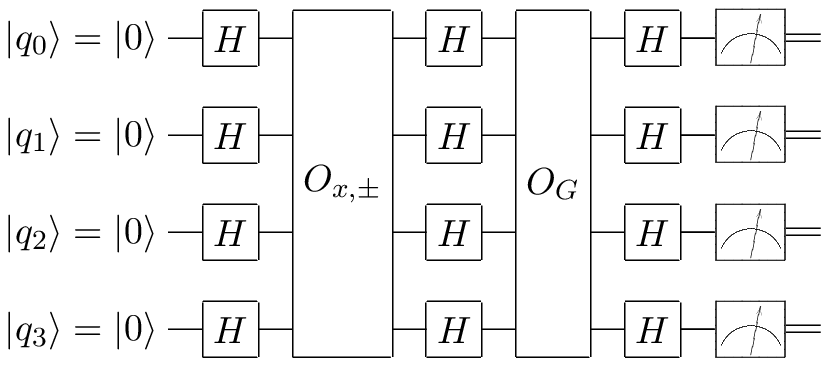}
\caption{Grover's search algorithm for input length $N = 16$, assuming that
$\epsilon = 1/4$ of the bits in $x$ are 1. In that case, we have
$\theta = \arcsin (1/2) = \pi/6$. Thus, one iteration will bring our working
state to $\ket{\mathcal G}$.}
\label{fig_grover}
\end{figure}

The case where either $x$ is all-zero or exactly $\epsilon N$ bits are 1 is
easy: We apply Grover's algorithm for $\epsilon$ and obtain an $i\in [n]$.
If $x_i = 1$, we output $i$ and if $x_i = 0$, we say that $x$ is all-zero.
If we only know that either $x$ is all-zero or at least $\epsilon N$ bits are
1, we can use a method by Boyer, Brassard, H\o yer and Tapp in
\cite{bbht} which solves our problem with $O(N/\epsilon )$ bit-probes using
several systematic guesses for the actual number of indices $i$ with $x_i = 1$.
\end{proof}
In this proof, we can identify the following basic ingredients:
\begin{itemize}
\item A uniform superposition $\ket{\mathcal U}$ over the whole search space,
\item A reflection through $\ket{\mathcal B}$, the uniform superposition
of the 0-elements,
\item A reflection through $\ket{\mathcal U}$.
\end{itemize}
In quantum walk algorithms, we apply these ingredients on a graph (or Markov
chain) instead of a bit string. The basic operations of quantum walk algorithms
are explained below. Our presentation here is based on lecture notes by
Ronald de Wolf \cite{lec_notes} which only treats quantum walks on graphs.
For a survey about walks on Markov chains, see \cite{quant_walk2}.
Let $\mathcal G = (\mathcal V,\mathcal E)$ be a $d$-regular graph and $G$ its
adjacency matrix. Let $|\mathcal V| = n$. For every $u\in\mathcal V$, let
$\mathcal V_u$ be the set of neighbours of $u$.
With each $x\in \mathcal I$ for some set $\mathcal I$ of possible inputs, we associate
a set $M_x\subseteq \mathcal V$. We want to design an algorithm for
deciding whether $M_x\neq \emptyset$ and finding a vertex $v\in M_x$ if it is
non-empty.
For a Hilbert space $\mathcal H_D$, let $D^x:\mathcal V\mapsto\mathcal H_D$
be a function that associates $x$ and a vertex $u\in \mathcal V$
with some quantum state $\ket{D^x(u)}$.
We call $\{ \ket{D^x(u)}\} _{u\in\mathcal V}$ the data structure associated
with $x$ and $\mathcal G$.
Fully quantum data structures were first used in \cite{quant_walk}
for this purpose; previously, only classical ones had been used.

Let $\mathcal H_L \simeq\mathcal H_R$ be Hilbert spaces with orthonormal
basis vectors $\ket 0_L, \ket u _L$ for $u\in\mathcal V$ and
$\ket 0 _R, \ket u_R$ for $u\in\mathcal V$, respectively. Since the basis
states of $\mathcal H_L$ and $\mathcal H_R$ share the same set of labels,
we use $L$ and $R$ subscripts to specify the Hilbert space that a basis vector
belongs to. The \textbf{Setup}
consists of constructing the quantum state
$$\ket{\mathcal U^x} = \frac{1}{\sqrt{dn}}\sum _{u\in\mathcal V}
\sum _{v\in\mathcal V_u}\ket u_L\ket v_R
\ket{D^x(u)}$$
given bit-probe access to $x$. We can view this state as the uniform
superposition over all edges of the graph. We denote the cost vector for this
operation by $S$.

As the \textbf{Checking} operation, we compute the following transform:
$$\ket u_L\ket v_R\ket{D^x(u)}
\mapsto\begin{cases}-\ket u_L\ket v_R\ket{D^x(u)} &\text{ if }u\in M_x\\
\ket u_L\ket v_R\ket{D^x(u)}&\text{ otherwise}\end{cases}$$
and we denote the cost vector for this operation with $C$. The checking
operation may also have some small error probability. This operation implements
the reflection through $\ket{\mathcal B}$.

It now remains to implement the reflection through $\ket{\mathcal U}$. For
this, we use the \textbf{Update} operation.
The \textbf{Update} operation consists of the following unitary transforms,
where $\mathcal V_u$ denotes the set of neighbours of $u$ in $\mathcal G$
\begin{align*}
P_A: \ket u_L\ket 0_R\ket{D^x(u)}&\mapsto\sum _{v\in\mathcal V_u}
\frac 1{\sqrt{d}}\ket u_L\ket v_R\ket{D^x(v)}\\
P_B: \ket 0_L\ket v_R\ket{D^x(v)}&\mapsto\sum _{u\in\mathcal V_u}
\frac 1{\sqrt{d}}\ket u_R\ket v_L\ket{D^x(u)}
\end{align*}
The cost vector for executing these two operations, plus their inverses is
denoted by $U$. We can use these operations to reflect through
$\ket{\mathcal U^x}$, but showing this is not straightforward. In the following
discussion, we will suppress the data structure in the notation to make it more
readable. For every vertex $u$, let
$$\ket{p_u} = \frac{1}{\sqrt{d}}\sum _{v\in\mathcal V_u}\ket{u}_L
    \ket{v}_R$$
and for every vertex $v$, let
$$\ket{q_v} = \frac{1}{\sqrt{d}}\sum _{u\in\mathcal V_v}\ket{u}_L
\ket{v}_R$$
and let
\begin{align*}
\mathcal A &= \text{span}\{ \ket{p_u} \mid u\in\mathcal V\}\\
\mathcal B &= \text{span}\{ \ket{q_v}\mid v\in\mathcal V\}
\end{align*}
Let $\text{ref}_{\mathcal A}$ and $\text{ref}_{\mathcal B}$ be the reflections
through $A$ and $B$ and let $W(G) = \text{ref}_{\mathcal B}\cdot\text{ref}_{
\mathcal A}$. We can compute the reflection through $\mathcal A$ by
applying $P_A^*$, putting a $-$ in the amplitude if the second register is not
in state $\ket 0_R$ and then applying $P_A$. Similarly, we compute $\text{ref}
_{\mathcal B}$. Thus, we can compute $W(G)$.
We define matrices
$$A = \sum _{u\in\mathcal V} \ket{p_u}\bra{u}_L\bra 0_R,
B= \sum _{v\in\mathcal V} \ket{q_v}\bra{v}_L\bra 0_R\text .$$
Define the \em discriminant matrix \em $D(A, B)$ as $A^*B$.

The discriminant matrix is related to the adjacency matrix of our graph.
For every $u\in\mathcal V$ we have
\begin{align*}
\inp{p_u}{q_v} = \frac 1 d\left(\sum _{v'\in\mathcal V_v}\bra u_L\bra{v'}_R\right)
\left( \sum _{u'\in\mathcal V_v}\ket{u'}_L\ket{v}_L\right) &= \frac 1 d \sum _{u',v'}\inp{u}{u'}\inp{v'}{v}\\
	&= \begin{cases} 1/d &\text{ if } u\text{ and }v\text{ are neighbours}\\
		0 &\text{ otherwise}
\end{cases}
\end{align*}
and thus, $D(A,B) = \frac 1 d G$. The operation $D(A,B) = A^*B$ can thus be
seen as an analogue of a step in the random walk.

A theorem by Mario Szegedy relates the eigenvalues and -vectors of $D(A, B)$
to those of $W(G)$. We prove a simplified version here that suffices
for our purpose.

\begin{thm}[{Spectral Lemma, \cite[Theorem 1]{szeg04}}]
\label{spec_lem}
Let $A, B\in \mathbb C^{n\times m}$ such that each column of $A$ and $B$ is a
vector in $\mathbb C^n$ of length 1, $A^*A = B^*B = I$ and $D(A, B) = A^*B$
is Hermitian, i.e., $D(A,B) = D(A,B)^*$.
Let $\mathcal A$ and $\mathcal B$ be the subspaces of $\mathbb C^n$
spanned by the column vectors of $A$ and $B$, respectively, and
$\pi _A = AA^*$ and $\pi _B = BB^*$ projectors on these spaces.
Let $W = (2BB^* - I)(2AA^* - I) = \text{ref}_{\mathcal B}\cdot
\text{ref}_{\mathcal A}$. The following statements hold:
\begin{enumerate}
\item Every eigenvalue of $D(A, B)$ is real and has absolute value at most 1.
\item For every $\theta\in [ 0, \pi ]$, if $\cos\theta$ is an eigenvalue of $D(A,B)$
then $e^{\pm i\theta}$ are eigenvalues of $W$. If $e^{i2\theta}$ or $e^{-i2\theta}$
is an eigenvalue of $W$ with eigenvector in
$$\mathcal A + \mathcal B = \{ a + b\mid a\in\mathcal A, b\in\mathcal B\}$$
then $\cos \theta$ is an eigenvalue of $D(A,B)$.
\item On $\mathcal A\cap\mathcal B$ and on $\mathcal A^\perp\cap\mathcal B^\perp$,
$W$ acts as the identity.
\item On $\mathcal A\cap\mathcal B^\perp$ and $\mathcal A^\perp\cap\mathcal B$,
$W$ acts as $-I$.
\item If $v$ is an eigenvector of $W$ in $\mathcal A+\mathcal B$ with eigenvalue
1, then $v\in\mathcal A\cap\mathcal B$.
\end{enumerate}
\end{thm}
\begin{proof}
For point 1, note that we can interpret $D(A, B)$ as an orthogonal projector
from $\mathcal B$ to $\mathcal A$ in the sense that for every $v\in\mathbb C^m$,
$\pi _{\mathcal A}Bv = AA^*Bv = AD(A, B)v$. Conversely, $D(A,B)^*$ can be viewed
as a projector from $\mathcal A$ to $\mathcal B$ since $\pi_{\mathcal B}Av =
BB^*Av = BD(A,B)^*v$. Since $D(A,B)$ is Hermitian, its eigenvalues are real.
If $v$ is an eigenvector of $D(A, B)$ with eigenvalue $\lambda$, we have
\begin{align*}
\pi _{\mathcal A}Bv &= AD(A,B)v = \lambda Av\\
\pi _{\mathcal B}Av &= BD(A,B)^*v = BD(A,B)v = \lambda Bv
\end{align*}
Combining these two equations, we get $\pi _{\mathcal B}\pi _{\mathcal A}Bv
= \lambda ^2 Bv$. Since projectors cannot increase the length of a vector,
it follows that $|\lambda |\leq 1$.

For point 2, let $v$ be a unit-length eigenvector of $D(A,B)$ with eigenvalue
$\cos \theta$. The angle $\theta$ has a geometric meaning: It is
the angle between $Av$ and $Bv$ since $\inp{Av}{Bv} = v^*A^*Bv = v^*D(A,B)v =
\cos\theta$.
Since $AA^*Bv = AD(A,B)v = \cos \theta Av$ and
$BB^*Av = BD(A,B)^*v = \cos \theta Bv$, the vector space $V$ spanned by $Av$ and
$Bv$ is invariant under $W$. Moreover, the action of $W$ on $V$ is a reflection
through $Av$ followed by a reflection through $Bv$ which corresponds to a
rotation with angle $2\theta$. Therefore, the eigenvectors in this subspace
have eigenvalues $e^{\pm2i\theta}$.

Since the eigenvectors of $D(A,B)$ form a basis of $\mathbb C^n$, the set
$$\{ Av, Bv \mid v\text{ is an eigenvector of }D(A,B)\}$$
is a generating set for $\mathcal A + \mathcal B$. Thus, $D(A,B)$ can
have no other eigenvalues with eigenvectors in the subspace
$\mathcal A + \mathcal B$.

Points 3 and 4 follow easily from the fact that $\text{ref} _{\mathcal A}$
acts as the identity on $\mathcal A$ and as $-I$ on $\mathcal A^\bot$
while $\text{ref}_{\mathcal B}$ acts as $I$ on $\mathcal B$ and as
$-I$ on $\mathcal B^\bot$. To see that point 5 holds, note that
a vector in $\mathcal A+\mathcal B$ can only be mapped to itself under $W$
if it is in $\mathcal A\cap\mathcal B$.
\end{proof}
\begin{rem}
Szegedy formulated his theorem for matrices $A$ and $B$ of arbitrary size,
using the \em singular values \em of $D(A,B)$. In the case that $D(A,B)$ is
square and Hermitian, the singular values coincide with the absolute values
of the eigenvalues.
\end{rem}

The subspace $\mathcal A \cap \mathcal B$ is spanned by $\ket{\mathcal U^x}$:
The projector from $\mathcal H_L\otimes \mathcal H_R$ on $\mathcal A\cap
\mathcal B$ is given by
\begin{align*}
\left( \sum _{u\in\mathcal V} \ket{p_u}\bra{p_u}\right)\cdot
\left( \sum _{v\in\mathcal V} \ket{q_v}\bra{q_v}\right) &=
\sum_{u\in\mathcal V}\sum _{v\in\mathcal V} \inp{p_u}{q_v}\ket{p_u}\bra{q_v}\\
&= \frac 1 d\sum _{u\in\mathcal V}\sum_{v\in\mathcal V_u}\ket{p_u}\bra{q_v}\\
\end{align*}
and thus, if $x = \sum _u\lambda _u\ket{p_u}$ is in
$\mathcal A \cap \mathcal B$, it must hold that
$$x = \frac 1 d\sum_{u'\in\mathcal V}\sum _{v\in\mathcal V_u}\sum _{u\in\mathcal V}
\lambda _{u}\ket{p_{u'}}\inp{q_v}{p_u}$$
and hence, for all $w\in\mathcal V$,
$$\lambda _w = \frac 1{d} \sum _{v\in\mathcal V}\sum_{u\in\mathcal V_{v}}
\lambda _{u}\inp{q_v}{p_{u}}\text .$$
That is, all $\lambda _w$ must be identical. Hence, every vector in $\mathcal A
\cap\mathcal B$ can be written as
$$\sum _{u\in\mathcal V}\lambda \ket{p_u} = \lambda\sum _{u\in\mathcal V}
\ket u_L\sum _{v\in\mathcal V_u}\frac 1{\sqrt d}\ket v_R$$
and thus it must be a scalar multiple of $\ket{\mathcal U^x}$.

This allows us to distinguish $\ket{\mathcal U^x}$ from other states in
$\mathcal A + \mathcal B$ as follows. From Theorem \ref{spec_lem}, we know that
$\ket{\mathcal U^x}$ is an eigenvector of $W(G)$ with eigenvalue 1. Let
$\delta$ be the spectral gap of $\frac 1 dG$. Then, every other eigenvector
of $W(G)$ in $\mathcal A + \mathcal B$ must have eigenvalue
$e^{\pm i2\theta}$ for some $\theta$ with
$$\delta \leq 1-|\cos (\theta )| \leq \theta ^2/2\Leftrightarrow
|\theta | \geq \sqrt{2\delta}\text .$$
To distinguish $\ket{\mathcal U^x}$ from other states in $\mathcal A +
\mathcal B$, we use a quantum algorithm called \em phase estimation\em :
Fix some unitary $U$. For any quantum state $\ket{\phi}$ that is an eigenvector
of $U$ with eigenvalue $e^{2\pi i\alpha}$, where $0\leq \alpha < 1$, we can
obtain an estimate of $\alpha$ with good probability. Let $\alpha '$ be
$\alpha$ rounded to $n$ binary digits. Phase estimation maps
$\ket{0^n}\ket{\phi}\mapsto \ket{2^n\alpha'}\ket{\phi}$ with high
probability.
This requires applying Hadamards on the first $n$ qubits, the inverse of the
\em quantum Fourier transform \em and $n$ times the transform $U$. The
Fourier transform on $n$ qubits is the unitary mapping
$$\ket{j}\mapsto \frac{1}{\sqrt{2^n}}\sum _{k=0}^{2^n-1}e^{2\pi i\cdot jk/2^n}\ket{k}\text .$$

To give some intuition about phase estimation, we describe it for the case
where $\alpha = \alpha '$ has \em exactly \em $n$ binary digits in which it gives
the correct result with certainty. We start with state
$\ket{0^n}\ket{\phi}$. Applying the Hadamard gates, we obtain
$\sum _{j=0}^{2^n-1}\ket{j}\ket{\phi}$. Then, we apply the transform that maps
$\ket{j}\ket{\phi}\mapsto \ket j U^j\ket{\phi} = e^{2\pi ij\alpha }\ket j\ket{\phi}$.
Applying this on our quantum state gives us $\sum _{j=0}^{2^n-1}e^{2\pi ij\alpha}\ket{j}\ket{\phi}$.
But this is also the state that results from applying the Fourier transform
on the first $n$ qubits of $\ket{2^n\alpha}\ket{\phi}$. Therefore, computing
the inverse of the Fourier transformation will give us the state we want.
The cost of this algorithm is the cost of generating the uniform superposition
($n$ gates), computing $n$ times the unitary $U$ and then computing the inverse
Fourier transformation ($O(n^2)$ gates for computing it exactly, $O(n\log n)$
for a good approximation). See \cite[Chapter 5]{mikeike} for a more complete
description of the phase estimation algorithm.

Using phase estimation with precision $O(1/\sqrt{\delta})$, we can
distinguish with good probability the case where a given state $\ket \phi$ in
$\mathcal A + \mathcal B$ is $\ket{\mathcal U^x}$
from the case where $\ket\phi$ is an eigenvector of $W$ with eigenvalue
$\neq 1$. We can then implement the reflection by performing the phase
estimation, putting a minus in the phase if the eigenvalue is not 1 and reversing
the phase estimation again. Assuming that the cost of the computation of
$W(G)$ dominates the other costs in the phase estimation, it costs
$O(1/\sqrt{\delta})U$ to reflect through $\ket{\mathcal U^x}$ and
$C$ to reflect through $\ket{\mathcal B}$. As in Grover's algorithm, we will
have to perform these two reflections $O(\sqrt{1/\epsilon})$ times each to find
a vertex $u\in M_x$ or determine that $M_x = \emptyset$. Thus,
we have the following theorems:

\begin{thm}
\label{walk2}
Let $\mathcal G$ be a graph and let $\delta$ be the spectral gap. Let $\epsilon > 0$ be such that
for all $M_x\neq \emptyset$, $\epsilon\leq |M_x|/|\mathcal V|$.
There is a bounded-error quantum algorithm that on input $x$, finds an element
of $M_x$ or determines that $M_x$ is empty with cost
$$O\left( S+\frac 1{\sqrt\epsilon}\left( C + \frac 1 {\sqrt\delta} U\right)\right)\text .$$
\end{thm}
\begin{thm}[Jeffery, Kothari, Magniez, \cite{quant_walk}]
\label{walk1}
For $\mathcal G$, $\delta$ and $\epsilon$ as before, there is a bounded-error
quantum algorithm that implements the transform
$$\ket{\mathcal U ^x}\mapsto\begin{cases}-\ket{\mathcal U ^x}&\text{ if }M_x\neq \emptyset\\
\ket{\mathcal U^x}&\text{ otherwise}\end{cases}$$
with cost
$$O\left( \frac 1{\sqrt\epsilon}\left( C + \frac 1 {\sqrt\delta} U\right)\right)\text .$$
\end{thm}

As a first example, let us see an algorithm for the \em element distinctness \em
problem given by Ambainis in \cite{amb}. In this problem, we are given as
input integers $x_1,\dots , x_n$ and
we have to determine if there are distinct indices $i$ and $j$ such that
$x_i = x_j$ and if yes, we have to output such indices. Assuming that we can
completely read any integer in the input with $O(1)$ cell-probes, there exists
a quantum walk algorithm that solves the element distinctness problem
with $O(n^{2/3})$ cell-probes. Classically, we have to make $\Omega (n)$
cell-probes to solve the problem. The complexity for quantum cell-probe
algorithms is lower-bounded by $\Omega (n^{2/3})$, which is proved in
\cite{aaronson_shi}, so Ambainis's algorithm
is asymptotically optimal.
The quantum walk takes place on a
\em Johnson graph \em which we define below.
\begin{defin}[Johnson graph]
For positive integers $n$ and $r$, the vertices of the Johnson graph $J(n,r)$
are the subsets of $[n]$ with exactly $r$ elements. Two vertices are neighbours
if and only if their symmetric difference has exactly two elements.
That is, we get from a vertex $u$ to a neighbour $v$ by removing
one element from $u$ and adding a different one.
\end{defin}
The spectral gap of $J(n,r)$ is
$\delta = n/(r(n-r)) = \Omega (1/r)$. We will treat $r$ as a parameter
for now and determine a suitable value for it later.
For input $x = x_1,\dots , x_n$, we let $M_x = \{ u\in J(n,r)\mid \exists
i,j\in u: x_i = x_j\}$, the set of vertices of $J(n,r)$ that contain colliding
indices. When we have found an element of $M_x$, we can find distinct $i$ and
$j$ such that $x_i = x_j$ by making $O(r)$ cell-probes to our input. We find
such a set $M_x$ via quantum walk. With input
$x$ and vertex $u$ of $J(n,r)$, we associate a
representation $\ket{D^x(u)}$ of the set $\{ (i, x_i)\mid i\in u\}$.
In the \textbf{Setup}-phase, we want to construct the state
$$\sqrt{\frac{1}{r(n-r)}\cdot \binom{n}{r}^{-1}}
\sum _{u\in J(n,r)}\sum _{v\in J(n,r)_u}\ket{u}_L\ket{v}_R\ket{D^x(u)}$$
and it costs $r$ quantum cell-probes to create $\ket{D^x(u)}$. The
\textbf{Update}-phase requires $O(1)$ quantum cell-probes since neighbouring
vertices $u$ and $v$ only differ in two elements. The \textbf{Checking}-step
requires no cell-probes since all the information we need to decide if
$u\in M_x$ is contained in $\ket{D^x(u)}$. Let us now determine $\epsilon$.
Suppose that $M_x\neq \emptyset$ and let $i$ and $j$ be distinct indices
such that $x_i = x_j$. If we select $u\in J(n,r)$ at random, there is a
probability of
$$\epsilon = \frac{r}{n}\cdot \frac{r-1}{n-1}$$
that $u$ contains those two indices. Thus, our quantum walk algorithm makes
$$O\left( r + \sqrt{\frac{n(n-1)}{r(r-1)}}\sqrt\frac{r(n-r)}{n}\right)
= O\left( r + \frac n{\sqrt r}\right)$$
cell-probes. If we set $r = n^{2/3}$, the complexity becomes
$O(n^{2/3} + n^{1-1/3}) = O(n^{2/3})$. When we have found an element of
$M_x$, we need to make $O(r)$ cell-probes to actually find colliding indices.
Thus, we can solve the element distinctness problem with $O(n^{2/3})$
cell-probes.

We now describe a framework for \em nested \em quantum walks given in
\cite{quant_walk}.
Such quantum walks consist of an outer walk on a graph $\mathcal G$
where the checking step is implemented by a quantum walk on another
$d'$-regular graph $\mathcal G' = \left( \mathcal V',\mathcal E'\right)$.
With each $x$ and $u$, we associate
a set $M_x^u\subseteq \mathcal V'$ such that $M_x^u\neq\emptyset$
if and only if $u\in M_x$.
Let
$$\ket{\mathcal U ^x_u} = \sum _{u'\in\mathcal V'}\sum _{v'\in\mathcal V'_{u'}}\frac 1{\sqrt{d'|\mathcal V'|}}\ket{u'}_{L'}\ket{v'}_{R'}\ket{D^x_u(u')}$$
where $D^x_u$ is a data structure for the walk on $\mathcal G'$.
This means that $\ket{\mathcal U ^x_u}$ is the initial state of the quantum walk
on $\mathcal G'$.
The \textbf{Setup} consists of
preparing a quantum state
$$\ket{\mathcal U ^x} = \sum_{u\in\mathcal V}\sum_{v\in\mathcal V_u}
\frac 1{\sqrt{d|\mathcal V|}}\ket{u}_L\ket{v}_R\ket{\mathcal U ^x_u}\text .$$
That
is, the data structure for the outer walk associated with $u$ is
the initial state of the walk on $\mathcal G'$. The
\textbf{Update} operation is as before.
The \textbf{Checking} operation maps
$$\ket u_L \ket v_R\ket{\pi ^x_u}\mapsto
\begin{cases}
-\ket u_L\ket v_R\ket{\mathcal U ^x_u}&\text{ if }u\in M_x\\
\ket u_L\ket v_R\ket{\mathcal U ^x_u}&\text{ otherwise}
\end{cases}$$
and since $u\in M_x$ if and only if $M_x^u\neq\emptyset$, we can
implement this operation with bounded error by applying Theorem \ref{walk1} to
the inner walk.
\begin{thm}[Nested Quantum Walks, \cite{quant_walk}]
Let $\delta '$ be the spectral gap of $\mathcal G'$ and $\epsilon ' >0$ be a
lower bound
for $|M_x^u|/|\mathcal V'|$ with $M_x^u\neq\emptyset$.
Suppose that an update of the inner walk has
update cost at most $U'$ and the checking cost for $u'\in M_x^u$ is at
most $C'$. Then, for $\tilde{O}(f(n)) = O(f(n)\cdot \text{polylog}(f(n)))$, we
have a bounded-error algorithm that finds an element in $M_x$ or determines
that $M_x = \emptyset$ with cost
$$\tilde{O}\left( S + \frac{1}{\sqrt\epsilon}
\left( \frac 1{\sqrt\delta} U + \frac{1}{\sqrt{\epsilon '}}\left( C' + \frac{1}{\sqrt{\delta '}} U'\right)\right)\right)$$
\end{thm}
\begin{proof}[Proof (sketch)]
As discussed above, this result follows by implementing the \textbf{Checking}
operation via Theorem \ref{walk1}. The polylogarithmic factor hidden in the
$\tilde O$-notation comes from the fact that we need to amplify the success
probability for the checking step of the inner and outer walks.
\end{proof}

As an application of this framework, we describe a quantum walk algorithm for
\em triangle finding \em in graphs from \cite{quant_walk}.
A \em triangle \em in a graph is a set of three vertices where every two vertices in
the set are neighbours. We are given oracle access to
(the adjacency matrix of) a graph $G$.
If the input graph $G$ has $n$ vertices (which translates to roughly
$n^2$ input bits), the quantum bit-probe complexity
of the algorithm is $\tilde O(n^{9/7})$. First, we introduce some notation.
If $G = (V,E)$ is a graph and
$R\subseteq V$, we let $G_R$ be the restriction of $G$ to $R$, i.e.,
$(R,E|_{R\times R})$. If $L$ is a possible set of edges on vertices $V$, we let
$G(L) = (V,E\cap L)$.
\begin{thm}
\label{walk3}
There is a quantum bit-probe algorithm with bounded error
that decides whether a graph $G$ contains a triangle using
$\tilde O(n^{9/7})$ bit-probes.
\end{thm}
\begin{proof}
The outer walk is on the Johnson graph $\mathcal G = J(n,r_1)$ and the inner
walk on $\mathcal G' = J(n,r_2)$ for $r_1$ and $r_2$ such that
$r_1\leq r_2\leq r_1^2$.  We will fix values for $r_1$ and $r_2$ later.
Let $G$ be the input graph. We identify the vertices of $J(n,r_1)$ and
$J(n,r_2)$ with $r_1$- or $r_2$-size sets of vertices of $G$.
We define $M_G$ such that for every vertex $R_1$ of $\mathcal G$, we have
$R_1\in M_G$ if and only if $R_1$ contains a vertex of $G$ that is part of a
triangle. In the inner walk, we let $R_2\in M_G^{R_1}$ if and only
if $R_1$ contains a $G$-vertex $v_1$ and $R_2$ contains a $G$-vertex
$v_2\neq v_1$ such that $v_1$ and $v_2$ are part of the same triangle in $G$.

The data structure of the inner walk is given by $D_G^{R_1}(R_2)$
being the subgraph of $G$ that contains exactly those edges that
have one endpoint in $R_1$ and one in $R_2$. Let $R_2$
and $R_2'$ be neighbours. Since there is exactly one vertex $r$ in
$R_2'$ that is not in $R_2$, updating the data structure for operations
$P_A$ and $P_B$ requires querying for each $r'\in R_1$
whether $r$ and $r'$ are neighbours. Thus, the \textbf{Update}
of the inner walk requires $U' = O(r_1)$ bit-probes.

We can compute the \textbf{Checking} step of the inner walk with sufficiently
low error probability using $C' = \tilde O(\sqrt{n}(r_1r_2)^{1/3})$ bit-probes.
This is done as follows: For any vertex $v$ of $G$, we can use a subroutine from
\cite[Appendix A]{quant_walk}, which we will describe later on, to look for
vertices $v_1\in R_1$ and $v_2\in R_2$ such that $v$, $v_1$ and $v_2$ form a
triangle. This subroutine requires $O((r_1r_2)^{1/3})$ bit-probes. We then use a
variant of Grover's algorithm to find out whether there exists a vertex $v$ of
$G$ that forms a triangle with one vertex from $R_1$ and one vertex from $R_2$.
This is done by replacing the oracle query of Grover's algorithm with the
subroutine we just mentioned. We need to make $O(\sqrt{n}(r_1r_2)^{1/3})$
bit-probes to do this.

Let us now describe the subroutine; it is similar to the algorithm for
element distinctness. Given two graphs $G_1 = (V_1, E_1)$ and $G_2 = (V_2,E_2)$,
let $G_1\times G_2$ be the graph with vertices $V_1\times V_2$
where two vertices $(u_1, u_2)$ and $(u'_1, u'_2)$ are neighbours if and only if
$u_1$ and $u'_1$ are neighbours in $G_1$ and $u_2$ and $u'_2$ are neighbours in
$G_2$. The subroutine is again a quantum walk, this time on the graph
$\mathcal G'' = J(r_1, k)\times J(r_2, k)$ for $k$ to be determined later. We
assign every vertex in $R_1$ a number in $[r_1]$ and every vertex in $R_2$ a
number in $[r_2]$. We then can view every vertex $(U_1, U_2)$ of $\mathcal G''$
as a pair that consists of a set of $k$ vertices in $R_1$ and a set of $k$
vertices in $R_2$. We say that a vertex $(U_1, U_2)$ of $\mathcal G''$ is \em
marked \em if and only if there is a vertex $v_1\in U_1$ and a vertex
$v_2\in U_2$ such that $v$, $v_1$ and $v_2$ form a triangle in $G$. Our goal
is to find a marked vertex of $\mathcal G''$.

During the walk we maintain a data
structure that records for every vertex in $U_1\cup U_2$ whether it is a
neighbour of $v$. Using this data structure and the information from
$D_G^{R_1}(R_2)$, we can check whether a vertex of $\mathcal G''$ is marked
without any bit-probes. For the \textbf{Setup}, we have to make $O(k)$
bit-probes and for the \textbf{Update} we have to make $O(1)$ bit-probes.
The spectral gap of $J(r_1, k)\times J(r_2,k)$ is $\delta '' = \Omega (1/k)$ and
a lower bound on the fraction of marked vertices in $\mathcal G''$ -- if there
are any -- is $\epsilon ''= k^2/(r_1r_2)$ which can be seen as
follows: Suppose there is a vertex $v_1$ in $R_1$ and a vertex $v_2$ in $R_2$
such that $v,v_1,v_2$ form a triangle. If we select a vertex of $J(r_1,k)$ at
random, the probability that it contains $v_1$ is $k/r_1$ and if we
select a random vertex of $J(r_2,k)$, the probability that it contains $v_2$ is
$k/r_2$. Thus, a random vertex of $\mathcal G''$ is marked with
probability at least
$$\epsilon '' = \frac{k}{r_1r_2}\text .$$
The complexity of the subroutine is
$$O\left( k + \frac 1{\sqrt{\epsilon''}}\cdot \frac 1{\sqrt{\delta ''}}\right)
= O\left( k + \frac{\sqrt{r_1r_2}}{k}\cdot\sqrt{k}\right) = O\left( k +
\frac{\sqrt{r_1r_2}}{\sqrt{k}}\right)$$
bit-probes. Choosing $k = (r_1r_2)^{1/3}$ results in a complexity of
$O((r_1r_2)^{1/3})$ bit-probes, as claimed.

Let us now return to the inner walk on $\mathcal G'$.
The spectral gap of $J(n,r_2)$ is $\delta ' = \Omega (1/r_2)$
and we have $\epsilon ' = r_2/n$. To see this, suppose
that there is a vertex $v_1$ in $R_1$ that is part of a triangle. Then, a
vertex $R_2$ of $\mathcal G'$ is in $M_G^{R_1}$ if and only if it contains a
vertex $v_2\neq v_1$ of $G$ such that $v_1$ and $v_2$ are part of the same
triangle in $G$. Fix such a vertex $v_2$. The probability that a random vertex
in $\mathcal G'$ contains $v_2$ is
$\epsilon ' = r_2/n$.

Thus, the \textbf{Checking} step for the outer walk has cost
$$C = \tilde O\left( \frac 1{\sqrt{\epsilon '}}\left( \frac 1{\sqrt{\delta '}}U' + C'\right)\right)
= \tilde O\left( \sqrt{n}r_1 + \frac{nr_1^{1/3}}{r_2^{1/6}}\right)$$
The \textbf{Setup} cost of the outer walk is $S = r_1r_2$ since the number of
possible edges in $D_G^{R_1}(R_2)$ is $S=r_1r_2$. The cost for the
\textbf{Update} operation is $U = O(r_2)$ since updating the data structure from
$R_1$ to a neighbour $R_1'$ requires querying for the $r'\in R_1'\setminus R_1$
and every $r\in R_2$ whether $r'$ and $r$ are neighbours.

We have $\delta = \Omega (1/r_1)$. Also, similar to the inner walk, we have
$\epsilon = r_1/n$. Thus, Theorem \ref{walk3} gives us an algorithm
that makes
$$
\tilde O\left( r_1r_2 + \frac{\sqrt{n}}{\sqrt{r_1}}\left( \sqrt{r_1}r_2 + \sqrt{n}r_1 + \frac{nr_1^{1/3}}{r_2^{1/6}}\right)\right)
= \tilde O\left( r_1r_2 + \sqrt{n}r_2 + n\sqrt{r_1} +
        \frac{n^{3/2}}{(r_1r_2)^{1/6}}\right)
$$
bit-probes. Setting $r_1 = n^{4/7}$ and $r_2 = n^{5/7}$ gives an
algorithm that makes $\tilde O(n^{9/7})$ bit-probes.
\end{proof}
This proof shows how we can reduce costs by putting the more expensive
operations into the outer walk: Since $r_2 > r_1$, the \textbf{Update}-step
of the outer walk is more expensive than the \textbf{Update}-step of the inner
walk. A classical algorithm for triangle finding must make $\Omega (n^2)$
queries: Consider a complete bipartite graph where the sets of vertices
$V_1$ and $V_2$ both have $n/2$ elements. This graph does not contain any
triangles. But if we add \em any \em edge between two vertices in $V_1$ or
two vertices in $V_2$, we have a triangle. Thus, it is necessary to check
$\Omega \left( 2\cdot \binom{n}{2} \right) = \Omega\left( n^2 \right)$
edges to distinguish the complete bipartite graph from a graph with a triangle.
It is not known whether $\tilde O(n^{9/7})$ is optimal in the quantum setting,
but no better algorithm has been found.
\section{Summary}
What can quantum computing do for the data structure problems that we investigated
in this survey? If one looks at it superficially, one might say that it does
not do much.
For the set membership problem, the Perfect Hash method offers a solution
that achieves the information-theoretic minimum of memory up to a constant
factor and has a time complexity of $\log m$ bit-probes. This time complexity
cannot be improved in the quantum bit-probe model as long as we require
an exact query algorithm or one with one-sided error. If we allow bounded
error, there are classical data structures which need only one bit-probe.
For the predecessor search problem, we showed that the time complexity of the
data structure by Beame and Fich cannot be improved without raising the space
complexity, even in the address-only quantum cell-probe model. While this
does \em not \em exclude the possibility that a quantum query algorithm exists
which uses fewer cell-probes, such an algorithm must be relatively complicated.
The qubits which receive the cell-probe results cannot be in some simple
state such as $\ket 0$. They have to be entangled with the work-space qubits.

While knowing that quantum computing can \em not \em help us in some given area
has some value on its own, our survey also demonstrates that the theory of
quantum computing is a valuable mathematical tool. Even though one might say
that arguments from quantum computing are really linear algebra arguments with
odd notation, quantum computing represents a unique style of mathematical
arguments that provides useful results.
The lower bounds presented in this survey carry over to classical probabilistic
computing and they are stronger and easier to prove than
previous lower bounds for these problems.
These are examples how the theory of quantum computing can be relevant to classical computing
(for more examples, see \cite{qu_method}).

Finally, we saw how fully quantum data structures allow us to beat the
information-theoretic lower bound for the set membership problem. However,
in that setting, we must take into account that accessing the data structure
might disturb the quantum state which limits the number of times that it can
be used. Also, fully quantum data structures can improve quantum walk
algorithms.

\subsubsection*{Acknowledgements}
I thank Ronald de Wolf for supervising the writing of this survey and for his
helpful comments and advice.

\bibliographystyle{plain}
\bibliography{refs}
\end{document}